\documentclass{LMCS}

\def\dOi{9(3:26)2013}
\lmcsheading%
{\dOi}
{1--68}
{}
{}
{Dec.~14, 2012}
{Sep.~25, 2013}
{}

\ACMCCS{[{\bf Theory of computation}]: Models of
  computation---Concurrency---Process calculi} 

\usepackage{amssymb}
\usepackage{amsmath}
\usepackage{mathtools}
\usepackage{amsthm}
\usepackage{latexsym}
\usepackage{stmaryrd}
\usepackage{pi}
\usepackage{amstext}
\usepackage{xspace}
\usepackage{bbm}
\usepackage{nicefrac}

\usepackage{rules}
\usepackage{maths}
\usepackage{dist}
\usepackage{networks}

\usepackage{multirow}

\usetikzlibrary{arrows,patterns}

\newcommand{\calL}{\ensuremath{{\mathcal{L}}}}

\newcommand{\leaveout}[1]{ }

\definecolor{greenish}{rgb}{.24,.5,.26} 

\begin{document}
\title[Modelling Probabilistic Wireless Networks]{Modelling Probabilistic Wireless Networks}

\author[A.~Cerone]{Andrea Cerone\rsuper a}
\address{{\lsuper a}IMDEA Software Institute, Madrid, Spain}
\email{andrea.cerone@imdea.org}
\author[M.~Hennessy]{Matthew Hennessy\rsuper b}
\address{{\lsuper b}School of Computer Science and Statistics\\ Trinity College Dublin\\Ireland}
\email{matthew.hennessy@scss.tcd.ie }   
\thanks{{\lsuper{a,b}}The financial support of SFI is gratefully acknowledged.}

\begin{abstract}
  We propose a process calculus to model high level wireless systems, where the 
  topology of a network is described by a digraph. The calculus enjoys  
  features which are proper of wireless networks, namely broadcast 
  communication and probabilistic behaviour.

  We first focus on the problem of composing wireless networks, then we 
  present a compositional theory based on a probabilistic generalisation of the well 
  known may-testing and must-testing preorders. Also, we define an extensional semantics 
  for our calculus, which will be used to define both simulation and deadlock simulation 
  preorders for wireless networks. We 
  prove that our simulation preorder is sound with respect to 
  the may-testing preorder; similarly, the deadlock simulation preorder is 
  sound with respect to the must-testing preorder, for a large class of networks.\\
  We also provide a counterexample showing that completeness of 
  the simulation preorder, with respect to the may testing one, does not hold. 
  We conclude the paper with an application of our theory to probabilistic 
  routing protocols.
\end{abstract}

\maketitle


\section{Introduction}
\label{sec:intro}
Wireless networks have spread worldwide in the last decades; nowadays they are used in many areas, from domestic appliances to mobile phone networks, to the newer sensor infrastructures.
One of the main problems of wireless networks is 
that of defining and implementing protocols for providing to users the services 
for which the network has been designed; also, because of their distributed nature, 
a more challenging problem is that of ensuring in a rigorous, mathematical way, 
the correct behaviour of a network with respect to some specification.

This problem becomes even more difficult to tackle if we consider that 
often wireless networks run protocols whose behaviour is probabilistic. 
Such protocols are indeed very useful for improving the performance of 
wireless networks, examples being the use of probabilistic \emph{routing 
protocols} \cite{sample} or probabilistic protocols for \emph{collision avoidance} 
at the \emph{MAC-sublayer} of the \emph{TCP/IP} reference model \cite{macsurvey}. 
Further, problems for which there is no solution in   
a deterministic setting can be solved (in unbounded time) by introducing 
probabilistic behaviour in wireless networks 
\cite{BraTou85}.
 
Many different formal frameworks have been developed in the literature 
for defining and \linebreak 
analysing the behaviour of wireless networks 
\cite{nanz,Godskesen07,GwFM10,LaneseS10,omegacalc,merro,gallina2011,songphd,wang}; 
these differ in many details, 
the most important being the level of abstraction used to 
represent a wireless network, the computational power of stations 
of wireless networks and the mathematical structure 
used to represent the topology of wireless networks. 
However, each of these calculi have the following features in common: 
wireless networks are represented as a collection of stations (also called 
nodes, or locations) running code, and local broadcast is used as 
the only way of communication. Roughly speaking in local broadcast 
communication, whenever a node broadcasts 
a message only the nodes in its range of transmission are affected.

In this paper we propose another process calculus for modelling probabilistic 
wireless networks; the main concepts underlying our calculus can be 
summarised as follows:

\begin{enumerate}[label=(\roman*)]
\item The topology of a wireless network is static, that is mobility is not 
considered in our model. 
This restriction has been done to allow a more clear presentation of the 
topics treated in this paper; however, 
we could have used the approach described in \cite{francalanza2008} to 
introduce node mobility.
The network topology is 
described by a digraph $\Gamma$; 
intuitively vertices in this graph represent network locations, while 
an edge from a node to another is used for expressing that the latter is 
in the range of transmission of the former. 

\item A probabilistic process calculus is defined for assigning code to locations. 
The basic constructs allowed in our calculus are messages' broadcast and reception, 
internal actions, matching and process definitions; further, we allow a special clause 
$\omega$ whose role will be presented shortly.  
The mapping that assigns code to locations is partial, meaning that locations can have no code 
assigned. At least informally, such nodes, which will be called \emph{external}, 
can be seen as terminals at which users can place code 
to test the behaviour of the network.

\item Communication between nodes is reliable; a message broadcast by a 
node along a channel $c$ will be received by all the nodes in the sender's 
range of transmission, provided that they are waiting to receive a message 
along such a channel. In other words, our calculus is designed for describing 
wireless networks at the \emph{network layer} of the \emph{TCP/IP reference model}; 
reliable communication is not ensured at lower levels, where 
issues such as the possibility of collisions \cite{macsurvey} and synchronisation between 
nodes \cite{time} arise.
\end{enumerate}

\noindent One of the main goals of the paper is that of defining a compositional behavioural 
theory of wireless networks; given two wireless networks $\calM$ and 
$\calN$, we want to establish whether they can be distinguished by 
an external user. 
To accomplish this task, we need to address several different topics. 
First, it is necessary to define how two wireless networks can be composed 
together. 
This topic has already 
been addressed, for different process calculi, in \cite{merro,gallina2011,CHM12,bugliesi2012,songphd}.
Here we define an asymmetric operator 
$\testP$ which can be used to extend one network with another. 
Despite being asymmetric, we show that the choice of the operator $\testP$ is 
driven by  some natural requirements we require in general 
from a composition operator between networks. 
We remark that our theory 
of composition is restricted to a particular class of networks, 
which we call \emph{well-formed}.

Once we have chosen a suitable composition operator $\testP$, 
we can define a compositional theory for wireless networks. 
In this paper we have chosen to focus on a probabilistic generalisation 
of the well-known \emph{De Nicola and Hennessy's testing preorders}, 
whose theory has been defined in \cite{DGHM09full} for a probabilistic 
version of \emph{CSP}.

Informally speaking, we can test a wireless network $\calM$ via 
another wireless network $\calT$ which can be composed 
 with the former (with respect to the operator $\testP$); 
that is, the network $(\calM \;\testP\; \calT)$ is 
defined. Intuitively, the network 
$\calM \;\testP\; \calT$ can be considered as an experiment 
in which the role of the testing component $\calT$ is that of 
determining whether $\calM$ satisfies some property for which 
the test has been designed for. The success of an experiment is 
denoted by the special construct of our calculus $\omega$ mentioned above. 

Having this in mind, each computation of the network 
$(\calM\;\testP\;\calT)$ induces a success outcome, denoting the probability 
of reaching a configuration in which the special clause $\omega$ is enabled 
in such a computation. This induces a set of success outcomes for the network 
$(\calM\;\testP\;\calT)$ by quantifying over all the possible computations for such a network. 

Knowing how to associate a set of success outcomes to a network, we can compare two 
networks $\calM, \calN$ by quantifying over all possible tests $\calT$, and comparing 
the sets of success outcomes of the experiments $(\calM \;\testP\; \calT)$ and 
$(\calN \;\testP\;\calT)$, provided that they are both defined. This leads to 
the definition of two testing preorders, the \emph{may-testing} preorder $\Mayleq$ 
and \emph{must-testing} preorder $\Mustleq$, according to the way in which the sets 
of success outcomes for the two experiments above are compared. 

It is important to note that determining directly whether the statement  
$\calM \Mayleq \calN$ ($\calM \Mustleq \calN$) 
is true is not easy, due to the quantification over all tests. 
Therefore there is the need to define a proof methodology for establishing 
if two networks can be related via the $\Mayleq$ ($\Mustleq$) preorder. 
This is the main topic of our paper. 
To this end, we define an extensional 
semantics for our calculus of wireless networks; the actions in 
this semantics correspond to activities that can be observed 
by the external nodes. The main idea here is that of defining 
sound coinductive proof methods for the testing preorders, 
based on the extensional behaviour of networks. 

Since our calculus is equipped with local broadcast communication, 
we need to take care of some issues in the development of such proof 
methods; roughly speaking, the broadcast of a message to a 
set of external nodes can be simulated by a multicast of the same message 
which can be detected by the same set of external nodes. This leads 
to a non-standard definition of weak extensional actions, which will 
be used to define two coinductive relations between networks. The first one 
is the well-known \emph{simulation preorder} \cite{DGHM09full}; the second 
one is a novel preorder, called the \emph{deadlock simulation preorder}, 
which is obtained from the previous one by adding sensitivity to deadlock 
configurations. 
The main results of the papers are that, for a large class of networks, 
the simulation preorder is sound with respect to the may-testing preorder, 
while the inverse of the deadlock simulation preorder is sound with 
respect to the must-testing preorder. However, we provide a 
counterexample that shows that such proof methods fail 
to be complete.

The rest of the paper is organised as follows: 
in Section \ref{sec:background} we recall the mathematical 
tools needed for the development of our theory. 

In Section \ref{sec:lang} we define the syntax and intensional 
semantics of our calculus of wireless networks, and we prove 
some basic properties of our calculus. 

 In Section \ref{sec:compositional} we give the formal definition
  of the behavioural preorders between networks. This depends on how tests
  are applied to networks or more generally how networks are composed
  to form larger networks.  So we first define our composition
  operator $\testP$, which is asymmetric, in
  Section~\ref{sec:comp.net} and then use it to develop the
  behavioural preorders $\Mayleq$ and $\Mustleq$ between networks. In
  Section~\ref{sec:justify} we return to our choice of composition
  operator $\testP$, justifying it as the largest one which satisfies 
  three natural requirements. In addition, somewhat surprisingly, we show
  that any symmetric composition operator satisfying the natural requirements
  generates a degenerate behavioural theory. 

In Section \ref{sec:ext.sem} we define the extensional semantics of 
our calculus of wireless networks; here we also give the non-standard 
definition of weak extensional actions and we prove composition and 
decomposition results for them, with respect to the composition operator 
$\testP$.

In Section \ref{sec:soundness} we define the notions of simulation and 
deadlock simulation preorders and we prove the main results of the paper, 
namely that the simulation preorder is sound with respect to the may-testing 
preorder, and the inverse of the deadlock simulation preorder is sound 
with respect to the must-testing preorder. Much of the technical 
development underlying these soundness results is relegated to the separate 
Section \ref{sec:technical}; this may be safely skipped by the uninterested reader.

In Section \ref{sec:completeness} we show that our proof methods fail to be 
complete; we also show the impossibility of defining 
a coinductive relation based on our notion of extensional actions, which characterises 
the may-testing relation. 

In Section \ref{sec:prob.routing} we consider an application of our theory 
by analysing a simple probabilistic, connectionless routing protocol, 
showing that it is behaviourally equivalent to a formal specification. 

We conclude our paper by summarising the topics we have covered and by illustrating 
the related work in Section \ref{sec:conclusions}. The topics 
covered in this paper were also the subject of an extended abstract 
\cite{extabs}.

\section{Background}
\label{sec:background}

In this Section we summarise the mathematical concepts, taken
from \cite{DGHM09full}, that will be needed throughout the paper.
First we introduce some basic concepts from  probability theory;
then we show how these can be used to model concurrent systems
which  exhibit  both probabilistic and non-deterministic
behaviour. 

Let $S$ be a set; a function $\Delta: S \rightarrow [0,1]$ is called a
(probability) sub-distribution over $S$ if $\sum_{s \in S} \Delta(s)
\leq 1$.  This quantity, $\sum_{s \in S} \Delta(s)$, is called the mass of the
sub-distribution, denoted as $\size{\Delta}$. If $\size{\Delta} = 1$,
then we say that $\Delta$ is a (full) distribution. The support of a
sub-distribution $\Delta$, denoted $\support{\Delta}$, is the subset of
$S$ consisting of all those elements which contribute to its mass,
namely $\support{\Delta} = \{s \in S\;|\; \Delta(s) > 0\}$. 

For any set $S$, the empty sub-distribution $\varepsilon \in 
\subdist{S}$ is the only sub-distribution with empty support, 
that is $\support{\varepsilon} = \emptyset$.
For each
$s \in S$, the point distribution $\pdist{s}$ is defined to be the
distribution which takes value $1$ at $s$, and $0$ elsewhere.  The set
of sub-distributions and distributions over a set $S$ are denoted by
$\subdist{S}$ and $\dist{S}$, respectively.

Given a family of sub-distributions $\{\Delta_k\}_{k \in K}$,
 $\left(\sum_{k \in K} \Delta_k\right)$ is the partial 
real-valued function in $S \rightarrow {\mathbbm R}$ defined by
$\left(\sum_{k \in K} \Delta_k\right)(s) \mathrel{:=} \sum_{k \in K} \Delta_k(s)$.
This is a partial operation because for a given $s \in S$ this sum  might not exist;
it is also a partial operation on sub-distributions because even if the sum does exist
it may be greater than $1$.

Similarly, if $p \leq 1$ and $\Delta$ is a sub-distribution , then $p\cdot\Delta$ is 
the sub-distribution over $S$ such that $(p\cdot\Delta)(s) = p \cdot \Delta(s)$.

It is not difficult to show that if $\{p_k\}_{k \in K}$ is a sequence 
of positive real numbers such that $\sum_{k \in K}p_k \leq 1$, 
and $\{\Delta_k\}_{k \in K}$ is a family of sub-distributions 
over a set $S$, then $\left(\sum_{k \in K} p_k \cdot \Delta_k\right)$ 
always defines a sub-distribution over $S$.
Further, if $\sum_{k \in K} p_k = 1$ and each $\Delta_k$ is a distribution, then 
$\left(\sum_{k \in K} p_k \cdot \Delta_k\right)$ is a distribution.

Finally, if $f:X \rightarrow Y$ and $\Delta$ is a sub-distribution
over $X$ then we use $f(\Delta)$ to be the sub-distribution over $Y$
defined by:
\begin{align}\label{eq:expected}
  &f(\Delta)(y) = \sum_{x \in X}\setof{\Delta(x)}{f(x)=y}.
\end{align}
This definition can be generalised to two arguments functions; if  
$f : X_1 \times X_2 \rightarrow Y$ is a function, 
and $\Delta, \Theta$ are two sub-distributions respectively 
over $X_1$ and $X_2$, then $f(\Delta, \Theta)$ denotes 
the sub-distribution over $Y$ defined as
\begin{align}\label{eq:twoargexpected}
  &f(\Delta, \Theta)(y) = \sum_{x_1 \in X_1, x_2 \in X_2}\setof{\Delta(x_1)\cdot\Theta(x_2)}{f(x_1,x_2)=y}.
\end{align}\vspace{2 pt}

\noindent Now we turn our attention to probabilistic concurrent systems. The
formal model we  use to represent them is a 
generalisation to a probabilistic setting of 
Labelled Transition Systems (LTSs) \cite{milner}.
\begin{defi}\rm
A \emph{probabilistic labelled transition system} (pLTS) is a 4-tuple\\
$\langle S, \Act_\tau, \rightarrow, \omega \rangle$, where
\begin{enumerate}[label=(\roman*)] 
\item $S$ is a set of states,
\item $\Act_\tau$ is a set of transition labels with a distinguished label $\tau$,
\item the relation $\rightarrow$ is a subset of $S \times \Act_\tau \times \dist{S}$,
\item $\omega: S \mapsto \{\mbox{ true },\mbox{ false }\}$ is a (success) predicate over the states $S$. 
\end{enumerate}
As usual, we will write $s \ar{\mu} \Delta$ in lieu of $(s,\alpha,\Delta) \in \ar{}$.
\qed
\end{defi}\enlargethispage{9 pt}
Before discussing pLTSs, some definitions first: a pLTS whose state space is finite is said to 
be finite state; further, we say that a pLTS $\langle S, \Act_\tau, \rightarrow, \omega \rangle$ 
is finite branching if, for every $s \in S$, the set 
 $\setof{\Delta}{s \ar{\mu} \Delta \,\mbox{ for some $\mu \in \Act_\tau$}   }$ 
 is finite.
 Finally, a finitary pLTS is one which is both finite 
state and finite branching.

We have included in the definition of a pLTS a success predicate $\omega$ over states, which will
be used when testing processes. Apart from this, the 
only difference between LTSs and pLTSs is given by the definition of the transition relation; 
in the latter this is defined to be a relation (parametric in some action $\mu$) between states 
and distributions of states, thus capturing the concept of probabilistic behaviour. 

However, this modification introduces some difficulties when sequences of transitions performed by a 
given pLTS have to be considered, as the domain and the image of the transition relation do not 
coincide. To avoid this problem, we will focus only on distributions of states by defining 
transitions for them. The following Definition serves to this purpose:
\begin{defi}[Lifted Relations]
\label{def:lift}
Let $\calR \subseteq S \times \subdist{S}$ be a relation from states to 
 sub-distributions. Then
$\lift{R} \;\subseteq\; \subdist{S} \times \subdist{S}$ is the smallest relation which satisfies 
\begin{itemize}
\item $s \,\calR\, \Delta$ implies $\pdist{s} \lift{\calR} \Delta$

\item If $I$ is a finite index set and $\Delta_i \lift{\calR} \Theta_i$ for each $i \in I$ then
      $(\sum_{i \in I} p_i \cdot \Delta_i ) \;\calR\; (\sum_{i \in I} p_i \cdot \Theta_i )$ whenever 
       $\left(\sum_{i \in I} p_i\right) \leq 1$.  \qed
\end{itemize}
\end{defi}
\noindent 
Lifting can also be defined for relations from states to probability distributions, 
by simply requiring $\sum_{i \in I} p_i = 1$ in the last constraint of the definition above.

Sometimes it will be convenient to consider also the lifting 
of relations of the form $\calR \subseteq S \times S$; 
this is defined by first lifting the relation 
$\calR$ to $\calR^e \subseteq S \times \subdist{S}$, 
by letting $s \;\calR^e\; \Theta$ iff $\Theta = \pdist{t}$ for some $t \in S$ such that 
$s\;\calR\;t$. Then we obtain the relation $\lift{\calR^e}$ by 
applying Definition \ref{def:lift} to $\calR^e$.

In a pLTS $\langle S, \Act_\tau, \rightarrow, \omega \rangle$, each transition relation 
$\ar{\mu} \subseteq S \times \dist{S}$ can be lifted to
$\lift{(\ar{\mu})} \;\subseteq\; \dist{S} \times \dist{S}$. 
With an abuse of notation, the latter will still be denoted as $\ar{\mu}$.

Lifted transition relations allow us to reason about the behaviour of pLTSs in terms of 
sequences of transitions; here we are mainly interested in the behaviour of a pLTS in 
the long run; that is, given a pLTS $\langle S, \Act_\tau, \rightarrow, \omega \rangle$ 
and a sub-distribution $\Delta \subseteq \subdist{S}$, we are interested in the sub-distributions 
$\Theta \subseteq \subdist{S}$ which can be reached from $\Delta$ after an unbounded 
number of transitions. 

For the moment we will focus only on internal actions of 
a pLTS, in which case the behaviour of a pLTS in the long run is captured by the 
concept of hyper-derivation:
\begin{defi}[Hyper-derivations]
\label{def:hypder}
  In a  pLTS a hyper-derivation consists of a collection of sub-distributions
$\Delta, \Delta_k^{\rightarrow}, \Delta_k^{\Stop}$, for $k \geq 0$,
with the following properties:
$$
\begin{array}{lclcl}    
  \Delta &=& \Delta_0^{\rightarrow} &+& \Delta_0^\Stop    \notag\\
  \Delta_0^{\rightarrow}  &\ar{\tau}& \Delta_1^{\rightarrow} &+& \Delta_1^\Stop \notag \\
&\vdots \\
  \Delta_k^{\rightarrow}  &\ar{\tau}& \Delta_{k{+}1}^{\rightarrow} &+& \Delta_{k{+}1}^\Stop~ \notag \\
&\vdots    &   \notag  \\&&\notag\\
\end{array}
$$
 If $\omega(s) = \mbox{ false } $ for each  $s \in \support{\Delta_k^{\rightarrow}}$ and $ k \geq 0 $
we call $\Delta' = \sum_{k=0}^\infty \Delta_k^{\Stop}$ a
\emph{hyper-derivative} of $\Delta$, and write
$\Delta\dar{}  \Delta'$.
\qed
\end{defi}
\noindent
We will often write $s \dar{} \Delta$ in lieu of $\pdist{s} \dar{} \Delta$.

\begin{exa}
\label{ex:hyperderivative}
Let us illustrate how hyper-derivations can be inferred in a pLTS via a simple example. 
A central role in hyper-derivations will be played by the empty sub-distribution 
$\varepsilon$. Note that, in any pLTS $\langle S, \Act_\tau, \rightarrow, \omega \rangle$, 
for any action $\alpha \in \Act_{\tau}$ we have that $\varepsilon \ar{\alpha} \varepsilon$.

Let us consider a pLTS whose state space is given by the set $\{h,t\}$, with 
the only transition $s \ar{\tau} 1/2 \cdot \pdist{h} + 1/2 \cdot \pdist{t}$ 
and with $\omega(t) = \ttrue$. 
This pLTS models a probabilistic experiment in which we continuously toss a fair coin 
until we obtain the outcome tail (represented by the state $t$), in which case we decree that the experiment succeeded; 
this last constraint is represented by letting $\omega(t) = \ttrue$.
It is well-known, from elementary probability arguments, 
that the probability of obtaining a success before the coin has been tossed $k$ times
is $\frac{2^{k}-1}{2^k}$, while in the long run the experiment will succeed with probability 
$1$. 
This behaviour can be inferred by using hyper-derivations. For example, 
for any $k \geq 0$ we can consider the infinite sequence of transitions 
\begin{align*}
\Delta_0^{\rightarrow} &=&  \pdist{h} &&\ar{\tau}&& \frac{1}{2}\cdot \pdist{h} &&+&& 
\frac{1}{2} \cdot \pdist{t}\\
 \Delta_1^{\rightarrow} &=& \frac{1}{2} \cdot \pdist{h} &&\ar{\tau}&& \frac{1}{2^2} \cdot \pdist{h} &&+&& \frac{1}{2^2} \cdot \pdist{t}\\
 \vdots&&&&\ar{\tau}&&\vdots&&\vdots\\
 \Delta_{k-2}^{\rightarrow} &=& \frac{1}{2^{k-2}} \cdot 
 \pdist{h} &&\ar{\tau}&& \frac{1}{2^{k-1}} \cdot \pdist{h} &&+&& 
 \frac{1}{2^{k-1}} \cdot \pdist{t}\\
\end{align*}
\begin{align*}
 \Delta_{k-1}^{\rightarrow} &=& \frac{1}{2^{k-1}}\cdot\pdist{h} &&\ar{\tau} && 
 \varepsilon &&+&& \frac{1}{2^{k}} \cdot \pdist{h} + \frac{1}{2^{k}} \cdot \pdist{t}\\
 \Delta_{k}^{\rightarrow} &=& \varepsilon &&\ar{\tau} && \varepsilon &&+&& \varepsilon\\
  \vdots&&&&\ar{\tau}&&\vdots&&\vdots\\
  \Delta_{k+n}^{\rightarrow} &=& \varepsilon &&\ar{\tau} && \varepsilon &&+&& \varepsilon\\
    \vdots&&&&\ar{\tau}&&\vdots&&\vdots
\end{align*}
Note that the sequence of transitions above models a situation in which the 
experiment is stopped after the coin has been tossed $k$ times. 
This is done by letting $\Delta_k^\rightarrow = \varepsilon$; 
at least informally this means that the computation proceeds with 
probability $0$ after the $k$-th $\tau$-transition has been performed.
The sequence of transitions above leads to the hyper-derivation 
\begin{eqnarray*}
h &\dar{\;}& \left(\sum_{i=1}^{k-1} \frac{1}{2^{i}} \cdot \pdist{t}\right) 
+ \left(\frac{1}{2^k} \pdist{h} + \frac{1}{2^k} \pdist{t}\right) 
+ \left(\sum_{i=k+1}^{\infty} \varepsilon\right) =\\
&=& \frac{1}{2^k} \cdot \pdist{h} + \left(\sum_{i=1}^k \frac{1}{2^i} \cdot \pdist{t}\right) =\\
&=& \frac{1}{2^k} \cdot \pdist{h} + \frac{2^{k}-1}{2^k} \cdot \pdist{t}
\end{eqnarray*}
That is, after $k$ transitions have been performed the probability of having 
successfully terminated the experiment is $(2^{k}-1)/2^k$.

Further, note that we can use hyper-derivations to describe the limiting behaviour of 
the experiment. In fact, we can consider the infinite sequence of 
transitions 
\begin{align*}
\pdist{h} &&\ar{\tau}&& \frac{1}{2} \cdot \pdist{h} &&+&& \frac{1}{2} \cdot \pdist{t}\\
 \frac{1}{2} \cdot \pdist{h} &&\ar{\tau}&& \frac{1}{2^2} \cdot \pdist{h} &&+&& \frac{1}{2^2} \cdot \pdist{t}\\
 \vdots&&\ar{\tau}&&\vdots&&&&\vdots\\
 \frac{1}{2^{k}} \cdot 
 \pdist{h} &&\ar{\tau}&& \frac{1}{2^{k+1}} \cdot \pdist{h} &&+&& \frac{1}{2^{k+1}} \cdot \pdist{t}\\
  \vdots&&\ar{\tau}&&\vdots&&&&\vdots\\
\end{align*}
which leads to the hyper-derivation 
\[
h \dar{} \left(\sum_{i=1}^{\infty} \frac{1}{2^{i}} \cdot \pdist{t}\right) = \pdist{t}
\]
\end{exa}\medskip

\noindent Hyper-derivations can be seen as the probabilistic counterpart of the weak $\dar{\tau}$ action in 
LTSs; see \cite{DGHM09full} for a detailed discussion. 
Intuitively speaking, they represent fragments of computations obtained by performing only internal 
actions. The last constraint in Definition \ref{def:hypder} is needed since we introduced a success 
predicate in our model; as we see pLTSs as nondeterministic, probabilistic 
experiments, we require that a computation stops  
when the experiment succeeds, that is when a state $s$ such that  
$\omega(s) = \mbox{true}$ has been reached.
States in which the predicate $\omega(\cdot)$ 
is true are called $\omega$-successful.

Further, we are mainly interested in maximal computations of distributions. That is, we require a computation 
to proceed as long as some internal activity can be performed.
To this end, we  say that $\Delta\darE{} \Delta'$ if
\begin{itemize}
\item $\Delta \dar{\;\;\;} \Delta'$,
\item for every  
$s \in \support{\Delta'}$, $s \ar{\tau} $ implies $\omega(s) = \mbox{true}$. 
\end{itemize}
\noindent
This is a mild generalisation of the notion of \emph{extreme derivative} from \cite{DGHM09full}.
Note that the last constraint models exactly the requirement of performing some internal activity 
whenever it is possible;  In other words extreme derivatives correspond to a probabilistic
version of maximal computations.

\begin{exa}
Consider again the pLTS of Example \ref{ex:hyperderivative}. Here we have that 
the hyper-derivation $h \dar{\;} \Theta = \frac{1}{2^k} \cdot \pdist{h} + \frac{2^{k}-1}{2^k} 
 \cdot \pdist{t}$, where $k \geq 0$, is not an extreme derivation, since $\omega(h) = \ffalse$ and $h \ar{\tau}$. 
On the other hand, the hyper-derivation $h \dar{} \pdist{t}$ is also an extreme derivation, 
since $\omega(t) =\ttrue$; therefore $h \darE{} \pdist{t}$.
\end{exa}

\begin{thm}\label{thm:hyper}
  In an arbitrary pLTS
  \begin{enumerate}[label=(\roman*)]
  \item $\dar{}$ is reflexive and transitive,
  \item if $\Delta \dar{} \Delta'$ and $\Delta' \darE{} \Delta''$, then $\Delta \darE{} \Delta''$; 
   this is a direct consequence of the previous statement, and the definition of extreme derivatives,
  \item suppose $\Delta = \left(\sum_{i \in I} p_i \cdot \Delta_i\right)$, where $I$ is an index set and 
  $\sum_{i \in I} p_i \leq 1$. If for any 
  $i \in I, \Delta_i \dar{} \Theta_i$ for some $\Theta_i$, then $\Delta \dar{}\Theta$, where 
  $\Theta = \left(\sum_{i \in I} p_i \cdot \Theta_i\right)$,
  \item for all sub-distributions $\Delta$, there exists a sub-distribution $\Theta$ such that 
   $\Delta \darE{} \Theta$.
  \end{enumerate}
\end{thm}
\begin{proof}
  See \cite{DGHM09full} for detailed proofs. 
\end{proof}

The last definition we need is that of convergent pLTSs. 
Intuitively these are pLTSs whose infinite computations 
have a negligible probability.
\begin{defi}[Convergence]
A pLTS $\langle S, \mbox{Act}_{\tau}, \ar{}, \omega \rangle$ 
is said to be convergent if $s \dar{} \varepsilon$ 
for no state $s \in S$. 
\end{defi}
\noindent
At least informally, $s \dar{} \varepsilon$ means that 
there exists a computation rooted in $s$ which contains 
only probability sub-distributions which 
can always perform a $\tau$-action. 
See \cite{DGHM09full}, Section \textbf{6} 
for a detailed discussion on divergence in pLTSs.
The main property we will require from convergent pLTSs 
	is the following:

\begin{prop}
\label{prop:convergent.moves}
Let $\Delta$ be a distribution in a convergent pLTS. 
If $\Delta \dar{} \Theta$ then $\size{\Theta} = 1$. 
\end{prop}
\begin{proof}
This is an immediate consequence of \emph{Distillation of 
Divergence}, Theorem \textbf{6.20} of \cite{DGHM09full}.
\end{proof}

\section{The calculus}
\label{sec:lang}
In this Section we introduce our calculus for modelling wireless networks. 
In this calculus, a wireless network is modelled as a pair of the form 
$(\Gamma \with M)$, where $\Gamma$ is a digraph representing the 
topology of a wireless network and $M$ is a term which assigns code to 
nodes. 

The syntax of our calculus is presented in Section \ref{sec:syntax}; 
here we also give some basic examples of wireless networks. 
In Section \ref{sec:int.sem} we formalise how networks evolve by 
introducing an intensional semantics for our calculus; 
finally we prove some basic properties of our calculus in 
Section \ref{sec:properties}.

\subsection{Syntax}
\label{sec:syntax}
The calculus we present is designed to model broadcast systems, particularly wireless 
networks, at a high level. We do not deal with low level issues, such as collisions of 
broadcast messages or multiplexing mechanisms \cite{tanenbaum}; instead, 
we assume that network nodes use protocols  at the \emph{MAC level} \cite{macsurvey} 
to achieve a reliable communication 
between nodes.

Basically, the language will contain both primitives for sending and receiving messages and 
will enjoy the following features:
\begin{enumerate}[label=(\roman*)]
\item communication can be obtained through the use of different channels; although the physical medium 
for exchanging messages in wireless networks is unique, it is reasonable to assume that network nodes use 
some multiple access technique, such as \emph{TDMA} or \emph{FDMA} \cite{tanenbaum}, to setup and 
communicate through virtual channels,
\item communication is broadcast; whenever a node in a given network sends a message, it can be detected by 
all nodes in its range of transmission,
\item communication is reliable; whenever a node broadcasts 
a message and a neighbouring node (that is, 
a node in the sender's range of transmission) is waiting 
to receive a message on the same channel, then the 
message will be delivered to the receiver. This is not 
ensured if low level issues are considered, 
as problems such as message collisions \cite{macsurvey}
and nodes synchronisation \cite{time} arise .
\end{enumerate}

\begin{figure}[t]
\rule{\linewidth}{0.5mm}
  \begin{align*}
&
\begin{array}{lcll}
  M,\;N  &::=&                 &\textbf{Systems} \\
         &  &  \Cloc{s}{n}    & \text{Nodes} \\
         &  & M \Cpar N       &  \text{Composition}\\
         &  & \Cnil           & \text{Identity}
\\\\
  p,\;q  &::=&                 &\textbf{(probabilistic) Processes}\\
         &  & s               &\\
         &  &   p \probc{p} q &\text{probabilistic choice} \\
\\
  s,\;t  &::= &               &\textbf{States}\\
         &     &c!\pc{e}.p           &\text{broadcast}\\
         &    &c?\pa{x}.p            &\text{receive}\\
         &    &\omega.\Cnil            &\text{test}\\
         &    &  s + t              &\text{choice} \\
         &    &\Cmatch{b}{s}{t}     &\text{branch}\\
         &    &  \tau.p             &\text{preemption}\\
         &     & A\pa{\tilde{x}}    &\text{definitions}\\
         &    & \Cnil               &\text{terminate}
\end{array}
\end{align*}

  \caption{Syntax \label{fig:syntax}}
 \rule{\linewidth}{0.5mm} 
\end{figure}

\noindent The language for system terms, ranged over by $M,\,N,\cdots$ is given in
Figure~\ref{fig:syntax}.  Basically a system consists of a collection
of named nodes at each of which there is some running code.  The
syntax for this code is a fairly straightforward instance of a standard
process calculus, augmented by a probabilistic choice; code
descriptions have the usual constructs for channel based communication,
with input $c?\pa{x}.p$ being the unique binder. 
This gives rise to the standard notions of alpha-conversion, 
free and bound occurrences of variables in system terms and 
closed system terms.

We only consider the 
sub-language of well-formed system terms in which all node names have
at most one occurrence.  We use $\sys$ to range over all closed
well-formed terms.  A (well-formed) system term can be viewed as a
mapping that assigns to node names the code they are executing. A
sub-term $\Cloc{s}{n}$ appearing in a system term $M$ represents node
$n$ running code $s$.
In the following we make use of standard conventions; we omit trailing 
occurrences of $\Cnil$ and we use $\prod_{i=1}^n \Cloc{s_i}{m_i}$ to 
denote the system term $\Cloc{s_1}{m_1} \Cpar \cdots \Cpar \Cloc{s_n}{m_n}$.

Additional information such as the
connectivity between nodes of a network is needed to formalise
communications between nodes. 
Network connectivity is represented by a graph $\Gamma = 
\pair{\Gamma_V}{\Gamma_E}$; here $\Gamma_V$ is a finite set of nodes
and $\Gamma_E \subseteq (\Gamma_V \times \Gamma_V)$ is an irreflexive relation 
between nodes in $\Gamma_V$. 
Intuitively, $(m,n) \in \Gamma_E$ models the possibility for node $n$ to detect 
broadcasts fired by $m$.
 
We use the more graphic notation $\Gamma \vdash v$ to mean $v \in \Gamma_V$ and 
$\Gamma \vdash \rconn{m}{n}$ for $(m,n) \in \Gamma_E$. 
Similarly we use $\Gamma \vdash \lconn{n}{m}$ to denote
$\Gamma \vdash \rconn{m}{n}$. 
Sometimes we also use the notations $\Gamma \vdash \isconn{m}{n}$ for 
$\{(n,m), (m,n)\} \subseteq \Gamma_E$ and 
$\Gamma \vdash \someconn{m}{n}$ to denote either $\Gamma \vdash \rconn{m}{n}$ 
or $\Gamma \vdash \lconn{m}{n}$.

A \emph{network} consists of a pair $(\Gamma \with M)$,
representing the system $M$, from $\sys$, executing relative to the connectivity graph
$\Gamma$.  All nodes occurring in $M$, $\nodes{M}$, will appear in
$\Gamma$ and the effect of  running  the code at $n \in \nodes{M}$
will depend on the connectivity of $n$ in $\Gamma$.  But in general
there will be nodes in $\Gamma$ which do not occur in $M$; let
$\type{\Gamma \with M} = \Gamma_V \setminus \nodes{M}$;  we call this set 
 the \emph{interface}
of the network $\Gamma \with M$, and its elements are called \emph{external nodes}.  Intuitively
these are nodes which may be used to compose the network $\Gamma
\with M$ with other networks, or to place code for testing the
behaviour of $M$.  

In the following we use the meta-variables $\calM, \calN, \cdots$ to 
range over networks. Also, the notation introduced for system 
terms and connectivity graphs is extended to networks in 
the obvious way; for example, if $\calM = (\Gamma \with M)$, 
$\nodes{\calM} = \nodes{M}$, $\calM_V = \Gamma_V$ and 
$\calM \vdash \rconn{m}{n}$ denotes $\Gamma \vdash \rconn{m}{n}$.

\begin{figure}[t]
  
                         
\begin{align*}
     \begin{tikzpicture}
          \node[state](m){$m$}; 
          \node[state](n)[below left=of m]{$n$}; 
          \node[state](i)[left=of m]{$i$}; 
           \node[state](e)[below =of n]{$e$}; 
           \node[state](o1)[right=of m]{$o_1$};
           \node[state](o2)[below right=of m]{$o_2$};
 \path[to]
       (m) edge [thick] (o1)
       (m) edge[thick] (o2)
       (e) edge[thick] (n)
       (n) edge[thick] (i)
       (i) edge[thick] (m);
   \begin{pgfonlayer}{background}
    \node [background,fit=(m) (n) (i)] {};
    \end{pgfonlayer}
    \end{tikzpicture}
&&  
     \begin{tikzpicture}
          \node[state](m){$m$}; 
          \node[state](n)[below left=of m]{$n$}; 
          \node[state](i)[left=of m]{$i$}; 
           \node[state](e)[below =of n]{$e$}; 
           \node[state](k)[below =of m]{$k$}; 
           \node[state](o1)[right=of m]{$o_1$};
           \node[state](o2)[right=of k]{$o_2$};
 \path[to]
       (m) edge [thick] (o1)
       (m) edge[thick] (k)
       (k) edge[thick] (o2)
       (e) edge[thick] (n)
       (n) edge[thick] (i)
       (i) edge[thick] (m);
   \begin{pgfonlayer}{background}
    \node [background,fit=(m) (n) (k)] {};
    \end{pgfonlayer}
    \end{tikzpicture}
\\
\calM = \Gamma_M \with \Cloc{A_n}{n} \Cpar \Cloc{A_i}{i} \Cpar \Cloc{A_m}{m} 
&&
\calN = \Gamma_N \with \Cloc{A_n}{n} \Cpar \Cloc{A_i}{i} \Cpar \Cloc{A_k}{m}  \Cpar \Cloc{A_m}{k}
\end{align*}

  \caption{Example networks}
  \label{fig:ex1}
\end{figure}

\begin{exa}\label{ex:ex1}
Consider $\calM$ described in Figure~\ref{fig:ex1}. 
There are six nodes, three occupied by code $n$, $i$ and $m$, and three in the interface
$\type{\calM}$ ,
$e, o_1$ and $o_2$. Here, and in future examples,  we  differentiate between the interface
and the occupied (internal) nodes using shading. Suppose the code at nodes is given by
\begin{align*}
  & A_n \Leftarrow c?\pa{x}.d!\pc{x}.\Cnil
  &&
  & A_i \Leftarrow d?\pa{x}.d!\pc{f(x)}.\Cnil
  &&
  &A_m \Leftarrow d?\pa{x}.(d!\pc{x}.\Cnil \,\probc{0.8}\, \Cnil)
\end{align*}
Then $\calM$ can receive input from node $e$ at its interface along the channel
$c$; this is passed on to the internal node $i$ using channel $d$, where it is transformed
in some way, described by the function $f$\footnote{For example, 
if we assume the set of closed values to be $\mathbb{Z}$, f could be the mapping 
$f: z \mapsto z^2$.}, and then forwarded to node $m$, where $80\%$ of the
time it is  broadcast to the external nodes $o_1$ and $o_2$. The remainder of the
time the message is lost. 

The network $\calN$ has the same interface as $\calM$, but has an extra internal node
$k$ connected to $o_2$,
and $m$ is only connected to one interface node $o_1$ and the internal node $k$. The nodes $i$ and $n$ 
have the same code running as in $\calM$, while nodes $m$ and $k$ will run the code
\[
A_k \Leftarrow d?\pa{x}.(d!\pc{x}.\Cnil \probc{0.9} \Cnil)
\]
Intuitively, the behaviour of $\calN$ is more complex than that of $\calM$; indeed, there is the 
possibility for a computation of $\calN$ to deliver a value only to one between the external nodes 
$o_1$ and $o_2$, while this is not possible in $\calN$. However, $81\%$ of the times this message 
will be delivered to both these nodes, and thus it is more reliable than $\calM$.
Suppose now that we change the code at the intermediate node $m$ in \calM, 
\begin{align*}
  &\calM_1 = \Gamma_M \with \ldots   \Cpar \Cloc{B_m}{m}
  &&
  &\text{where}\;\; B_m \Leftarrow d?\pa{x}.(\tau.(d!\pc{x}.\Cnil \probc{0.5} \Cnil ) + \tau.d!\pc{x}.\Cnil)
\end{align*}
In $\calM_1$ the behaviour at the node $m$ is non-deterministic; it may act like a perfect forwarder,
or one which is only 50\% reliable.  Optimistically it could be more reliable than $\calM$, or pessimistically it 
could be less reliable than the latter. Further, there is no possibility for the network $\calM_1$ to forward 
the message  to only one of the external nodes $o_1$, $o_2$, so that its behaviour is somewhat less complex than 
 that of $\calN$.

As a further variation let $\calM_2$ be the result of replacing the code at $m$ with 
\begin{align*}
  C_m &\Leftarrow d?\pa{x}. D \\
  D   &\Leftarrow \tau.(d!\pc{x}.\Cnil \,\probc{0.5}\, \tau.D)
\end{align*}
Here the behaviour is once more deterministic, with the probability that the message will be
eventually transmitted successfully through node $k$ approaching $1$ in the limit. Thus, this 
network is as reliable as $\calM_1$,  when the latter is viewed optimistically.
\end{exa}
\subsection{Intensional Semantics}
\label{sec:int.sem}
We now turn our attention on the operational semantics of networks.  
Following \cite{DGHM09full}, processes will be interpreted as probability 
distributions of states; such an interpretation is encoded by the function $\interprP{\cdot}$ 
defined below: 
\begin{eqnarray*}
  \interprP{s} &=& \pdist{s} \\
  \interprP{p_{1} \probc{p} p_{2}} &=& p \cdot \interprP{p_{1}} + (1 - p) \cdot \interprP{p_{2}}.
\end{eqnarray*}
Sometimes we will need to consider the probability distribution associated to system terms; 
this is done by letting 
\begin{eqnarray*}
  \interprP{\Cnil} &=& \pdist{\Cnil}\\
  \interprP{\Cloc{p}{n}} &=& \Cloc{\interprP{p}}{n}\\
  \interprP{M_1 \Cpar M_2} &=& \interprP{M_1} \Cpar \interprP{M_2}
\end{eqnarray*}
where $\Cloc{\interprP{p}}{n}$ represents a distribution over $\sys$, 
obtained by a direct application of Equation \eqref{eq:expected} 
to the function $\Cloc{\cdot}{n}$ which maps states into system terms. 
Similarly, $\interprP{M_1 \Cpar M_2}$ is obtained by 
applying Equation \eqref{eq:twoargexpected} to the function 
$(\cdot \Cpar \cdot): \sys \times \sys \rightarrow \sys$.

The intensional semantics of networks is defined incrementally. 
We first define a pre-semantics for states, which is then used 
for giving the judgements of (state based) networks.

\begin{figure}[t]
\rule{\linewidth}{0.5mm}
  \begin{alignat*}{2}
    &\linferSIDE[\Slts{Snd}]
       {}
       {c!\pc{e}.p \ar{c!\interpr{e}} \interprP{p}\Us}
       {}
    &\qquad
    &
 \linferSIDE[\Slts{Rcv}]
       {}
       {c?\pa{x}.p \ar{c?v} \interprP{p\{v/x\}}\Us}
       {}
\\
  &\linferSIDE[\Slts{\tau}]
       {}
       {\tau.p \ar{\tau} \interprP{p}}
       {}
    &\qquad
    & \linfer[\Slts{Suml}]
       {s \ar{\alpha} \Delta}
       {s+t \ar{\alpha} \Delta}
\\
  &\linfer[\Slts{SumR}]
        {t \ar{\alpha} \Delta}
        {s+t \ar{\alpha} \Delta}
   &\qquad
   &\linfer[\Slts{then}]
         {s \ar{\alpha} \Delta}
         {\Cmatch{b}{s}{t} \ar{\alpha} \Delta}
         {\interpr{b} = \mbox{true}}
\\
\\  
 &\linfer[\Slts{else}]
          {t \ar{\alpha} \Delta}
          {\Cmatch{b}{s}{t} \ar{\alpha} \Delta}
          {\interpr{b} = \mbox{false}}
  &
  &\linferSIDE[\Slts{Call}]
         {s\{\tilde{e}/\tilde{x}\} \ar{\mu} \Delta}
         {A\pa{\tilde{e}} \ar{\mu} \Delta}
         {A\pa{\tilde{x}} \Leftarrow s}
  \end{alignat*}
  \caption{Pre-semantics of states}
  \label{fig:stsem}

\rule{\linewidth}{0.5mm}
\end{figure}

The pre-semantics for states takes the form
\begin{align*}
  s \ar{\alpha} \Delta
\end{align*}
where $s$ is a closed state, that is containing no free occurrences of a variable,  
$\Delta$  is a distribution of states and $\mu$ can take one of the forms $c!v,\,c?v$ or $\tau$. 
The deductive rules for inferring these judgements are given in Figure \ref{fig:stsem}
and should be self-explanatory. It assumes some mechanism for evaluating closed data-expressions
$e$ to values $\interpr{e}$. Note that we assume that definitions have the form $A\pa{\tilde{x}} \Leftarrow s$, 
where $s$ is a state; this is because actions for definitions are inherited by the state that is 
associated to them, and judgements are not defined for (probabilistic) processes. Also, note that 
we have a special state $\omega$ for which no rule has been defined. The role of 
this construct will become clear in Section \ref{sec:compositional}.

\begin{figure}[t]
\rule{\linewidth}{0.5mm}
  \begin{alignat*}{2}
    &\linfer[\Rlts{broad}]
       {s \ar{c!v} \Delta}
       {\Gamma \with \Cloc{s}{n}  \ar{n.c!v} \Cloc{\Delta}{n}\Us}
       {}
    &\qquad
    &\linferSIDE[\Rlts{rec}]
       {s \ar{c?v} \Delta}
       {\Gamma \with \Cloc{s}{n}  \ar{m.c?v} \Cloc{\Delta}{n}\Us}
       {\Gamma \vdash \rconn{m}{n}}
\\
  &\linferSIDE[\Rlts{deaf}]
       {s \nar{c?v} }
       {\Gamma \with \Cloc{s}{n}  \ar{m.c?v} \pdist{\Cloc{s}{n}}\Us}
       {}
    &\qquad
    &\linferSIDE[\Rlts{disc}]
       {}
       {\Gamma \with \Cloc{s}{n}  \ar{m.c?v} \pdist{\Cloc{s}{n}}\Us}
       {\Gamma \vdash \notrconn{m}{n}}
\\
   &\linfer[\Rlts{\Cnil}]
        {}
        {\Cnil  \ar{m.c?v} \;\pdist{\Cnil}\Us}
\\
   &\linfer[\Rlts{\tau}]
        {s \ar{\tau} p}
        {\Gamma \with \Cloc{s}{n}  \ar{n.\tau} \Cloc{\Delta}{n}\Us}
        {\interprP{p} = \Delta}
   &\qquad
   &\linfer[\Rlts{\tau.prop}]
           {\Gamma \with M \ar{n.\tau} \Delta }
           {\Gamma \with M \Cpar N \ar{n.\tau} \Delta \Cpar \pdist{N}\Us}
\\
 &\linfer[\Rlts{prop}]
         {\Gamma \with M \ar{m.c?v} \Delta,\;\Gamma \with N \ar{m.c?v} \Theta }
         {\Gamma \with M \Cpar N \ar{m.c?v} \Delta \Cpar \Theta\Us}
 &\qquad
 &
\linfer[\Rlts{sync}]
         {\Gamma \with M \ar{m.c!v} \Delta,\;\Gamma \with N \ar{m.c?v} \Theta }
         {\Gamma \with M \Cpar N \ar{m.c!v} \Delta \Cpar \Theta\Us}
  \end{alignat*}
  \caption{Intensional semantics of networks}
  \label{fig:opsem}

\rule{\linewidth}{0.5mm}
\end{figure}

Judgements in the intensional semantics of networks take the form
\begin{align*}
  & \Gamma \with M \ar{\mu} \Delta
\end{align*}
where  $\Gamma$ is a network connectivity, $M$ is a system from $\sys$, 
and $\Delta$ is a distribution over $\sys$; 
intuitively this means that relative to the  connectivity $\Gamma$ the 
system $M$ can perform 
the action $\mu$, and with probability $\Delta(N)$ be transformed into 
the system $N$, for every  $N \in \support{\Delta}$. 
 The action labels can take the form
\begin{enumerate}[label=(\roman*)]
\item receive, $n.c?v$, meaning that the value $v$ is detected on 
channel $c$ by all nodes in $\nodes{M}$ which are reachable from $n$ in $\Gamma$,

\item broadcast, $n.c!v$: meaning the node $n$ (occurring in
  $\nodes{M}$, and therefore in $\Gamma$) broadcasts the value $v$ on
  channel $c$ to all nodes directly connected to $n$ in $\Gamma$

\item internal activity, $n.\tau$, meaning some internal activity performed 
by node $n$.
\end{enumerate}

\noindent The rules for inferring judgements are given in
Figure~\ref{fig:opsem}. Here we have omitted the 
symmetric counterparts of rules $\Rlts{Sync}$ and 
$\Rlts{TauProp}$.
Rule $\Rlts{broad}$ models the capability for a node to broadcast a value $v$ through 
channel $c$, assuming the code running there
is capable of broadcasting along $c$. 
Here the distribution $\Delta$ is in turn obtained from the residual of the state $s$ after the broadcast action. 
\begin{exa}
  Consider the simple network $\Gamma \with \Cloc{s}{n}$ where $\Gamma$ is 
  an arbitrary connectivity graph and the code $s$ has the form 
$c!\pc{v}. (s_1 \probc{\nicefrac{1}{4}} s_2)$. 

The pre-semantics of states determines that $s \ar{c!v} 
\interprP{s_1 \probc{\nicefrac{1}{4}} s_2} = \nicefrac{1}{4} \cdot \pdist{s_1} + \nicefrac{3}{4} \cdot \pdist{s_2}$, 
using Rule $\Rlts{Snd}$
Thus according to the rule  
\Rlts{broad} we have the judgement
\begin{align*}
  \Gamma \with \Cloc{s}{n}  \,\ar{n.c!v}\,  \frac{1}{4} \cdot  \Cloc{s_1}{n}  + \frac{3}{4} \cdot \Cloc{s_2}{n}
\end{align*}
\end{exa}\vspace{6 pt}

\noindent Rules $\Rlts{rec}, \Rlts{deaf}$ and $\Rlts{disc}$ express how a node reacts 
when a message is broadcast; the first essentially models the capability of a node which is 
listening to a channel $c$, and which appears in the sender's range of transmission, 
to receive the message correctly. The other two rules model situations in which a node is 
not listening to the channel used to broadcast a message, or it is not in the range of transmission of the sender;
 In both cases this node cannot detect the transmission at all.

The rules $\Rlts{\tau}$ and $\Rlts{\tau.prop}$ model internal activities performed 
by some node of a system term; the latter (together with its symmetric 
counterpart) expresses the inability for a node which 
performs an internal activity to affect other nodes in a system term. Here again, $\Delta \Cpar \Theta$ 
is a distribution over $\sys$, this time obtained by instantiating Equation \eqref{eq:twoargexpected} to 
the function $(\cdot \Cpar \cdot) : \sys \times \sys \rightarrow \sys$.

Finally, rules $\Rlts{sync}$ and $\Rlts{prop}$ describe how communication 
between nodes of a network is handled; here the result of a synchronisation between 
an output and an input is again an output, thus modelling broadcast 
communication \cite{Prasad95}.


%
%

\subsection{Properties of the Calculus}
\label{sec:properties}
We conclude this section by summarising the main properties enjoyed by the 
intensional semantics of our calculus.

Here (and in the rest of the paper) 
it will be convenient to identify networks 
and distributions of networks
up to a structural congruence relation $\equiv$. This is first 
defined for states as the smallest equivalence relation which is 
a commutative monoid with 
respect to $\;+\;$ and $\Cnil$, and which satisfies the equations 
$\Cmatch{b}{s}{t} \equiv s$ if $\interpr{b} = \ttrue$, 
$\Cmatch{b}{s}{t} \equiv t$ if $\interpr{b} = \ffalse$ and 
$A\pa{\tilde{e}} \equiv s\{\tilde{e}/\tilde{x}\}$ if 
$A\pa{\tilde{x}} \Leftarrow s$.
For system terms, we let $\equiv$ 
be the smallest equivalence relation which is a commutative 
monoid with respect to $(\cdot\Cpar\cdot)$ and $\Cnil$, and which satisfies the 
equation $s \equiv t$ implies $\Cloc{s}{m} \equiv \Cloc{t}{m}$ 
for any node $m$. Finally, we let $(\Gamma_M \with M) \equiv 
(\Gamma_N \with N)$ iff $\Gamma_1 = \Gamma_2$ and $M \equiv N$. 
Structural congruence is also defined for distributions of 
networks via the lifted relation $\lift{\equiv^e}$. 
With an abuse of notation, the latter is still 
denoted as $\equiv$.

The properties that we prove in this section give an explicit 
form to the structure of a network $(\Gamma \with M)$ and 
a distribution $\Delta$ in the case that an action $(\Gamma \with M) 
\ar{\mu} \Delta$ can be inferred in the intensional semantics 
presented in Section \ref{sec:int.sem}.

\begin{prop}[Tau-actions]
\label{prop:tau}
Let $\Gamma \with M$ be a network; then 
$\Gamma \with M \ar{m.\tau} \Delta$ if and only if 
\begin{enumerate}[label=(\roman*)]
\item $M \equiv \Cloc{(\tau.p) + s}{m} \Cpar N$,
\item $\Delta \equiv \interprP{\Cloc{p}{m}} \Cpar \pdist{N}$
\end{enumerate}
\end{prop}

\begin{proof}[Outline of the Proof]
We first need to prove a similar statement for states. 
Let $t$ be a state; then $t \ar{\tau} \Theta$ if and 
only if $t \equiv \tau.p + s$ for some $p, s$ such that 
$\interprP{p} = \Theta$. The two implications 
of this statement are proved separately. 


To prove Proposition \ref{prop:tau} suppose first 
that $\Gamma \with M \ar{m.\tau} \Delta$ for 
some distribution $\Delta$. We show that 
$M \equiv \Cloc{\tau.p + s}{m} \Cpar N$, $\Delta 
\equiv \interprP{\Cloc{p}{m}} \Cpar \pdist{N}$ 
by structural induction on the proof of the derivation above. 

If the last rule applied is $\Rlts{\tau}$, then $M = \Cloc{t}{m}$ 
for some $t$ such that $t \ar{\tau} \Theta$ , $\Delta = \Cloc{\Theta}{m}$. 

Since $t \ar{\tau} \Theta$ then $t \equiv \tau.p + s$ for some 
process $p$ such that $\interprP{p} = \Theta$, which also gives 
$\Delta = \interprP{\Cloc{p}{m}}$. 
By definition of structural congruence $M \equiv \Cloc{\tau.p + s}{m} 
\equiv \Cloc{\tau.p + s}{m} \Cpar \Cnil$. Further, $\Delta 
\equiv \interprP{\Cloc{p}{m}} \Cpar \pdist{\Cnil}$, and there 
is nothing left to prove.

If the last rule applied is $\Rlts{\tau.prop}$, then 
$M \equiv N \Cpar L$ for some $N, L$ such that 
$N \ar{m.\tau} \Delta_N$; further, $\Delta = \Delta_N \Cpar \pdist{L}$. 
By inductive hypothesis we have that $N \equiv \Cloc{\tau.p + s}{m} \Cpar N'$ 
and $\Delta_N \equiv \interprP{\Cloc{p}{m}} \Cpar \pdist{N'}$. 
By performing some simple calculations we find that 
$M \equiv N \Cpar L \equiv \Cloc{\tau.p + s}{m} \Cpar (N' \Cpar L)$ 
and $\Delta \equiv \Delta_N \Cpar \pdist{L} \equiv \interprP{\Cloc{p}{m}} \Cpar 
( \pdist{N'} \Cpar \pdist{L} ) \equiv \interprP{\Cloc{p}{m}} \Cpar (\pdist{N' \Cpar L})$.

Conversely, suppose that $M \equiv \Cloc{\tau.p + s}{m} \Cpar N$; in 
this case it suffices to perform a rule induction on the proof of 
the equivalence above to show that $\Gamma \with M \ar{m.\tau} 
\Delta$, where $\Delta \equiv \interprP{\Cloc{p}{m}} \Cpar \pdist{N}$.
\end{proof}

\begin{prop}[Input]
\label{prop:input}
For any network $\Gamma \with M$ we have that 
$\Gamma \with M \ar{m.c?v} \Delta$ iff
\begin{enumerate}[label=(\roman*)]
\item $m \notin \nodes{M}$, 
\item $M \equiv \prod_{i \in I} \Cloc{(c?\pa{x}.p_i) + s_i}{n_i} 
\Cpar \prod_{j \in J} \Cloc{s_j}{n_j}$, 
\item for any $i \in I$, $\Gamma \vdash \rconn{m}{n_i}$, 
\item for any $j \in J$, either $\Gamma \vdash \notrconn{m}{n_j}$ 
or $s_j \ar{c?v}\;\hspace{-15pt}{\not\;}{\hspace{15pt}}$, 
\item $\Delta \equiv \interprP{\prod_{i \in I} \Cloc{p_i\{v/x\}}{n_i}} 
\Cpar \prod_{j \in J} \pdist{\Cloc{s_j}{n_j}}$.
\qed
\end{enumerate}

\end{prop}

\begin{prop}[Broadcast]
\label{prop:broadcast}
Let $\Gamma \with M$ be a network; then $\Gamma \with M \ar{m.c!v} 
\Delta$ for some $\Delta$ iff 
\begin{enumerate}[label=(\roman*)]
\item $M \equiv \Cloc{(c!\pc{e}.p + s)}{m} \Cpar N$, 
where $\interpr{e} = v$, 
\item $\Gamma \with N \ar{m.c?v} \Theta$, 
\item $\Delta \equiv \interprP{\Cloc{p}{m}} \Cpar \Theta$.
\qed
\end{enumerate}

\end{prop}

\noindent An immediate consequence of the results above is 
that actions in the intensional semantics 
are preserved by structurally congruent networks.
\begin{cor}
\label{cor:reduction}
Let $\calM, \calN$ be two networks 
such that $\calM \equiv \calN$. 
If $\calM \ar{\mu} \Delta$ then 
$\calN \ar{\mu} \Theta$ for some 
$\Theta$ such that $\Delta \equiv \Theta$.
\qed
\end{cor}

Another trivial consequence that follows from 
the results above is that the intensional 
semantics does not change the structure 
of a network. 

\begin{defi}[Stable distributions]
A (node)-stable sub-distribution $\Delta \in 
\subdist{\sys}$ is one for which 
whenever $M, N \in \support{\Delta}$ 
it follows that $\nodes{M} = \nodes{N}$.
A distribution over networks is said to 
be (node)-stable if it has the form 
$\Gamma \with \Delta$, and $\Delta$ 
is a stable sub-distribution in 
$\subdist{\sys}$.
\end{defi}
\begin{cor}
\label{cor:stable.dist}

Whenever $\Gamma\with M \ar{\mu} \Delta$ 
then $\Delta$ is node-stable; further, 
for any $N \in \support{\Delta}$ 
we have that $\nodes{M} = \nodes{N}$.\qed
\end{cor}
%

\section{Compositional Reasoning for Networks}
\label{sec:compositional}

The aim of this Section is to  
develop preorders of the form
\begin{align}\label{eq:behaviour}
  \calM \behavPre \calN
\end{align}
Intuitively this means that the network $\calM$ can be replaced by
$\calN$, as a part of a larger overall network, without any loss of
behaviour. The intention is that the internal structure of the
  networks $\calM,\, \calN$ should play no role in this comparison; the
  names used to identify their internal stations and the their
  communication topology should not be important. Intuitively the only
  behaviour to be taken into account in this extensional comparison is
  the reception of values at the their interface, the values
  subsequently broadcast at the interface.  

To formalise this concept 
we need to  say how networks are composed to form larger networks.
In  
Section~\ref{sec:comp.net} we propose a specific composition  operator, $\testP$
for this purpose, and briefly discuss  its properties.  We then use
this operator in Section~\ref{sec:ts} to say how network behaviour is determined. 
In Section~\ref{sec:maypreord}, we  give 
the formal definition of the behavioural preorders, in a relatively standard manner following 
\cite{dnh}; this section also treats some  examples.
The nature of these preorders depends on our particular choice of
composition  operator $\testP$. In Section~\ref{sec:justify} we return
to this point and offer a justification for our choice; this section may
be safely ignored by the reader who is uninterested in this subtlety.

However first let us reconsider the informal requirements of the
proposed behavioural preorder (\ref{eq:behaviour}) above. 
We have already mentioned that it should not depend on the internal 
structure of $\calN, \calM$. But equally well $\calN, \calM$ 
should not be able to  make any assumptions about the topology 
of their external  environment. 
\begin{figure}
\begin{align*}
     \begin{tikzpicture}
           \node[state](m){$m$}; 
           \node[state](o1)[above right=of m]{$o_1$};
           \node[state](o2)[below right=of m]{$o_2$};
 \path[to]
       (m) edge[thick] (o1)
       (m) edge[thick] (o2);
   \begin{pgfonlayer}{background}
    \node [background,fit=(m)] {};
    \end{pgfonlayer}
    \end{tikzpicture}
&&  
     \begin{tikzpicture}
           \node[state](m){$m$}; 
           \node[state](o1)[above right=of m]{$o_1$};
           \node[state](o2)[below right=of m]{$o_2$};     
 \path[to]
       (m) edge[thick] (o1)
       (m) edge[thick] (o2);
  \path[tofrom]
  			(o1) edge[thick] (o2);
   \begin{pgfonlayer}{background}
    \node [background,fit=(m)] {};
    \end{pgfonlayer}
    \end{tikzpicture}\\
\calM = \Gamma_M \with \Cloc{c!\pc{v}}{m}
&&
\calN = \Gamma_N \with \Cloc{c!\pc{v}}{m}
\end{align*}

\caption{A well-formed and an ill-formed network}
\label{fig:wellformed}
\end{figure}

\begin{exa}
\label{ex:wellformed}
Consider the networks $\calM$ and $\calN$ depicted in 
Figure \ref{fig:wellformed}. At least intuitively the extensional 
behaviour of these two networks is the same: a broadcast of 
value $v$ along channel $c$ can be detected by the external nodes 
$o_1$ and $o_2$. However $\calN$ makes an assumption about the
external environment, namely that there is a connection between
the external nodes $o_1$ and $o_2$. This slight difference 
can be  exploited to distinguish between them behaviourally.
 Suppose that we place the code $c?\pa{x}.c!\pc{w}$ at 
node $o_1$ and the code $c?\pa{x}.c?\pa{y}.\omega$ at 
node $o_2$ to test the behaviour of both networks. 
In practice, let $T = \Cloc{c?\pa{x}.c!\pc{w}}{o_1} \Cpar 
\Cloc{c?\pa{x}.c?\pa{y}.\omega}{o_2}$, and consider the 
networks $\calM' = \Gamma_M \with \Cloc{c!\pc{v}}{m} \Cpar T$ and 
$\calN' = \Gamma_N \with \Cloc{c!\pc{v}}{m} \Cpar T$. 
In the first network the node $o_2$ can detect both the 
broadcast fired by node $m$ and node $o_1$, leading to 
a state in which the special action $\omega$ is enabled. 
However, the same is not possible in $\calN'$, since there 
is no connection between the external nodes $o_1$ and $o_2$. 
That is, node $o_2$ can only detect the broadcast fired by 
node $m$, ending in a state in which the special action 
$\omega$ remains guarded by an input. As we will see later, 
the clause $\omega$ plays a crucial role in distinguishing 
networks.
\end{exa}

The problem in  Example \ref{ex:wellformed} is caused 
by the presence of a connection between the two 
external nodes in the network $\calN$ 
of Figure \ref{fig:wellformed};
intuitively this represents an assumption of 
$\calN$ about its external environment. To avoid this 
problem, we focus on a specific 
class of networks in which connections between 
external nodes are not allowed. Also, we 
require external nodes to have at least a connection 
with some internal node.

\begin{defi}[Well-Formed Networks]
\label{def:well.formed}
A network $\calM$ is well-formed iff 
\begin{enumerate}[label=(\roman*)]
\item whenever $\calM \vdash \someconn{m}{n}$ then 
either $m \in \nodes{\calM}$ or $n \in \nodes{\calM}$,
\item whenever $\calM \vdash m$ and $\calM \vdash \someconn{m}{n}$ 
for no $n \in (\calM)_V$, then $m \in \nodes{\calM}$. 
\qed
\end{enumerate}
\end{defi}
\noindent 
Henceforth we will only focus on well-formed networks, unless stated otherwise. 
We denote the set of well-formed networks as $\nets$.

Finally let us provide some definitions which will be useful in the sequel.
Let $\calM$ be a (well-formed) network. We say that a 
node $i \in \type{\calM}$ is an \emph{input node} if $\calM \vdash 
\lconn{m}{i}$ for some $m$\footnote{Note that by well-formedness this implies 
$m \in \nodes{\calM}$.}; conversely, if a node $o \in \type{\calM}$ is 
such that $\Gamma \vdash \rconn{m}{o}$ we say 
that $o$ is an \emph{output node}. If we let $\inp{\calM} = 
\{i \;|\; i \mbox{ is an input node in  } \calM\}$ and $\outp{\calM} = 
\{o \;|\; o \mbox{ is an output node in } \calM\}$ it is easy to check that 
$\type{\calM} = \inp{\calM} \cup \outp{\calM}$.


\subsection{Composing Networks}
\label{sec:comp.net}
One use of composition operators is to enable compositional reasoning. For example the task of establishing
\begin{align}\label{eq:behav}
  {\mathcal N}_1 \behavPre {\mathcal N}_2
\end{align}
can be simplified if we can discover a common component, that is some ${\mathcal N}$ such that 
\begin{math}
   {\mathcal N}_1 =  {\mathcal M}_1 \interleave  {\mathcal N}
\end{math}
and 
\begin{math}
   {\mathcal N}_2 =  {\mathcal M}_2 \interleave  {\mathcal N}
\end{math}
for some composition operator $\interleave$; then (\ref{eq:behav}) can be reduced to establishing 
\begin{align*}
  {\mathcal M}_1 \behavPre {\mathcal M}_2,
\end{align*}
assuming that the behavioural preorder in question, $\;\behavPre$, is preserved by the composition operator $\interleave$. 

However another use of a composition operator is in the definition of the behavioural preorder $\behavPre$ itself.
Intuitively we can define 
\begin{align}
  {\mathcal N}_1 \behavPre {\mathcal N}_2
\end{align}
to be true if for every component ${\mathcal T}$ which can be composed
with both $ {\mathcal N}_1$ and ${\mathcal N}_2$, the external
observable behaviour of the composite networks ${\mathcal N}_1
\interleave {\mathcal T}$ and ${\mathcal N}_2 \interleave {\mathcal
  T}$ are related in some appropriate way. ${\mathcal T}$
is a \emph{testing network} which is probing ${\mathcal N}_1$ and
${\mathcal N}_2$ for behavioural differences, along the lines used informally in 
Example~\ref{ex:wellformed}.  Intuitively this should
be \emph{black-box testing}, in which the tester, namely $\mathcal T$,
should have no access to the internal stations of the networks being
tested, namely ${\mathcal N}_1$ and ${\mathcal N}_2$. All it can do is
place code at their external interfaces,  to transmit values and
examine the subsequent effects, as seen again at the interfaces.  

\begin{defi}[Composing networks]\label{def:comp.nets}
For any two networks $\calM = \Gamma_M \with M$ and $\calP = \Gamma_P
\with P$
let $\calM \testP \calP$ be given by:
\begin{align*}
    &(\Gamma_M \with M ) \testP (\Gamma_P \with P) = 
    \begin{cases}
        (\Gamma_M \cup \Gamma_P) \with (M \Cpar P), & \text{if}\;  \nodes{\calM} \cap (\calP)_V =
  \emptyset\\
        \qquad\qquad\text{undefined},                          &\text{otherwise}
    \end{cases} 
\end{align*}
The composed connectivity graph $\Gamma_M \cup \Gamma_P$ 
is defined by letting $(\Gamma_M \cup \Gamma_P)_V = (\Gamma_M)_V 
\cup (\Gamma_P)_V$ and  $(\Gamma_M \cup \Gamma_N)_E = (\Gamma_M)_E \cup (\Gamma_P)_E$.
\qed 
\end{defi}
\noindent
The intuition here is that the composed network $\calM \testP \calP$
is constructed by \emph{extending} the network under test, $\calM$,
allowing code to be placed at its interface, and allowing completely
fresh stations to be added. These fresh stations can be used by the
tester to compute the results of probes made on $\calM$. 

\begin{prop}
\label{prop:testp.assoc}
Suppose $\calM, \calN, \calP  \in \nets$. 
Then 
\begin{enumerate}
\item $\calM \testP \calP  \in
  \nets$, whenever it is defined 

\item $(\calM \testP \calN)  \testP \calP  = 
 \calM \testP (\calN  \testP \calP )$, whenever both are defined 
\end{enumerate}
\end{prop} 
\begin{proof}
The two statements are proved separately. The proofs are given 
in a  separate appendix, Appendix \ref{sec:operator.results}; 
see pages \pageref{proof:testP.closed} and \pageref{proof:testP.assoc}.
\end{proof}

\begin{figure}[t]
\begin{align*}
\begin{tikzpicture}
            \node[state](m){$m$};
           \node[state](o1)[above right=of m]{$o_1$};
           \node[state](o2)[below right=of m]{$o_2$}; 
 \path[to]
       (m) edge [thick] (o1)
       (m) edge[thick] (o2);
   \begin{pgfonlayer}{background}
    \node [background,fit=(m)] {};
    \end{pgfonlayer}
  \end{tikzpicture}
  &&
\begin{tikzpicture}
            \node[state](o1){$o_1$};
            \node[state](e)[below right = of o1]{$e$};
           \node[state](o2)[below left=of e]{$o_2$}; 
 \path[tofrom]
       (o1) edge [thick] (o2);
\path[to]
       (o1) edge[thick] (e)
       (o2) edge[thick] (e); 
   \begin{pgfonlayer}{background}
    \node [background,fit=(o1) (o2) ] {};
    \end{pgfonlayer}
  \end{tikzpicture} 
  &&
  \begin{tikzpicture}
            \node[state](m){$m$};
           \node[state](o1)[above right=of m]{$o_1$};
           \node[state](o2)[below right=of m]{$o_2$}; 
           \node[state](e)[below right=of o1]{$e$}; 
 \path[tofrom]
       (o1) edge [thick] (o2);
\path[to]
       (m) edge[thick] (o1)
       (m) edge[thick] (o2)
       (o1) edge[thick] (e)
       (o2) edge[thick] (e); 
   \begin{pgfonlayer}{background}
    \node [background,fit=(m) (o1) (o2)] {};
    \end{pgfonlayer}
  \end{tikzpicture}
  \\
  \calM
  &&
  \calN
  &&
  (\calM \testP \calN)
  \end{align*}
  \caption{Network composition via $\testP$\label{fig:notsymmetric}}
\end{figure}
There is an inherent asymmetry in the definition of our
composition operator; in $\calM \testP \calP$ we allow $\calP$ to
place code at the interface nodes of $\calM$ but not the converse. As
a result the operator is not in general symmetric, as can be seen from
Example \ref{ex:notsymmetric}. 
\begin{exa}[$\testP$ is asymmetric]
\label{ex:notsymmetric}
Let $\calM, \calN$ be the networks depicted in Figure~\ref{fig:notsymmetric}. 
Here the network $\calM \testP \calN$ is well-defined and 
depicted on the right of Figure \ref{fig:notsymmetric}. Intuitively, this is 
the network obtained by extending $\calM$ with the information provided 
by $\calN$; these include the code running at nodes $o_1, o_2$, a connection 
between such nodes, and a connection from $o_1, o_2$, respectively, to a 
fresh node $e$, whose running code is left unspecified. 
The network $\calM \testP \calN$ is well-defined because 
none of the connections specified in $\calN$ involve the node $m \in \nodes{\calM}$; that is, when extending $\calM$ with $\calN$ it is ensured 
that the latter can interact with the node $m$ only via the nodes $o_1,o_2$. 

On the other hand, the composition $\calN \testP \calM$ is not defined. Intuitively,  
in $\calN$ nodes $o_1$ and $o_2$ can interact with the external environment only 
via the external node $e$. This is in contrast 
with  the definition of $\calM$, where node $m$ can be used to broadcast messages 
to such nodes.
\end{exa}

One might wonder if a symmetric composition operator could be used in place of our 
$\testP$; this point is discussed at length in Section~\ref{sec:justify}.
Nevertheless $\testP$ is a natural operator, and the next result shows that it can be
used to construct all non-degenerate networks starting from single nodes. 
Let $\mathbb{G}$ be the collection of networks which contain exactly one occupied node, 
that is $\mathbb{G} = \{(\Gamma \with G) \;|\; G \equiv \Cloc{s}{m} \mbox{ for some } s, m\}$. 

\begin{prop}\label{prop:network.decomp}
  Suppose ${\mathcal M}\in \nets$ is a network such that 
  $\nodes{\calM}$ is not empty.  Then ${\mathcal M} \equiv {\mathcal N} \testP
    {\mathcal G}$ for some$\calN \in \nets$ and ${\mathcal G} \in \mathbb{G}$.
\end{prop}
%
\begin{proof}
See Appendix \ref{sec:operator.results}, Page \pageref{proof:generators}.
\end{proof}
\noindent 
We use the term \emph{generating networks} to refer to elements of $\mathbb{G}$. 
The import of Proposition \ref{prop:network.decomp} is that all 
non-trivial well-formed networks can be constructed from basic generating 
networks, using $\testP$. 

\leaveout{
The last property we prove for the extension operator $\testP$ is that 
it is deeply connected to the symmetric operator $\comp$ defined in Example \ref{ex:mas.op}. 
In fact, it turns out that every network of the form $\calM \testP \calN$ 
can be rewritten as $\calM \comp \calN'$, where $\calN'$ is determined 
completely by the topological structure of $\calM$ and $\calN$. Conversely, if 
$\calM \comp \calN$ is defined, we can rewrite it as $\calM \testP \calN''$ 
for some network $\calN''$.

\begin{defi} 
Let $\calM, \calN$ be two networks. 
\begin{enumerate}[label=(\roman*)]
\item The \emph{symmetric counterpart} of $\calN$ 
with respect to $\calM$,   
is the network $\symm{\calM}{\calN} = \calN'$ such that 
\begin{eqnarray*}
(\calN')_V &=& (\calN)_V \cup \{m : \calM \vdash \someconn{m}{n} \mbox{ for some } n \in \nodes{\calN}\}\\
(\calN')_E &=& (\calN)_E \cup \{(m,n)\;:\;n \in \nodes{\calN}, \calM \vdash \rconn{m}{n}\} \cup\\
&\cup& \{(m,n)\;:\;m \in \nodes{\calN}, \calM \vdash \rconn{m}{n}\}
\end{eqnarray*}
\item The \emph{extension counterpart} of $\calN$ with respect to 
$\calM$ is the network $\ext{\calM}{\calN} = \calN''$ such that 
\begin{eqnarray*}
(\calN'')_V &=& (\calN)_V \setminus \nodes{\calM}\\
(\calN'')_E &=& \{(m,n) \;:\; (\calN'') \vdash m,n \;,\; \calN \vdash \rconn{m}{n}\}
\end{eqnarray*}
\end{enumerate}
\qed
\end{defi}

\begin{prop}
\label{prop:change}
Let $\calM, \calN$ be two well-formed networks. 
\begin{enumerate}
\item if $\nodes{\calM} \cap \nodes{\calN} = \emptyset$, then both 
$\symm{\calM}{\calN}$ and $\ext{\calM}{\calN}$ are defined and well-formed,
\item if $\mathscr{P}_e(\calM, \calN)$ then $\ext{\calM}{\symm{\calM}{\calN}} = \calN$, 
\item if $\mathscr{P}_s(\calM, \calN)$ then $\symm{\calM}{\ext{\calM}{\calN}} = \calN$, 
\item $(\calM \testP \calN) = (\calM \comp \symm{\calM}{\calN})$,
\item $(\calM \comp \calN) = (\calM \testP \ext{\calM}{\calN})$.
\end{enumerate}
\end{prop}

\begin{proof}
See \cite{phdthesis}, propositions \textbf{4.2.3}, \textbf{4.2.4} and \textbf{4.2.5}.
\end{proof}

\begin{figure}
\begin{align*}
     \begin{tikzpicture}
          \node[state](m){$m$}; 
          \node[state](n)[above right=of m]{$n$};
          \node[state](l)[below right=of m]{$l$};
          \node[state](o)[above=of m]{$o$};
          \node[state](i)[below=of m]{$i$}; 
 \path[to]
       (m) edge [thick] (o)
       (i) edge[thick] (m);    
 \path[tofrom]
       (m) edge [thick] (n)
       (m) edge [thick] (l);
    \begin{pgfonlayer}{background}
    \node [background,fit=(m)] {};
    \end{pgfonlayer}
    \end{tikzpicture}
&&  
     \begin{tikzpicture}
          \node[state](m){$m$};
          \node[state](n)[above right=of m]{$n$}; 
          \node[state](e)[below right= of n]{$e$};
          \node[state](l)[below left=of e]{$l$};
 \path[tofrom]
       (m) edge [thick] (n)
       (m) edge [thick] (l)
       (n) edge [thick] (e)
       (l) edge [thick] (e);
    \begin{pgfonlayer}{background}
    \node [background,fit=(n) (l)] {};
    \end{pgfonlayer}
    \end{tikzpicture}
&&
     \begin{tikzpicture}
          \node[state](n){$n$};
          \node[state](e)[below right =of m]{$e$}; 
          \node[state](l)[below left= of e]{$l$};
 \path[tofrom]
       (n) edge [thick] (e)
       (l) edge [thick] (e);
    \begin{pgfonlayer}{background}
    \node [background,fit=(n) (l)] {};
    \end{pgfonlayer}
    \end{tikzpicture}
\\
\calM = \Gamma_M \with M 
&&
\calN = \Gamma_N \with N
&&
\calN' = \Gamma'_N \with N
\end{align*}

\caption{Three networks $\calM, \calN$ and $\calN'$, such that 
$\calN = \symm{\calM}{\calN'}$, $\calN' = \ext{\calM}{\calN}$.}
\label{fig:change}
\end{figure}

\begin{exa}[Change of connectivity graphs]
Consider the networks $\calM, \calN$ and $\calN'$ depicted in 
Figure \ref{fig:change}. 
Here it is easy to show that $\calN' = \ext{\calM}{\calN}$; conversely, 
$\calN = \symm{\calM}{\calN'}$. In practice, $\symm{\calM}{\calN'}$ is 
defined by adding to $\calN$ the connections between internal nodes 
of $\calM$ and internal nodes of $\calN$ which are defined in the former network. 
On the other hand, $\ext{\calM}{\calN}$ is defined by removing in 
$\calN$ all the nodes that are internal in $\calM$, together with the 
associated connections. 

This ensures that $(\calM \comp \calN)$ and $(\calM \testP \calN')$ are 
defined; further, by Proposition \ref{prop:change} 
they are also equivalent.
\end{exa}

The technical result below will be useful in the following.
\begin{lem}
\label{lem:input.struct}
Let $\calM$, $\calN$, $\calG$ be respectively two networks and 
a generating network such that $\calM \comp \calG$, $\calN \comp \calG$ 
are defined. Suppose also that $\outp{\calM} = \outp{\calN}$.
If $\inp{\symm{\calM}{\calG}} = \emptyset$ then $\inp{\symm{\calN}{\calG}} = \emptyset$.
\end{lem}

\begin{proof}
Let $\calM = \Gamma_M \with M$, $\calG \equiv \Gamma_G \with \Cloc{s}{n}$; by straightforward calculations 
we can check that $\inp{\symm{\calM}{\calG}} = \emptyset$ 
if and only if $n \notin \outp{\calM}$. 
Then if $\symm{\calM}{\calG} = \emptyset$ it follows that $n \notin \outp{\calM}$, 
which by hypothesis gives that $n \notin \outp{\calN}$. But this is equivalent 
to $\inp{\symm{\calN}{\calG}} = \emptyset$.
\end{proof}


}
\subsection{Testing Structures}
\label{sec:ts}

The introduction of the composition operator $\testP$ allows the development of 
a behavioural theory based on a probabilistic 
generalisation of  the  \textit{De Nicola} and \textit{Hennessy}  testing preorders \cite{dnh}. In order to develop 
such a framework, we will exploit the mathematical tools introduced in Section \ref{sec:background}; our aim is 
 to be able to relate networks with different connection topologies.

In our framework, as has already been indicated, testing can be summarised as follows: the network to be tested  is 
composed with another one, usually called a \emph{testing network}. 
The composition of these two networks is then isolated from the external environment, 
in the sense that no external agent (in our case nodes in the interface of the composed network) 
can interfere with its behaviour; we will shortly present how such a task can be 
accomplished. The composition of the two networks isolated from the external 
environment takes the name of \emph{experiment}.

Once these two operations (composition with a test and isolation from the external environment) 
have been performed, the behaviour of the resulting experiment is analysed to check whether 
there exists a computation that yields a state which is successful. This task can be 
accomplished by relying on testing structures, which will be presented shortly.

At an informal level, successful states in our language coincide with those associated with networks 
where at least the code running at one node has the special \emph{action} $\omega$ enabled. 
Since networks have probabilistic behaviour, each computation will be associated with 
the probability  of reaching a successful state; thus, every experiment will be 
associated with a set of success probabilities, by quantifying over all its computation.

Let us now look at how the procedure explained above can be formalised; the topic 
of composing networks has already been addressed in detail in Section \ref{sec:comp.net}, 
in which we defined the operator $\testP$ and proved basic properties for it. To 
model experiments and their behaviour, we rely on the following mathematical structure.

\begin{defi}\label{def:testing}
  A \emph{Testing Structure} (TS) is a triple 
$\langle S, \red, \omega \rangle$ where 
\begin{enumerate}[label=(\roman*)]
\item $S$ is a set of states,

\item the relation $ \red$ is a subset of $S \times \dist{S}$,

\item $\omega$ is a success predicate over $S$, that is $\omega: S
  \rightarrow \sset{\ttrue,\ffalse}$. \qed
\end{enumerate}
\end{defi}
\noindent
Testing structures can be seen as (degenerate) pLTSs where the only
possible action corresponds to the internal activity $\tau$, and the
transition $\ar{\tau}$ is defined to coincide with the reduction
relation $\red$. 
 Conversely every pLTS automatically determines a testing structure,
by concentrating on the relation $\ar{\tau}$.

Our goal is to turn a network into a testing structure. This amounts to defining, 
for networks, a reduction relation and the success predicate $\omega$.
As we have mentioned in the beginning of this section, when converting a 
network into a testing structure, we want to  make it 
isolated from the external environment. 

When considering simpler process languages, like \emph{CCS} or \emph{CSP} 
(and, more generally, their probabilistic counterparts), 
processes are converted into testing structures by identifying the reduction 
relation with the internal activity $\tau$; that is, processes are 
not allowed to synchronise with some external agent via a visible action. 

Networks, however, are more complicated objects; here the nature of 
broadcast is non-blocking, meaning that a broadcast can be 
fired by a node in a network without requiring any synchronisation 
with a (possibly external) node. Thus we expect broadcast 
actions to induce reductions when converting a network into a testing structure.

On the other hand, input actions always originate from non-internal 
nodes. Hence they can be seen as external activities which can influence 
 the behaviour of a 
network. Therefore, input actions should not be included in the definition 
of the reduction relation for networks. 

Finally, the success predicate $\omega$ is defined to be true for exactly 
those networks in which the success clause $\omega$ is enabled in 
at least one node.

\begin{exa}\label{ex:ts}
The main example of a TS is given by
\begin{align*}
  \langle \nets, \red, \omega \rangle
\end{align*}
where 
\begin{enumerate}[label=(\roman*)]
\item $(\Gamma \with M) \red (\Gamma \with \Delta)$ whenever
  \begin{enumerate}[label=(\alph*)]
  \item $\Gamma \with M \ar{m.\tau} \Delta$ for some $m \in \nodes{M}$

  \item or, $\Gamma \with M \ar{m.c!v} \Delta$ for some value $v$, node name $m$ and channel $c$
  \end{enumerate}
  
\item ~
\begin{align*}
  \omega(M) = 
  \begin{cases}
    \ttrue, &  \text{if } M \equiv M' \Cpar \Cloc{\omega + s}{n}, \;\;\text{for some } s,n, M'\\
    \ffalse, & \text{otherwise}
  \end{cases}
\end{align*}
\end{enumerate}
If $\omega(M) = \ttrue$ for some system term $M$, we say that a network $\Gamma \with M$, where 
$\Gamma$ is an arbitrary connectivity graph, is $\omega$-successful, or simply successful. 
Note that when recording an $\omega$-success we do not take into account the node involved.
  \end{exa}
\noindent
As TSs can be seen as pLTSs, we can use in an arbitrary TS the various constructions
introduced in Section~\ref{sec:background}. Thus the reduction relation 
 $\red$ can be lifted to 
$\subdist{\sys} \times \subdist{\sys}$ and 
we can make use of the concepts of 
hyper-derivatives and extreme-derivatives, introduced in Section \ref{sec:background}, to 
model fragments of executions and maximal executions of a testing structure, respectively.
Hyper-derivations in testing structures are denoted with the symbol $\Red$, while 
we use the symbol $\RedE$ for extreme derivations.

Below we provide two simple examples that show how to reason about the behaviour of the testing 
structures presented in Example \ref{ex:ts}.

\begin{exa}
\label{ex:hypder}
Consider the testing structure associated with the network $\calN$ in
the center of Figure \ref{fig:counterex}, where the code $q$ is given
by the definition $q \Leftarrow q \probc{0.5} c!\pc{v}.\Cnil$. We can
show that, in the long run, this network will broadcast message $v$ to
the external location $o$ by exhibiting a hyper derivation for it
which terminates in the point distribution $\pdist{\Gamma_N \with
  \Cloc{c!\pc{v}.\Cnil}{k}}$.  If we let $\calN_1$ denote the configuration $ \Gamma_N \with
\Cloc{c!\pc{v}.\Cnil}{k}$,  we have the following 
hyper-derivation:
$$
\begin{array}{lclcl}
  \pdist{\calN} &=& \frac{1}{2} \;\cdot\; \pdist{\calN} &+ &
  \frac{1}{2} \;\cdot\;\pdist{\mathcal{N}_1}\\
  \frac{1}{2} \;\cdot\;\pdist{\calN}  &\red& \frac{1}{2^2}\;\cdot\; \pdist{\calN} &+ &
  \frac{1}{2^2}\;\cdot\; \pdist{\mathcal{N}_1} \\
&\vdots& \\
  \frac{1}{2^n}\;\cdot\; \pdist{\calN}  &\red& \frac{1}{2^{n+1}}\;\cdot\; \pdist{\calN} &+ &
  \frac{1}{2^{n+1}}\;\cdot\; \pdist{\calN_1} \\
&\vdots    &   \\
\end{array}
$$
\noindent
Let $\Delta' = \sum_{n=1}^\infty \ \frac{1}{2^{n}}\;\cdot\; \pdist{\calN_1}$. 
It is straightforward to check  that
\begin{math}
\Delta' =
\pdist{\calN_1}
\end{math}
and therefore we have the hyper-derivation $\calN \Red{\;\;\;} \pdist{\calN_1} $. 
\end{exa}

An arbitrary network $\calN$ can be tested by another (testing) network $\calT$ provided 
$\calN \testP \calT$ is well-defined.
Executions of the resulting testing structure will then be checked to establish whether 
the network $\calM$ satisfies a property the test was designed for; in such a case, the testing component 
of an experiment will reach a $\omega$-successful state. 

Executions, or maximal computations, correspond to
extreme derivatives in the testing structure associated with $(\calN \testP \calT)$, 
as defined in Section~\ref{sec:background}. 
Since the framework is probabilistic, each execution (that is  extreme derivative)  will be associated with 
a probability value, representing the probability that it will lead to an $\omega$-successful state.
Since the framework is also nondeterministic the possible results of this test application is given by a
non-empty set of probability values. 
\begin{defi}[Tabulating results]\label{def:resultsets}
 The \emph{value} of a sub-distribution in a TS is given by the function   
$\ValN : \subdist{S} \rightarrow [0,1]$,  defined by
    $\Val{\Delta} = \sum\setof{ \Delta(s)}{ \omega(s) = \ttrue}$.
Then the set of possible results from a sub-distribution $\Delta$  is defined by  
$\Results{\Delta} = \setof{\Val{\Delta'}}{ \Delta \RedE{} \Delta'}$. 
\qed
\end{defi}

  
\begin{exa}

Let $\calN$ be the network from Example \ref{ex:hypder} and 
consider the testing network $\calT$ given in 
Figure \ref{fig:counterex}, where the code is determined by
$t \Leftarrow c?\pa{x}.\omega.\Cnil$. 
It is easy to check that $\calN \testP \calT$ is well-defined and is equal to 
$\Gamma \with \Cloc{q}{k} \Cpar \Cloc{t}{o}$, where $\Gamma$ is the connectivity 
graph containing the three nodes $k, o,l$ and having the connections from $\Gamma 
\vdash \rconn{k}{o}$ and $\Gamma \vdash \lconn{o}{t}$.
So consider the testing structure associated with it; recall that we have the definition 
$q \Leftarrow q \probc{0.5} c!\pc{v}.\Cnil$. 
For convenience  let $\calN_1 = \Gamma_N \with \Cloc{c!\pc{v}.\Cnil}{k}$ as in the previous example, 
$\calN_2 = \Gamma_N \with \Cloc{\Cnil}{k}$ and $\calT_{\omega} = \Gamma_T \with \Cloc{\omega.\Cnil}{o}$.
Then we have the following  hyper-derivation for $\calN \testP \calT$:
$$
  \begin{array}{lclcl}
    \pdist{\calN \testP \calT} &\red& (\frac{1}{2}\cdot \pdist{\calN \testP \calT} + \frac{1}{2} \cdot \pdist{\calN_1 \testP \calT})
    &+& \varepsilon\\
     \frac{1}{2}\cdot \pdist{\calN \testP \calT} + \frac{1}{2} \cdot \pdist{\calN_1 \testP \calT} &\red &
    (\frac{1}{2^2}\cdot \pdist{\calN \testP \calT} + \frac{1}{2^2} \cdot \pdist{\calN_1 \testP \calT}) &+&
    \frac{1}{2} \cdot \pdist{\calN_2 \testP \calT_{\omega}}\\
     \vdots&    &\vdots& &\vdots\\
    \frac{1}{2^n}\cdot \pdist{\calN \testP \calT} + \frac{1}{2^n} \cdot \pdist{\calN_1 \testP \calT} & \red&
    (\frac{1}{2^{n+1}}\cdot \pdist{\calN \testP \calT} + \frac{1}{2^{n+1}} \cdot \pdist{\calN_1 \testP \calT}) & + &
    \frac{1}{2^n} \cdot \pdist{\calN_2 \testP \calT_{\omega}}\\
     \vdots& &\vdots&  &\vdots    
  \end{array}
$$
  were we recall that $\varepsilon$ denotes the empty sub-distribution, that is the one with $\support{\varepsilon} = \emptyset$.
  We have therefore the hyper-derivation 
  \begin{equation*}
    \pdist{\calN \testP \calT} \Red{} \varepsilon + 
    \sum_{n = 1}^{\infty} \frac{1}{2^n} \pdist{\calN_2 \testP \calT_{\omega}}
    = \pdist{\calN_2 \testP \calT_\omega}.
  \end{equation*}
  Further, the above hyper-derivation satisfies the constraints required by $\RedE{}$, defined in Section~\ref{sec:background}, and 
  therefore we have the extreme derivative $\pdist{\calN \testP \calT} \RedE{} \pdist{\calN_2 \testP \calT_{\omega}}$. 
  Since $\Val{ \pdist{\calN_2 \testP \calT_{\omega}}} = 1$    we can therefore deduce that $1 \in \Results{\calN \testP \calT}$.
  \end{exa}

\subsection{The behavioural preorders}
\label{sec:maypreord}

We now combine the concepts of the previous two sections to obtain our behavioural preorders.
We have seen how to associate a non-empty set of probabilities, tabulating 
the possible outcomes from applying the test $\calT$ to the network $\calN$. 
As explained in \cite{DGHM09full}
there are two natural ways to compare such sets, optimistically or pessimistically. 

\begin{defi}[Relating sets of outcomes]
\label{def:relsets}
Let ${\mathcal O}_1$, ${\mathcal O}_2$ be two sets of values in 
$[0,1]$.
\begin{enumerate}[label=(\roman*)]
\item The \emph{Hoare's Preorder} is defined by letting 
${\mathcal O}_1 \sqsubseteq_{H} {\mathcal O}_2$ whenever 
for any $p_1 \in {\mathcal O}_1$ there exists $p_2 \in {\mathcal O}_2$ 
such that $p_1 \leq p_2$.
\item The \emph{Smith's Preorder} is defined by letting 
${\mathcal O}_1 \sqsubseteq_{S} {\mathcal O}_2$ if 
for any $p_2 \in {\mathcal O}_2$ there exists $p_1 \in {\mathcal O}_1$ 
such that $p_1 \leq p_2$.\qed
\end{enumerate}
\end{defi}
%
%
\noindent Given two networks $\calM, \calN$ we can relate their behaviour, 
when extended with a testing network $\calT$, by comparing the success 
outcomes of $\calM \testP \calT$ and $\calN \testP \calT$ (provided both these 
networks are defined) via Definition \ref{def:relsets}. 
We can go further and consider what is the relationship between such sets of 
outcomes with respect to all possible tests $\calT$ which can be used 
to extend the networks $\calM, \calN$.
 
\begin{defi}[Testing networks]
\label{def:maytest}
 For ${\mathcal M}_1,\,{\mathcal M}_2 \in \nets$ 
  such that
  $\inp{{\mathcal M}_1} = \inp{{\mathcal M}_2}$, 
  $\outp{\calM_1} = \outp{\calM_2}$, we write 
  \begin{enumerate}
  
  \item $\calM_1 \Mayleq \calM_2$ iff 
for every (testing) network $\mathcal T \in \nets$ such that both 
        ${\mathcal M}_1 \testP {\mathcal T}$ and  ${\mathcal M}_2 \testP {\mathcal T}$ are defined,
        \begin{math}
           \Results{{\mathcal M}_1 \testP {\mathcal T}}     \sqsubseteq_{H} \Results{{\mathcal M}_2 \testP {\mathcal T} }.   
        \end{math} 
  \item $\calM_1 \Mustleq \calM_2$ iff for every (testing network) $\calT \in \nets$ such that 
  both $\calM_1 \testP \calT$ and $\calM_2 \testP \calT$ are defined, 
  \begin{math}
  \Results{\calM_1 \testP \calT} \sqsubseteq_{S} \Results{\calM_2 \testP \calT}
  \end{math}
  \end{enumerate}
We use ${\mathcal M}_1 \Mayeq {\mathcal M}_2$ as an abbreviation for ${\mathcal M}_1 \Mayleq {\mathcal M}_2$  and 
${\mathcal M}_2 \Mayleq {\mathcal M}_1$. The relation $\Musteq$ is defined similarly. 
Finally, we say that $\calM_1 \sqsubseteq \calM_2$ iff both $\calM_1 \Mayleq \calM_2$ and 
$\calM_1 \Mustleq \calM_2$ hold, and $\calM_1 \simeq \calM_2$ iff $\calM_1 \sqsubseteq \calM_2$ 
and $\calM_2 \sqsubseteq \calM_1$.
~\qed
\end{defi}

Some explanation is necessary for the requirement on the interface of networks we have 
placed in Definition \ref{def:maytest}. 
This constraint establishes that two networks $\calM$ and $\calN$ are always distinguished 
if the sets of their input or output nodes differ. As we already mentioned, external nodes 
can be seen as terminals that can be accessed by the external environment to interact with 
the network. Roughly speaking, the constraint we have placed corresponds to the intuition 
that the external environment can distinguish two networks 
$\calM$ and $\calN$ by simply looking at the terminals that it can use to interact with these 
two networks. 


\begin{figure}
\begin{tikzpicture}
            \node[state](o1){$o_1$};
           \node[state](o2)[below=of o1]{$o_2$};
           \node[state](e)[left  =of o2]{$e$}; 
           \node[state](o)[right =of o1]{$o$}; 
 \path[tofrom]
       (o1) edge [thick] (o2)
       (o1) edge[thick] (o)
       (o2) edge[thick] (o); 
   \begin{pgfonlayer}{background}
    \node [background,fit=(o1) (o2) (o) (e)] {};
    \end{pgfonlayer}
  \end{tikzpicture}

\caption{ A test} 
\label{fig:testex}
\end{figure}

\begin{exa}\label{ex:testing} 
  Consider the testing network 
\begin{align*}
  &{\mathcal T} = \Gamma_T \with \Cloc{c!\pc{0}}{e} \Cpar \Cloc{d?\pa{x}.c!\pc{x}}{o_1} 
                       \Cpar \Cloc{d?\pa{x}.c!\pc{x}}{o_2} 
                       \Cpar \Cloc{c?\pa{x}.c?\pa{y}.\omega}{o} 
\end{align*}
where the connectivity is described in Figure~\ref{fig:testex}. This can be used to test the networks 
$\calM, \calN$ from Example~\ref{ex:ex1} in the testing structure of Example~\ref{ex:ts}. 
 
Intuitively the test sends the value $0$ along the channel $c$ at the node $e$, awaits for results along the channel
$d$ at the nodes $o_1$ and $o_2$. These results are processed at node $o$, where success might be announced. 

The combined network
$(\calM \testP {\mathcal T})$ is deterministic in this TS, although probabilistic, and so has only one extreme
derivative; $\Results{ \calM \testP {\mathcal T}} = \sset{0.8}$.
 
A similar calculation shows that  $\Results{ \calN \testP {\mathcal T}} = \sset{0.81}$;  it 
therefore follows that $\calN \not \Mayleq \calM$ and $\calM \not\Mustleq \calN$. 

Consider now the networks $\calM_1, \calM_2$ from Example \ref{ex:ex1}. 
Here $(\calM_1 \testP {\mathcal T})$ is both probabilistic and nondeterministic, and
 $\Results{ \calM_1 \testP {\mathcal T}} = \setof{p}{ 0.5 \leq p \leq 1}$. Moreover 
 we have 
\begin{math}
   \Results{ \calM \testP {\mathcal T}} \sqsubseteq_{H}  \Results{\calM_1 \testP {\mathcal T}}. 
\end{math}

The combined network $(\calM_2 \testP {\mathcal T})$ is also deterministic, although it has 
\emph{limiting behaviour}; $\Results{ \calM_2 \testP {\mathcal T}} = \sset{1}$. Thus, in this case we 
have both $\Results{ \calM \testP {\mathcal T}} \sqsubseteq_{H} \Results{ \calM_2 \testP {\mathcal T}}$ and 
$\Results{ \calM_1 \testP {\mathcal T}} \sqsubseteq_{H} \Results{ \calM_2 \testP {\mathcal T}}$. 
Further, we have that 
$\Results{\calM_1 \testP \calT} \sqsubseteq_{S} \Results{\calM_2 \testP \calT}$, but 
$\Results{\calM_2 \testP \calT} \not\sqsubseteq_{S} \Results{\calM_1 \testP \calT}$.
 \end{exa}

\begin{figure}
\begin{align*}
     \begin{tikzpicture}
          \node[state](m){$m$}; 
          \node[state](n)[below=of m]{$n$};
          \node[state](o1)[right=of m]{$o_1$}; 
          \node[state](o2)[right=of n]{$o_2$};  
 \path[tofrom]
       (m) edge[thick] (o1)
       (n) edge[thick] (o2);
   \begin{pgfonlayer}{background}
    \node [background,fit=(m) (n)] {};
    \end{pgfonlayer}
    \end{tikzpicture}
&&
      \begin{tikzpicture}
          \node[state](m){$m$}; 
          \node[state](o1)[above right=of m]{$o_1$}; 
          \node[state](o2)[below right=of m]{$o_2$};
 \path[tofrom]
       (m) edge[thick] (o1)
       (m) edge[thick] (o2);
   \begin{pgfonlayer}{background}
    \node [background,fit=(m)] {};
    \end{pgfonlayer}
    \end{tikzpicture}
\\
\calM = \Gamma_M \with \Cloc{c!\pc{v}.\Cnil}{m} \Cpar \Cloc{c!\pc{v}.\Cnil}{n}
&&
\calN = \Gamma_N \with \Cloc{c!\pc{v}}{m}
\end{align*}

\caption{Broadcast vs Multicast}
\label{fig:bcast}
\end{figure}

 \begin{exa}[Broadcast vs Multicast]\label{ex:bcast1}
   Consider the networks $\calM$ and $\calN$ in Figure~\ref{fig:bcast}. 
Intuitively in $\calN$ the value $v$ is (simultaneously) broadcast to both nodes $o_1$
and $o_2$ while in $\calM$ there is a multicast. More specifically $o_1$ receives $v$ from
mode $m$ while in an independent broadcast $o_2$ receives it from $n$. 

This difference in behaviour can be detected (when we compare the networks optimistically) by the testing network
\begin{eqnarray*}
  \calT &=& \Gamma_T \with \Cloc{c?\pa{x}.c!\pc{w}.\Cnil}{o_1} \Cpar \Cloc{c?\pa{x}.c?\pa{y}. \Cmatch{y=0}{\Cnil}{\omega}}{o_2}
\end{eqnarray*}
assuming $v$ is different than $w$; here we assume $\Gamma_T$ is the simple network which connects $o_1$ to 
$o_2$.  
Both $\calM \testP \calT$ and $\calN \testP \calT$ are well-formed and note that
they are both non-probabilistic. 

Because $\calN$ simultaneously broadcasts to $o_1$ and $o_2$ the second value received by
$o_2$ is always $w$ and therefore the test never succeeds; $\Results{\calN \testP \calT} = \sset{0}$. 
On the other-hand there is a possibility for the test succeeding when applied to $\calM$, 
$ 1 \in \Results{\calM \testP \calT}$.  This is because in $\calM$ node $m$ might first transmit $v$ to 
$o_1$ after which $n$ transmits $w$ to $o_2$; now node $n$ might transmit the value $v$ to $o_2$ 
and assuming it is different than $w$ we reach a success state. It follows that 
$\calM \not\Mayleq \calN$. 

Note that we can slightly modify the test $\calT$ to show that we 
also have $\calN \not\Mustleq \calN$. To this end, let 
\[
\calT' = \Gamma_T \with \Cloc{c?\pa{x}.c!\pc{0}.\Cnil}{o_1} \Cpar \Cloc{c?\pa{x}.c?\pa{y}. \Cmatch{y=0}{\omega}{\Cnil}}{o_2}
\]
In this case we have that $\Results{\calM \testP \calT'} = \{0,1\}$, while 
$\Results{\calN \testP \calT'} = \{1\}$, and by definition 
$\Results{\calN \testP \calT'} \not\sqsubseteq_{S} \Results{\calM \testP \calT'}$.

One might also think it possible to use the difference between broadcast and multicast to design a test 
$\calT''$ for which 
$\Results{\calN \testP \calT''} \not\sqsubseteq_{H} \Results{\calM \testP \calT''}$ 
and $\Results{\calM \testP \calT''} \not\sqsubseteq_{S} \Results{\calN \testP \calT''}$. 

For example, if we let $\calT'' = \calT'$ we obtain that $1 \in \Results{\calN \testP \calT'}$. 
This is because because in $\calN \testP \calT'$ the second value received by $o_2$ is always $w$. 
However we also have that $1 \in \Results{\calM \testP \calT'}$, since the
simultaneous broadcast in $\calN$ can be simulated by a multicast in $\calM$, 
by node $m$ first broadcasting to $o_1$ followed by
$n$ broadcasting to $o_2$. As this line of reasoning is independent from the test $\calT'$, 
it also applies to all those networks that can be used to test the behaviour of $\calM$ and $\calN$; 
this leads to the intuition that $\calN \Mayleq \calM$, which will be proved formally later 
as a consequence of Example \ref{ex:bcast} and Theorem \ref{thm:may.sound}. 
Similarly, Theorem \ref{thm:must.sound} shows that $\calM \Mustleq \calN$.

\end{exa}

One pleasing property of the behavioural preorders is that they allow 
compositional reasoning over networks.
\begin{prop}[Compositionality]
\label{prop:compmay}
Let $\calM_1, \calM_2$ be two networks such that $\calM_1 \Mayleq \calM_2$ 
($\calM_1 \Mustleq \calM_2$), and 
let $\calN$ be another network such that both $(\calM_1 \testP \calN)$ and $(\calM_2 \testP \calN)$ 
are defined. Then $(\calM_1 \testP \calN) \Mayleq (\calM_2 \testP \calN$) 
($(\calM_1 \testP \calN) \Mustleq (\calM_2 \testP \calN)$).
\end{prop}
\begin{proof}
A direct consequence of $\testP$ being both associative and interface preserving.
\end{proof}
\noindent
We end this section with an 
application of this compositionality result.

\begin{figure}[t]
                                
\begin{align*}
     \begin{tikzpicture}
          \node[state](m){$m$}; 
          \node[state](n)[below=of m]{$n$}; 
          \node[state](k)[right=of m]{$k$}; 
           \node[state](o)[right=of k]{$o$}; 
 \path[to]
       (m) edge [thick] (n)
       (m) edge[thick] (k)
       (n) edge[thick] (k)
       (k) edge[thick] (o);
   \begin{pgfonlayer}{background}
    \node [background,fit=(m) (n) (k)] {};
    \end{pgfonlayer}
    \end{tikzpicture}
&&  
      \begin{tikzpicture}
          \node[state](m){$m$}; 
          \node[state](k)[right=of m]{$k$}; 
           \node[state](o)[right=of k]{$o$}; 
 \path[to]
       (m) edge[thick] (k)
       (k) edge[thick] (o);
   \begin{pgfonlayer}{background}
    \node [background,fit=(m) (k)] {};
    \end{pgfonlayer}
    \end{tikzpicture}
\\
\calM = \Gamma_M \with \Cloc{A_m}{m} \Cpar  \Cloc{A_n}{n} \Cpar \Cloc{A_k}{k} 
&&
\calN = \Gamma_N \with \Cloc{B_m}{m} \Cpar  \Cloc{A_k}{k}
\end{align*}
  \caption{Two networks with a common sub-network\label{fig:similar}}

\begin{align*}
     \begin{tikzpicture}
          \node[state](m){$m$}; 
          \node[state](n)[below=of m]{$n$}; 
          \node[state](k)[right=of m]{$k$};  
 \path[to]
       (m) edge [thick] (n)
       (m) edge[thick] (k)
       (n) edge[thick] (k);
   \begin{pgfonlayer}{background}
    \node [background,fit=(m) (n)] {};
    \end{pgfonlayer}
    \end{tikzpicture}
&&
&&
      \begin{tikzpicture}
          \node[state](m){$m$}; 
          \node[state](k)[right=of m]{$k$}; 
 \path[to]
       (m) edge[thick] (k);
   \begin{pgfonlayer}{background}
    \node [background,fit=(m)] {};
    \end{pgfonlayer}
    \end{tikzpicture}
&&
&&
     \begin{tikzpicture}
          \node[state](k){$k$}; 
           \node[state](o)[right=of k]{$o$}; 
 \path[to]
        (k) edge[thick] (o);
   \begin{pgfonlayer}{background}
    \node [background,fit=(k)] {};
    \end{pgfonlayer}
    \end{tikzpicture}
\\
\calM_1 = \Gamma_1 \with \Cloc{A_m}{m} \Cpar \Cloc{A_n}{n}
&&
&&
\calN_1 = \Gamma_2 \with \Cloc{B_m}{m}
&&
&&
\calK = \Gamma_K \with \Cloc{A_k}{k}
&&
\end{align*}
   \caption{Decomposition of the networks $\calM$ and $\calN$\label{fig:decomp}  }
\end{figure}
\begin{exa}
Consider the networks $\calM$ and $\calN$ in Figure \ref{fig:similar}, where the codes at the various
nodes are given by
\begin{eqnarray*}
  A_m &\Leftarrow& c!\pc{v}.\Cnil\\
  A_n &\Leftarrow&c?\pa{x}.d!\pc{w}.\Cnil\\
  A_k &\Leftarrow& c?\pa{x}.d?\pa{y}.e!\pc{u}.\Cnil\\
  B_m &\Leftarrow& c!\pc{v}.d!\pc{w}.\Cnil
\end{eqnarray*}
 It is 
possible to write both of them respectively as $\calM_1 \testP \calK$ and 
$\calN_1 \testP \calK$, where the networks $\calM_1, \calN_1$ and $\calK$ 
are depicted in Figure \ref{fig:decomp}. In order to prove that $\calM 
\Mayleq \calN$ ($\calN \Mustleq \calM$), it is therefore sufficient to focus on their 
respective sub-networks $\calM_1$ and $\calN_1$,  and prove
$\calM_1 \Mayleq \calN_1$ ($\calN_1 \Mustleq \calM_1$). The equivalence of the two 
original networks will  then follow from a direct application 
of Proposition \ref{prop:compmay}.
\end{exa}

\subsection{Justifying the operator $\testP$}
\label{sec:justify}

Here we revisit Definition~\ref{def:comp.nets} and in particular investigate the possibility
of using alternative composition operators. The remainder of the paper is independent of this
section and so it may be safely skipped by the uninterested reader. 

Here we take a more general approach to composition; rather than give a particular operator we discuss
natural properties we would expect of such operators. Let us just presuppose a \emph{consistency predicate}
$\mathscr{P}$ on pairs of networks determining when their composition should be defined. The only requirement
on $\calM  \,\mathscr{P}\, \calN$ is that whenever it is  defined the resulting composite network is well-formed. 
Since the composite network should be  determined by that of its components this amounts to requiring that 
\begin{align}\label{eq:iscons}
  & \mathscr{P}(\calM, \calN) \;\;\text{implies}\;\; \nodes{\calM} \cap \nodes{\calN} = \emptyset 
   \qquad\qquad \text{for any $\calM, \calN \in \nets$}
\end{align}

\noindent Given a consistency predicate satisfying (\ref{eq:iscons}) we can now generalise Definition~\ref{def:comp.nets} to 
give a range of different composition operators. 
\begin{defi}[General composition of  networks]\label{def:gen.comp.nets}
  Let $\mathscr{P}$ be a consistency  predicate on networks in \nets satisfying 
  (\ref{eq:iscons}). Then we define the associated partial composition
  relation by:
\begin{align*}
    &(\Gamma_M \with M ) \interleave_{\mathscr{P}} (\Gamma_N \with N) = 
    \begin{cases}
        (\Gamma_M \cup \Gamma_N ) \with (M \Cpar N), & \text{if}\;  \mathscr{P}((\Gamma_M \with M ), (\Gamma_N \with N)) \\
        \qquad\qquad\text{undefined},                          &\text{otherwise}
    \end{cases} 
\end{align*}
The connectivity graph $\Gamma_M \cup \Gamma_N$ is defined as in Definition~\ref{def:comp.nets}.
\qed 
\end{defi}

\begin{exa}\label{ex:mas.op}
  Let $\mathscr{P}_s$ be the partialss binary predicate defined by letting 
$\mathscr{P}_s(\calM, \calN)$ whenever
\begin{itemise}
\item $\nodes{\calM} \cap \nodes{\calN} = \emptyset$
\item   $\calM \vdash \rconn{m}{n}$ if and only if 
 $\calN \vdash \rconn{m}{n}$, for every $m \in \nodes{\calM}$ and $n \in \nodes{\calN}$, 
 \item $\calM \vdash \lconn{m}{n}$ if and only if 
 $\calN \vdash \lconn{m}{n}$, for every $m \in \nodes{\calM}$ and $n \in \nodes{\calN}$,
\end{itemise}

\noindent By definition this satisfies the requirement (\ref{eq:iscons}) above,
and intuitively it only allows the composition whenever the two
individual networks agree on the interconnections between internal and external nodes.

For notational convenience we denote the operator $\interleave_{\mathscr{P}_s}$ with 
$\comp$. It is easy to check that this operator 
is both associative and commutative.
\end{exa}

\begin{figure}[t]                                 
\begin{align*}
     \begin{tikzpicture}
          \node[state](l){$l$}; 
          \node[state](n)[right=of l]{$o$}; 
 \path[to]
       (l) edge [thick](n);
   \begin{pgfonlayer}{background}
    \node [background,fit=(l) ] {};
    \end{pgfonlayer}
    \end{tikzpicture}
&&  
&& 
     \begin{tikzpicture}
          \node[state](l){$k$}; 
          \node[state](n)[right=of l]{$o$}; 
 \path[to]
       (l) edge [thick](n);
   \begin{pgfonlayer}{background}
    \node [background,fit=(l) ] {};
    \end{pgfonlayer}
    \end{tikzpicture}  
&&
&&
     \begin{tikzpicture}
          \node[state](l){$o$}; 
          \node[state](n)[right=of l]{$l$}; 
 \path[to]
       (n) edge [thick](l);
   \begin{pgfonlayer}{background}
    \node [background,fit=(l) ] {};
    \end{pgfonlayer}
    \end{tikzpicture}  
\\
{\mathcal M} = \Gamma_M \with \Cloc{p}{l}
&&
&&
{\mathcal N} = \Gamma_N \with \Cloc{q}{k} 
&&
&&
{\mathcal T} = \Gamma_T \with \Cloc{t}{o}                                
\end{align*}

  \caption{A problem with the  composition operator}
\label{fig:counterex} 
 
\end{figure}

Thus a priori this composition operator $\comp$ could equally well 
be used to develop the testing theory in Section~\ref{sec:maypreord}. 
Unfortunately the resulting theory would be degenerate.
\begin{exa}[Example~\ref{ex:mas.op} continued]
  Consider the networks $\calM, \calN$ in Figure~\ref{fig:counterex}, 
  where we choose $p = c!\pc{v}$ and $q = c!\pc{w}$ for two different 
  values $v,w$.
  Then intuitively ${\mathcal M}$ and ${\mathcal N}$ 
  should have different observable behaviour, observable by placing 
  a test at the node $o$. However if the operator $\comp$ is used 
  to combine a test with the network being observed they are 
  indistinguishable. 

   This is because if there is a  network $\mathcal T$  such that
  ${\mathcal M} \comp {\mathcal T}$ and ${\mathcal N} \comp
  {\mathcal T}$ are well-defined then $o$ can not be in
  $\nodes{\mathcal T}$.  For if $o$ were in $\nodes{\mathcal T}$, then
  since $\calM \vdash \rconn{l}{o}$ the definition of the operator
  implies that $\calT \vdash \rconn{l}{o}$.  This in turn implies
  that $\calN \vdash \rconn{l}{o}$, which is not true.

  Now since no testing network which can be applied to both ${\mathcal
    M}$ and ${\mathcal N}$ can place any code at $o$, no difference
  can be discovered between them.
 \end{exa}

The question now naturally arises about which consistency predicate $\mathscr{P}$ 
lead to reasonable composition operators  $\interleave_{\mathscr{P}_s}$, in the sense 
that at least the resulting testing theories are not degenerate. 
We want to be able to compare networks with different connectivity graphs, and
possibly different nodes, such as $\calM$ and $\calN$ in Figure \ref{fig:counterex}. 
We also should not be able to change the
connectivity of the internal nodes of a network when we test it; we wish to implement
\emph{black-box} testing, where the nodes containing running code cannot 
be accessed directly.  

These informal requirements can be formulated as natural requirements on composition
operators. 
The
first says that the composed network is completely determined by the
components:
\begin{description}
\item[(I) Merge] the operator $\interleave$ should be determined by
  some predicate $\mathscr{P}$ using
  Definition~\ref{def:comp.nets}.
\end{description}

Intuitively the interface of a network is how their external behaviour
is to be observed. Since our aim is to enable 
compositional reasoning over networks, we would expect composition to preserve
interfaces:
\begin{description}
\item[(II) Interface preservation] If $\inp{{\mathcal M}} = \inp{{\mathcal
      N}}$, $\outp{\calM} = \outp{\calN}$ and ${\mathcal T}$ can be composed with both,  that is both 
       ${\mathcal M} \interleave {\mathcal T}$ and ${\mathcal N} \interleave {\mathcal T}$
are well-defined,  then $\inp{{\mathcal
      M} \interleave {\mathcal T}} = \inp{{\mathcal N} \interleave {\mathcal
      T}}$, 
      $\outp{{\mathcal M} \interleave {\mathcal T}} = \outp{{\mathcal N} \interleave {\mathcal
      T}}$.
\end{description}

The final requirement captures the intuitive idea that reorganising
the internal structure of a network should not affect the ability to
perform a test; in fact the reorganisation is simply a renaming of
nodes. Let $\sigma$ be a permutation of node names. We use $(\Gamma
\with M)\sigma$ to denote the result of applying $\sigma$ to the node
names in $M$ and in the connectivity graph $\Gamma$.
\begin{description}
\item[(III) Renaming] Suppose ${\mathcal M} \interleave {\mathcal T}$ is
  defined. Then ${\mathcal M}\sigma
  \interleave {\mathcal T}$ is also defined, provided $\sigma$ is a node permutation
  which satisfies
\begin{itemise}
  \item $\sigma(e)= e$ for every $e \in \type{{\mathcal M}}$
  \item no $n \in \nodes{{\mathcal T}}$ appears in the range of $\sigma$; that is
    $n \in \nodes{{\mathcal T}}$ implies $\sigma(n) = n$. 
\end{itemise}
\end{description}

\begin{exa}
The operator $\comp$ does not satisfy \textbf{(III)}, as can be
seen using the simple networks in Figure~\ref{fig:counterex}; ${\mathcal M} \comp {\mathcal T}$
is obviously well-defined. However, consider the renaming $\sigma$ which swaps node names $l$ to $k$, which is valid with respect to 
$\mathcal T$; the network ${\mathcal M\sigma} \comp {\mathcal T}$ is not defined, as $\Gamma_T \not\vdash \rconn{k}{o}$.
A slight modification will demonstrate that interfaces are also not preserved by this operator. 
\end{exa}

\begin{prop}\label{prop:I2III}
  Suppose $\interleave$ satisfies the conditions (I) - (III) above. Then 
  $\nodes{\mathcal M} \cap \type{{\mathcal N}} = \emptyset$
  whenever  ${\mathcal  M} \interleave {\mathcal N}$ is defined. 

\end{prop}
\begin{proof}
By contradiction; let ${\mathcal M} = \Gamma_{M} \with M$ and 
${\mathcal N} = \Gamma_{N} \with N$. Assume that $m$ is a node included in 
$\nodes{M} \cap \type{\calN}$, and that 
${\mathcal M} \interleave \mathcal{N}$ is defined. 
Finally, let $\sigma$ be an arbitrary permutation such that 
$\calM\sigma \interleave \calN$ is defined. 
Note that the following statements are true: 
\begin{enumerate}
\item $m \notin \nodes{\calN}$,
\item $\calN \vdash m$,
\item $\type{\calM} = \type{\calM \sigma}$.
\item $m \notin \type{\calM \interleave \calN}$.
\end{enumerate}
\noindent For proving the last statement just note that 
$m \in \nodes{\calM}$ by hypothesis, hence 
$m \in \nodes{\calM \interleave \calN}$. 
By definition of interface it follows that 
$m \notin \type{\calM \interleave \calN}$.

Let $l$ be a node which is not contained in $\Gamma_M$, nor in $\Gamma_N$. 
Consider the permutation $\sigma$ which swaps nodes $m$ and $l$; that is 
$\sigma(m) = l, \sigma(l) = m$ and $\sigma(k) = k$ for all $k\neq m,l$. 
Since $\calM \not\vdash l$ we also have that $l \notin \nodes{\calM}$, 
so that $\sigma(l)= m \notin \nodes{{\mathcal M}\sigma}$.
Further, the permutation $\sigma$ is consistent with condition (III), renaming, when 
applied to networks $\mathcal M$ and $\mathcal N$; 
therefore ${\calM}\sigma \interleave \calN$ is 
defined.\\ 

Since $m \notin \nodes{\calM\sigma}$, by (1) above it follows 
that $m \notin \nodes{{\mathcal M\sigma}\interleave {\mathcal N}}$;  
by (2), we obtain that $(\Gamma_M)\sigma \cup \Gamma_N \vdash m$. These two statements ensure 
that $m \in \type{{\mathcal M}\sigma \interleave {\mathcal{N}}}$.\\
As a direct consequence of (3) and condition (II), interface preservation, we 
also have that $m \in \type{{\mathcal{M}} \interleave {\mathcal{N}}}$, 
but this contradicts (4).
\end{proof}

\begin{cor}
Let $\interleave$ be any 
symmetric composition operator which satisfies the conditions (I) - (III). 
Suppose $ \calM_1 \interleave \calM_2$ is well-defined, and of the form 
$\Gamma \with M$. Then $\Gamma \vdash \notsomeconn{m_1}{m_2}$ whenever $m_i \in 
\nodes{\calM_i}$. 
\end{cor}
\begin{proof}
  A simple consequence of the previous result.
\end{proof}
\noindent
What this means is that if we use such a symmetric operator when applying a test to
a network, as in Definition~\ref{def:testing}, then the resulting testing preorder will
be degenerate; it will not distinguish between any pair of nets. 
In some sense this result is unsurprising. For $\mathcal T$ to test $\mathcal
M$ in ${\mathcal M}\interleave {\mathcal T}$ it must have code running at the
interface of $\mathcal M$. But, as we have seen, condition (III) more or less forbids $\mathcal
T$ to have code running at the interface of $\mathcal M$.

Thus we have ruled out the possibility of basing our testing theory on a symmetric 
composition operator. The question now remains what composition operator is the most appropriate? 
We have already stated that conditions (I)-(III) are natural, 
and since there are no further obvious requirements we
could choose  the operator with greatest expressive power among those that satisfy conditions (I)-(III). 
Here an operator $\interleave_{\mathscr{P}_1}$ is more expressive than another 
$\interleave_{\mathscr{P}_2}$ if, whenever $\calM \interleave_{\mathscr{P}_1} \calN$ is defined 
for any two arbitrary networks $\calM, \calN$, 
then so is $\calM \interleave_{\mathscr{P}_2} \calN$, and the result of the two compositions 
above is the same. 
The next Lemma shows that the operator which we are looking for is exactly $\testP$.
\begin{lem}
  Let $\mathscr{P}$ be any consistency predicate satisfying  (\ref{eq:iscons}) above. Then 
 if  $\calM  \interleave_{\mathscr{P}}  \calN $ is defined so is $\calM \testP \calN$ and 
 moreover $\calM  \interleave_{\mathscr{P}}  \calN \,=\,\calM \testP \calN$.  
\end{lem}
\begin{proof}
  Obvious from the definition of $\testP$ in Definition~\ref{def:comp.nets}. 
\end{proof}

\section{Extensional Semantics}
\label{sec:ext.sem}
As explained in papers such as \cite{RS08-dbtm,dpibook}, contextual
equivalences and preorders are determined by so-called \emph{extensional actions},
which consist of the observable activities which a system can
have with its external environment. 
We present an extensional Semantics for probabilistic 
wireless networks in Section \ref{sec:ext.act}. 
The final two sections develop technical properties of these actions,
which will be used in the later soundness proofs. They may be safely skipped
at first reading. 
Section \ref{sec:comp.decomp} is devoted to 
basic decomposition and composition results for 
extensional actions, while in 
Section \ref{sec:relating} we relate 
the extensional actions we introduced 
with the reduction relation of 
the testing structures associated with 
networks.

\subsection{Extensional Actions}
\label{sec:ext.act}

Here we design a pLTS whose set of actions can be detected 
(hence tested) by the external environment. 
The intensional semantics in Section~\ref{sec:lang}
already provides a pLTS and it is instructive to see why this is not
appropriate.

Consider $\calM$ and $\calN$ from
Figure~\ref{fig:counterex}, and suppose further that the code $p$ and
$q$, running at $l$ and $k$ respectively, is identical,
  $c!\pc{v}$.  Then we would expect $\calM$ and
  $\calN$ to be behaviourally indistinguishable.  However 
$\calM$ will have an output action, labelled $l.c!v$, which is
not possible for $\calN$.  So output actions cannot record
their source node. What  turns out to be important is the set of target
nodes.  For example if in $\calM$ we added a new output node $m$ to
the interface, with a connection from $l$ then we would be able to
distinguish $\calM$ from $\calN$; the required test
would simply place some appropriate testing code at the new node $m$.

Now we present an extensional semantics for networks; here the
visible actions consist of activities which can be detected (hence
tested) by placing code at the interface of a network. In this
semantics we have internal, input and output actions.

\begin{defi}[Extensional actions]\label{def:sea}   
The actions of the extensional semantics are defined as follows:
\begin{enumerate}

\item \textbf{internal}, $(\Gamma \with M ) \extar{\tau} (\Gamma \with \Delta)$; some internal activity reduces the system
$M$, relative to the connectivity $\Gamma$, to some system $N$, where $N \in \support{\Delta}$. Here the internal 
activity of a network coincides either with some node performing a silent move $m.\tau$ or broadcasting a value 
which cannot be detected by any node in the interface of the network itself. 

Formally, $(\Gamma \with M ) \extar{\tau} (\Gamma \with \Delta)$ whenever $M \not\equiv \Cloc{\omega + s}{m} 
\Cpar N$ and 
\begin{enumerate}
\item $\Gamma \with M  \ar{m.\tau} \Delta$
\item or $\Gamma \with M  \ar{n.c!v} \Delta$ for some value $v$, channel $c$ and node name $n$
satisfying $\Gamma \vdash \rconn{n}{m}$ implies $m \in \nodes{M}$
 \end{enumerate}

\noindent Note that  we are using the notation given in Section \ref{sec:lang}
for defining distributions.  Here $\Delta$ is a distribution over
$\sys$ and so $(\Gamma \with \Delta)$ is a distribution over
networks; however all networks in its support use the same network connectivity
$\Gamma$.

\item \textbf{input}, $(\Gamma \with M) \extar{n.c?v} (\Gamma \with \Delta)$; an observer
  placed at node $n$ can send the value $v$ along the channel $c$ to the network
$(\Gamma \with M)$. For the observer to be able to place the code at node $n$ we must have
$n \in \inp{\Gamma \with M}$. 

Formally $(\Gamma \with M) \extar{n.c?v} (\Gamma \with \Delta)$ whenever 
$M \not\equiv \Cloc{\omega + s}{m} \Cpar N$ and 
 \begin{enumerate}
 \item 
$\Gamma \with M  \ar{n.c?v} \Delta$
 \item $n \in \inp{\Gamma \with M}$
\end{enumerate}

\item \textbf{output}, $(\Gamma \with M) \extar{\eout{c!v}{\eta}} (\Gamma \with
    \Delta)$, where $\eta$ is a non-empty set of nodes; an observer
    placed at any node $n \in \eta$ can receive the value $v$ along the
    channel $c$.  For this to happen each node $n \in \eta$ must be in
    $\outp{\Gamma \with M}$, and there must be some code running at
    some node in $M$ which can broadcast along channel $c$ to each
    such $n$.

Formally, $(\Gamma \with M) \extar{\eout{c!v}{\eta}} (\Gamma \with
    \Delta)$ whenever $M \not\equiv \Cloc{\omega + s}{m} \Cpar N$

 \begin{enumerate}[label=(\roman*)]
 \item  $(\Gamma \with M) \ar{m.c!v} \Delta $ for some node $m$
 \item  $\eta = \setof{n \in \outp{\Gamma \with M}}{\Gamma \vdash \rconn{m}{n}}  \not = \emptyset$. \qed
 \end{enumerate}
\end{enumerate}
\end{defi}

\noindent In the following we will use the metavariable $\lambda$ to range 
over extensional actions. These actions endow the set of networks with the
structure of a pLTS. Thus the terminology used for pLTSs is extended to networks, 
so that in the following we will use terms such as finitary networks or finite branching networks; 
we use the symbol $\extAr{\;\;\;}$ to denote hyper-derivations in the 
extensional pLTS of networks, and $\extArE{\;\;\;}$ to denote extreme derivations.
Also note that we allow an extensional actions to 
be performed only in the case that a network is not $\omega$-successful. 
As we have already stated, we see pLTSs as non-deterministic probabilistic experiments 
whose success is obtained by reaching an $\omega$-successful state. When an 
$\omega$-successful state is reached, we require the experiment to terminate.

A trivial application of Corollary \ref{cor:reduction} ensures that extensional actions 
are preserved by structurally congruent networks.
Further, they do not change the topological structure of a network.

\begin{prop}
\label{prop:static.topology}
Suppose that $\Gamma \with M \extar{\lambda} \Delta$, 
$\Delta \in \dist{\nets}$. 
Then $\Delta = \Gamma \with \Theta$, and $\Theta \in \dist{\sys}$ is 
node stable. Further, for any $\calN \in \support{\Delta}$ 
we have that $\nodes{\calN} = \nodes{\calM}$.
\end{prop}

\begin{proof}
The definition of extensional actions ensures that 
whenever $\Gamma \with M \extar{\lambda} \Delta$ 
then $\Delta = \Gamma \with \Theta$ for some 
$\Theta \in \dist{\sys}$. 
The fact that for any $N \in \support{\Theta}$ 
we have that $\nodes{M} = \nodes{N}$ 
follows from Corollary \ref{cor:stable.dist}.
\end{proof}
Note that, if two distributions $\Delta, \Theta \in \dist{\nets}$ 
are node-stable, if $\calM \testP \calN$ is defined for some  
$\calM \in \support{\Delta}, \calN \in \support{\Theta}$, then 
$\calM' \testP \calN'$ is defined for any $\calM' \in \support{\Delta}, 
\calM' \in \support{\Theta}$. Thus, for node-stable distributions 
of networks it makes sense to lift the operator $\testP$ 
to distributions of networks, defined directly via an application 
of Equation \ref{eq:twoargexpected}. A similar argument holds 
for the symmetric operator $\comp$.

In the following we will need \emph{weak} versions of extensional actions, which
abstract from internal activity, 
provided by the relation $\extar{\tau}$. Internal activity can be modelled by 
the hyper-derivation relation $\extAr{\;\;\;}$, which is a probabilistic generalisation of 
the more standard weak internal relation $\extar{\tau}^{\ast}$. 

\begin{defi}[Weak extensional actions]\label{def:wea}
  \begin{enumerate}
  \item Let ${\mathcal M} \extAr{\tau} \Delta$ whenever we have the hyper-derivation 
        ${\mathcal M} \extAr{\;\;\;} \Delta$

  \item  ${\mathcal M} \extAr{n.c?v} \Delta$ whenever  ${\mathcal M} \extAr{\;\;\;}  \extar{n.c?v} \extAr{\;\;\;}  \Delta$

  \item Let ${\mathcal M} \extAr{\eout{c!v}{\eta}} \Delta$ be the least relation satisfying:
    \begin{enumerate}[label=(\alph*)]
    \item ${\mathcal M} \extAr{\;\;\;} \extar{\eout{c!v}{\eta}} \extAr{\;\;\;} \Delta$ implies ${\mathcal M} \extAr{\eout{c!v}{\eta}} \Delta$

   \item  ${\mathcal M} \extAr{\eout{c!v}{\eta_1}} \Delta'$,  $\Delta' \extAr{\eout{c!v}{\eta_2}} {\Delta}$, where 
            $\eta_1 \cap \eta_2 = \emptyset$, implies 
             ${\mathcal M} \extAr{\eout{c!v}{(\eta_1 \cup \eta_2)}} \Delta$  \qed
    \end{enumerate}
  \end{enumerate}
\end{defi}
\noindent
These weak actions endow the set of networks \nets with the structure
of another pLTS, called the \emph{extensional pLTS} and denoted by
$\pLTSnets$. 

%
%

Some explanation is necessary for the non-standard definition of
output actions in Definition \ref{def:wea}(3).  Informally speaking, the definition of
weak extensional output actions expresses the capability of simulating
broadcast through multicast, which has already been observed 
in Example \ref{ex:bcast1}. A single (weak) broadcast action
detected by a set of nodes $\eta$ can be matched by a sequence of weak
broadcast actions , detected
respectively by $\eta_1, \cdots, \eta_i \subseteq \eta$, provided that
the collection $\{\eta_1,\cdots,\eta_i\}$ is a partition of
$\eta$. This constraint is needed to ensure that
\begin{enumerate}[label=(\roman*)]
\item every node in $\eta$ will detect the transmitted value and
\item no node in $\eta$ will detect the value more than once.
\end{enumerate}
As we will see in Section \ref{sec:soundness}, the ability of a multicast to 
simulate a broadcast is captured by the testing preorders. Roughly speaking, 
in a generic network $\calM$ a broadcast can be converted into a multicast\footnote{
Note that this operation would require a change in the internal topology of $\calM$.} 
leading to a network $\calN$ such that $\calM \Mayleq \calN$. 
Conversely, in the must-testing setting a multicast in a network $\calM$ 
can be replaced by a broadcast leading to a network $\calN$ such that 
$\calM \Mustleq \calN$.

%
%

\subsection{Composition and Decomposition Results}
\label{sec:comp.decomp}
In this Section we prove decomposition and composition results for 
the extensional actions introduced in the previous section. 
In its most general form, the results we want to develop can be 
summarised as follows: given a network $\calM \testP \calN$, 
\begin{description}
\item[Strong Decomposition] for actions of the form $(\calM \testP \calN) \extar{\lambda} \Lambda$ 
we want to determine two actions of the form $\calM \extar{\lambda_1} \Delta$ and 
$\calN \extar{\lambda_2} \Theta$, where $(\Delta \testP \Theta) = \Lambda$,
\item[Weak Composition] conversely, given two actions of the form 
$\calM \extAr{\lambda_1} \Delta$ and $\calN \extAr{\lambda_2} \Theta$, 
we want to determine an action of the form $(\Delta \testP \Theta) 
\extAr{\lambda} (\Delta \interleave \Theta)$.
\end{description}\smallskip

\noindent Unfortunately, the following example shows that this task cannot 
be achieved by relying solely on the extensional semantics. 
\begin{figure}
\begin{align*}
     \begin{tikzpicture}
           \node[state](m){$m$}; 
           \node[state](o)[right=of m]{$o$};
 \path[to]
       (m) edge[thick] (o);
   \begin{pgfonlayer}{background}
    \node [background,fit=(m)] {};
    \end{pgfonlayer}
    \end{tikzpicture}
&&  
     \begin{tikzpicture}
           \node[state](o){$o$};   
   \begin{pgfonlayer}{background}
    \node [background,fit=(m)] {};
    \end{pgfonlayer}
    \end{tikzpicture}\\
\calM = \Gamma_M \with \Cloc{c!\pc{v}}{m}
&&
\calN = \Gamma_N \with \Cloc{c?\pa{x}.\Cnil}{o}
\end{align*}

\caption{A problem with decomposition of actions with respect to $\testP$}
\label{fig:testpdec}
\end{figure}

\begin{exa}
\label{ex:testpdec}
Consider the networks $\calM, \calN$ of Figure \ref{fig:testpdec}. 
It is straightforward to note that $\calM \testP \calN$ is 
defined; further, $(\calM \testP \calN) \extAr{\tau} \pdist{(\calM' \testP \calN')}$, 
where $\calM' = (\Gamma_M \with \Cloc{\Cnil}{m})$, $\calN' = (\Gamma_N \with 
\Cloc{\Cnil}{o})$. 

One could wish to be able to infer this action from the broadcast action 
of the form $\calM \extar{\eout{c!v}{\{o\}}} \pdist{\calM'}$ and an input action 
of the form $\calN \extar{m.c?v} \pdist{\calN'}$. Unfortunately, this last action 
cannot be inferred, since in $\calN$ node $o$ cannot detect the broadcasts 
performed by node $m$.
\end{exa}

The problem of Example \ref{ex:testpdec} arises because the connection between 
from node $m$ to node $n$, which is present in $\calM \testP \calN$, is not 
present in the right hand side of the composition $\calN$. 
As a consequence, the action $\calN \extar{m.c?v} \calN'$ cannot be derived. 
However, note that we could still decompose the transition $\calM \testP \calN \extar{\tau} 
\calM' \testP \calN'$ if we were to focus on the code run by node $o$ in $\calN$, 
rather than on the network $\calN$ itself.

\begin{exa}
\label{ex:testpdec2}
Consider again the networks $\calM,\calN$ of Figure \ref{fig:testpdec}, and 
recall that $\calN = \Gamma_N \with \Cloc{c?\pa{x}.\Cnil}{o}$. Given the action 
$\calM \testP \calN \extar{\tau} \pdist{\calM' \testP \calN'}$, where 
we recall that $\calN' = \Cloc{\Cnil}{o}$, we can 
now infer the 
extensional transition $\calM \extar{\eout{c!v}{\{o\}}} \pdist{\calM'}$ 
and the process transition $c?\pa{x}.\Cnil \ar{c?v} \Cnil$.

The first transitions says that some node in $\calM$ performs 
a broadcast which affects node $o$, while the second one says that 
how the code which node $o$ is running reacts to a broadcast. 
Therefore, it is not surprising to note that the two 
transitions $\calM \extar{\eout{c!v}{\{o\}}} \pdist{\calM'}$ and 
$P \ar{c?v} \Cnil$ can be combined together to obtain 
the original transition $\calM \testP \calN \extar{\tau} \pdist{\calM' \testP \calN'}$.
\end{exa}

Example \ref{ex:testpdec2} leads us to the intuition that composition and decomposition 
results of extensional actions can be developed if we focus on composed networks of the 
form $\calM \testP \calG$, where $\calG \in \mathbb{G}$. Intuitively, 
such results can be obtained by reasoning on the 
extensional transitions of $\calM$ and the 
process transitions performed by the only internal node in $\calG$. 



\begin{prop}[Strong decomposition in \pLTSnets]\label{prop:decomp}
  Let $\calM \in \nets, \calG \in \mathbb{G}$ be two networks such that $\calM \testP \calG$ 
  is defined. Further, let $\Gamma_N, n, s$ be such 
  that $\calG = (\Gamma_N \with \Cloc{s}{n})$; 
  Suppose $\calM \testP\calG \extar{\lambda} \Lambda$, $\Lambda \in \dist{\nets}$. 
  Then $\Lambda = \Delta \testP (\Gamma_N \with \Cloc{\Theta}{n}$ for some 
  $\Delta, \Theta \in \dist{\nets}$ such that 
  
  \begin{enumerate}
		\item if $\lambda = \tau$ then either 
		\begin{enumerate}[label=(\roman*)]
			\item $\calM \extar{\tau} \Delta$, $\Theta = \pdist{s}$, or
			\item $\Delta = \pdist{\calM}$ and $s \ar{\tau} \Theta$, or 
			\item $\calM \extar{\eout{c!v}{\{n\}}} \Delta$, 
			$s \ar{c?v} \Theta$, or
			\item $\calM \extar{\eout{c!v}{\{n\}}} \Delta$, $s \nar{c?v}$ and 
			$\Theta = \pdist{s}$, or
			\item $\outp{\calG} = \emptyset$, 
			$\calM \extar{n.c?v} \Delta, s \ar{c!v} \Theta$, or 
			\item $\outp{\calG} = \emptyset$, 
			$n \notin \inp{\calM}$, $\Delta = \pdist{\calM}$ and 
			$s \ar{c!v} \Theta$.
		\end{enumerate}
		
		\item if $\lambda = \eout{c!v}{\eta}$ then either
		 \begin{enumerate}[label=(\roman*)]
		 \item $\calM \extar{\eout{c!v}{\eta}} \Delta$, 
		 $n \notin \eta$ and $\Theta = \pdist{s}$, or  
		 \item $\calM \extar{\eout{c!v}{\eta \cup \{n\}}} \Delta$, 
		 $n \notin \eta$ and $s \ar{c?v} \Theta$, or 
		 \item $\calM \extar{\eout{c!v}{\eta \cup \{n\}}} \Delta$, 
		 $n \notin \eta$, $s \nar{c?v}$ and $\Theta = \pdist{s}$, or
		 \item $\calM \extar{n.c?v} \Delta$, $s \ar{c!v} \Theta$ 
		 and $\eta = \outp{\calG}$, or 
		 \item $n \notin \inp{\calM}$, $\Delta = \pdist{\calM}$, 
		 $s \ar{c!v} \Theta$ and $\eta = \outp{\calG}$,
		 \end{enumerate}
		
		\item if $\lambda = m.c?v$, where $m\neq n$, then either 
		 \begin{enumerate}[label=(\roman*)]
		 \item $\calM \extar{m.c?v} \Delta$, $\Theta = \pdist{s}$ 
		 and $m \notin \inp{\calG}$, or
		 \item $\Delta = \pdist{\calM}$, $m \notin \inp{\calM}$, $s \ar{c?v} \Theta$, or 
		 \item $\Delta = \pdist{\calM}$, $m \notin \inp{\calM}$, $s \nar{c?v}$ 
		 and $\Theta = \pdist{s}$, or
		 \item $\calM \extar{m.c?v} \Delta$, $s \ar{c?v} \Theta $ 
		 and $m \in \inp{\calG}$, or  
		 \item $\calM \extar{m.c?v} \Delta$, $s \nar{c?v}$ and $\Theta = \pdist{s}$.
		 \end{enumerate}
		\end{enumerate}
\end{prop}
\begin{proof}
  The proof of this Proposition is quite technical, and it is 
  therefore relegated to Appendix \ref{sec:decomposition.results}, Page \pageref{proof:decomp}.
\end{proof}

Next we consider how the weak actions performed by a stable sub-distribution of the form $\Delta 
\testP (\Gamma_N \with \Cloc{s}{n})$, can be inferred by a weak action performed by 
the node-stable distribution $\Delta$ and a process transition performed by the 
state $s$, respectively. 
For our purposes it will suffice to combine a weak extensional transition with 
a strong process transition.

\begin{prop}[Weak/Strong composition in \pLTSnets]
\label{prop:wea.comp}
	Let $\Delta \in \subdist{\nets}$ be a 
	node-stable sub-distributions of networks. 
	Let also $\Gamma_N, n, s$ be such that $\Delta \testP (\Gamma_N \with 
	\Cloc{s}{n})$ is well-defined. Further, 
	let $\calG \in \support{\Gamma_N \with \Cloc{s}{n}}$; 
	here note that the sets $\outp{\calG}$ and $\inp{\calG}$ are 
	completely determined by the connectivity graph $\Gamma_N$.
	\begin{enumerate}
	\item Composition resulting in internal activity:
	\begin{enumerate}[label=(\roman*)]
\item
    $\Delta \extAr{\tau} \Delta'$ 
implies 
 $\Delta \testP (\Gamma_N \with \Cloc{\pdist{s}}{n}) \extAr{\tau}
 \Delta' \testP (\Gamma_N \with \Cloc{\pdist{s}}{n})$,

\item $s \ar{\tau} \Theta$ implies 
$\Delta \testP (\Gamma_N \with \Cloc{\pdist{s}}{n}) 
\extAr{\tau} \Delta \testP (\Gamma_N \with \Cloc{\Theta}{n})$, 

\item $\Delta \extAr{\eout{c!v}{\{n\}}} \Delta'$ and 
$s \ar{c?v} \Theta$ implies 
$\Delta \testP (\Gamma_N \with \Cloc{\pdist{s}}{n}) 
\extAr{\tau} \Delta' \testP (\Gamma_N \with \Cloc{\Theta'}{n}) $, 

\item $\Delta \extAr{\eout{c!v}{\{n\}}} \Delta'$ and 
$s \nar{c?v}$ implies 
$\Delta \testP (\Gamma_N \with \Cloc{\pdist{s}}{n}) 
\extAr{\tau} \Delta' \testP (\Gamma_N \with \Cloc{\pdist{s}}{n}) $, 

\item $\Delta \extAr{n.c?v} \Delta'$, 
$\outp{\calG} = \emptyset$ 
and $s \ar{c!v} \Theta$, then , 
then $\Delta \testP 
(\Gamma_N \with \Cloc{\pdist{s}}{n}) \extAr{\tau} \Delta' \testP 
(\Gamma_N \with \Cloc{\Theta}{n})$,

\item for any $\calM \in \support{\Delta}, n \notin \inp{\calM}$, 
$\outp{\calG} = \emptyset$ and $s \ar{c!v} \Theta$ implies 
$\Delta \testP (\Gamma_N \with \Cloc{\pdist{s}}{n}) 
\extAr{\tau} \Delta \testP (\Gamma_N \with \Cloc{\Theta}{n})$. 
\end{enumerate}

\item Composition resulting in an extensional output:
\begin{enumerate}[label=(\roman*)]
\item  $\Delta \extAr{\eout{c!v}{\eta}} \Delta'$, $n \notin \eta$ 
  implies 
  $\Delta \testP (\Gamma_N \with \Cloc{\pdist{s}}{n}) 
  \extAr{\eout{c!v}{\eta}} \Delta' \testP (\Gamma_N \with \Cloc{\pdist{s}}{n})$,

\item $\Delta \extAr{\eout{c!v}{\eta}} \Delta'$, $ 
 \{n\} \subset \eta$  
  and $s \ar{c?v} \Theta$ 
  implies 
  $\Delta \testP (\Gamma_N \with \Cloc{\pdist{s}}{n}) 
  \extAr{\eout{c!v}{\eta \setminus \{n\}}} 
  \Delta' \testP (\Gamma_N \with \Cloc{\Theta}{n})$,

\item $\Delta \extAr{\eout{c!v}{\eta}} \Delta'$, $ 
 \{n\}  \subset \eta$  
  and $s \nar{c?v}$ 
  implies 
  $\Delta \testP (\Gamma_N \with \Cloc{\pdist{s}}{n}) 
  \extAr{\eout{c!v}{\eta \setminus \{n\}}} 
  \Delta' \testP (\Gamma_N \with \Cloc{\pdist{s}}{n})$,

\item $s \ar{c!v} \Theta$, $\outp{\calG} \neq \emptyset$, 
$\Delta \extAr{n.c?v} \Delta'$ implies that 
$\Delta \testP (\Gamma_N \with \Cloc{\pdist{s}}{n}) 
\extAr{\eout{c!v}{\outp{\calG}}} \Delta' \testP (\Gamma_N \with \Cloc{\Theta}{n})$, 

\item $n \notin \inp{\calM}$ for any $\calM \in \support{\Delta}$ 
and $s \ar{c!v} \Theta'$, $\outp{\calG} \neq \emptyset$  implies that 
$\Delta \testP (\Gamma_N \with \Cloc{\pdist{s}}{n}) 
\extAr{\eout{c!v}{\outp{\calG}}} \Delta \testP (\Gamma_N \with \Cloc{\Theta}{n})$. 
\end{enumerate}

\item Composition resulting in an input:
\begin{enumerate}[label=(\roman*)]
\item $\Delta\extAr{m.c?v} \Delta'$ and 
	$m \notin \inp{\calG}$ implies 
	$\Delta \testP (\Gamma_N \with \Cloc{\pdist{s}}{n}) \extAr{m.c?v} 
	\Delta' \testP (\Gamma_N \with \Cloc{\pdist{s}}{n}$, 
	
\item $m \notin \inp{\calM}$ for $\calM \in \support{\Delta}$, 
$m \in \inp{\calG}$ and $s \ar{c?v} \Theta$ 
	implies $\Delta \testP (\Gamma_N \with \Cloc{\pdist{s}}{n}) \extAr{m.c?v} 
	\Delta \testP (\Gamma_N \with \Cloc{\Theta}{n}$,

\item $m \notin \inp{\calM}$ for$\calM \in \support{\Delta}$, 
$m \in \inp{\calG}$ and $s \nar{c?v}$ 
	implies $\Delta \testP (\Gamma_N \with \Cloc{\pdist{s}}{n}) \extAr{m.c?v} 
	\Delta \testP (\Gamma_N \with \Cloc{\pdist{s}}{n}$,

\item $\Delta \extAr{m.c?v} \Delta'$, 
	$m \in \inp{\calG}$ and 
      $s \ar{c?v} \Theta$ 
      implies 
      $\Delta \testP (\Gamma_N \with \Cloc{\pdist{s}}{n}) \extAr{m.c?v} 
	\Delta' \testP (\Gamma_N \with \Cloc{\Theta}{n}$, 

\item $\Delta \extAr{m.c?v} \Delta'$, 
	$m \in \inp{\calG}$ and 
      $s \nar{c?v} $ 
      implies 
      $\Delta \testP (\Gamma_N \with \Cloc{\pdist{s}}{n}) \extAr{m.c?v} 
	\Delta' \testP (\Gamma_N \with \Cloc{\Theta}{n}$.
	\end{enumerate}
  \end{enumerate}
\end{prop}
\begin{proof}
  See Appendix \ref{sec:decomposition.results}, Page \pageref{proof:composition}.
\end{proof}

%
%

\subsection{Relating Extensional Actions and Reductions}
\label{sec:relating}
In this section we investigate the relationship between 
extensional actions and reductions, defined in Section 
\ref{sec:ts}. 

It follows immediately, from the definition of the 
reduction relation $\red$, that for any 
network $\calM$ we have $\calM \red \Delta$ 
if and only if either $\calM \extar{\tau} \Delta$ 
or $\calM \extar{\eout{c!v}{\eta}} \Delta$. 
However, this result is not true anymore if we 
focus on the extensional transitions performed 
by a distribution of networks. 
\begin{figure}
\begin{center}
\begin{tikzpicture}
          \node[state](m){$m$}; 
          \node[state](o)[right=of m]{$o$}; 
  \path[to]
       (m) edge [thick] (o);
   \begin{pgfonlayer}{background}
    \node [background,fit=(m)] {};
    \end{pgfonlayer}
    \end{tikzpicture}
\end{center}
\caption{Relating reductions with extensional actions.}
\label{fig:red.ext}
\end{figure}

\begin{exa}[Reductions and Extensional Actions] 
\label{ex:red.ext}
Consider the network distribution $\Gamma \with \Delta$, 
where $\Gamma$ is depicted in Figure \ref{fig:red.ext} and 
\[
\Delta = \nicefrac{1}{2} \cdot \pdist{\Cloc{\tau.\Cnil}{m}} + 
\nicefrac{1}{2} \cdot \pdist{\Cloc{c!\pc{v}.\Cnil}{m}}.
\]
\noindent
Note that we have the reductions $\Gamma \with \Cloc{\tau.\Cnil}{m} 
\red \pdist{\Gamma \with \Cloc{\Cnil}{m}}$, and 
$\Gamma \with \Cloc{c!\pc{v}.\Cnil}{m} \red 
\pdist{\Gamma \with \Cloc{\Cnil}{m}}$. These two 
reductions can be combined together to infer  
$\Gamma \with \Delta \red \pdist{\Gamma \with \Cloc{\Cnil}{m}}$.

However, there exists no extensional action of the form $\Gamma \with \Delta 
\extar{\lambda} \Gamma \with \pdist{\Cloc{\Cnil}{m}}$; in fact, 
the only possibility for the network $\Gamma \with \Cloc{\tau.\Cnil}{m}$ is 
to perform a $\tau$-extensional action, while the network $\Gamma \with 
\Cloc{c!\pc{v}.\Cnil}{m}$ can only perform an output action of 
the form $\eout{c!v}{\{o\}}$. In order to infer the action 
$\Gamma \with \Delta \extar{\lambda} \Gamma \with \pdist{\Cloc{\Cnil}{m}}$ 
we require that every network in $\support{\Gamma \with \Delta}$ 
performs the same action $\lambda$; but as we have just noted, 
this is not true for $\Gamma \with \Delta$.
\end{exa}

The problem in Example \ref{ex:red.ext} arises because reductions 
have been identified with two different activities of 
networks; internal actions and broadcasts of messages. 
However, it is possible to avoid this problem if we 
modify a network $\calM$ by removing the nodes in its 
interface. As a consequence, the only activities allowed 
for such a network would be internal actions of the 
form $\tau$.

\begin{defi}[Closure of a Network]
\label{def:net.closure}
Let $\calM$ be a network; we define its closure 
$\closure{\calM} = \calM'$ by letting 
\begin{eqnarray*}
(\calM')_V &=& \calM_V \setminus \type{\calM}\\
\calM' \vdash m, n, \calM \vdash \rconn{m}{n} &\mbox{implies}& \calM' \vdash \rconn{m}{n} 
\end{eqnarray*}
\end{defi}
\noindent 
Obviously the operator $\closure{\cdot}$ preserves well-formed networks.

The actions of a network of the form $\closure{\calM}$ are completely 
determined by those performed by $\calM$, as the following result shows. 
\begin{prop}
\label{prop:ext.closure}
Suppose that $\calM \extar{\lambda} \Delta$, where 
either $\lambda = \tau$ or $\lambda = \eout{c!v}{\eta}$; 
then $\closure{\calM} \extar{\tau} \closure{\Delta}$. 

Conversely, if $\closure{\calM} \extar{\lambda} \Theta$, then 
$\lambda = \tau$ and either 
\begin{enumerate}[label=(\roman*)]
\item $\calM \extar{\tau} \Delta$ and $\Theta = \closure{\Delta}$, 
\item or $\calM \extar{\eout{c!v}{\eta}} \Delta$ and $\Theta = \closure{\Delta}$.
\end{enumerate}
\end{prop}

\begin{proof}
Let $\calM = \Gamma \with M$. 
Suppose that $\Gamma \with M \extar{\eout{c!v}{\eta}} \Gamma \with \Delta$ 
for some $\Gamma, M$ and $\Delta \in \dist{\sys}$. 
By definition of extensional actions there exists a node $m \in \nodes{M}$ 
such that $\Gamma \with M \ar{m.c!v} \Delta$, and 
$\{o\;|\; \Gamma \vdash \rconn{m}{o}\} = \eta$. 
By Proposition \ref{prop:broadcast} it follows that 
$M \equiv \Cloc{c!e.p + q}{m} \Cpar N$, with $\interpr{e} = v$. 
By definition of the operator $\closure{\cdot}$, 
$\closure{\Gamma \with M} = \Gamma' \with M$, where 
 $\Gamma' \vdash \rconn{m}{n}$ whenever 
$\Gamma \vdash \rconn{m}{n}$ and $m,n \in \nodes{M}$.
By an application of propositions \ref{prop:broadcast} and 
\ref{prop:input} we obtain that $\Gamma' \with M 
\ar{m.c!v} \Delta$\footnote{In reality these propositions ensure that 
$\Gamma' \with M \ar{m.c!v} \Delta'$, where $\Delta' \equiv \Delta$. 
However, since the system term $M$ is not changed when performing 
the operation $\closure{\Gamma \with M}$, 
it can be proved that $\Delta = \Delta'$.}. By 
Definition \ref{def:sea}(1) we get 
that $\Gamma' \with M \extar{\tau} \Gamma' \with \Delta$. 
But $\Gamma' \with \Delta$ is exactly $\closure{\Gamma \with \Delta}$.
The case $\Gamma \with M \extar{\tau} \Gamma \with \Delta$ is 
treated similarly.

Conversely, suppose that $\Gamma' \with M \extar{\lambda}  
\Gamma' \with \Theta$. Since $\type{\Gamma' \with M} = \emptyset$, 
it follows that it cannot be $\lambda = m.c?v$, nor $\lambda = \eout{c!v}{\eta}$. 
Therefore the only possibility is that $\Gamma' \with M 
\extar{\tau} \Gamma' \with \Theta$. By Definition 
\ref{def:sea}(1) there are two possible cases. 

\begin{enumerate}
\item $\Gamma' \with M \ar{m.\tau} \Theta$. In this 
case we have that $M \equiv \Cloc{\tau.p + q}{m} \Cpar N$ 
and $\Theta \equiv \interprP{\Cloc{p}{m}} \Cpar \pdist{N}$. 
It follows that $\Gamma \vdash M \ar{m.\tau} \Theta$, 
hence $\Gamma \with M \extar{m.\tau} \Gamma \with \Theta$.

\item $\Gamma' \with M \ar{m.c!v} \Theta$; in this 
case we have that $\Gamma \with M \ar{m.c!v} \Theta$, 
since whenever $\Gamma' \vdash \rconn{m}{n}$ it also 
follows that $\Gamma \vdash \rconn{m}{n}$. Let 
$\eta = \{o \in \outp{\Gamma \with M}\;|\; \Gamma \vdash \rconn{m}{o}\}$. 
If $\eta = \emptyset$, by Definition \ref{def:sea}(1) 
we obtain that $\Gamma \with M \extar{\tau} \Gamma \with \Theta$, 
otherwise Definition \ref{def:sea}(2) ensures that 
$\Gamma \with M \extar{\eout{c!v}{\eta}} \Gamma \with \Theta$.\qedhere
\end{enumerate}

\end{proof}

\noindent Thus only $\tau$-actions are allowed in networks of the form 
$\closure{\calM}$, and the extensional outputs performed 
by $\calM$ are converted in $\tau$-actions in $\closure{\calM}$. 
This relationship can be used to relate the reductions performed 
by the network $\calM$ with the extensional actions of the 
network $\closure{\calM}$

\begin{prop}
Let $\Gamma' \with M = \closure{\Gamma \with M}$; then 
$\Gamma \with M \red \Gamma \with \Delta$ if and 
only if $\Gamma' \with M \red \Gamma' \with \Delta$.
\end{prop} 

\begin{proof}
Follows immediately from Proposition \ref{prop:ext.closure} 
and the definition of reductions.
\end{proof}

\begin{cor}
\label{cor:extreme.reductions}
Let $\Gamma' \with M = \closure{\Gamma \with M}$; then 
$\Gamma \with M \Red{} \Gamma \with \Delta$ if 
and only if $\Gamma' \with M \Red{} \Gamma' \with \Delta$. 
Further, $\Gamma \with M \RedE{} \Gamma \with \Delta$ 
if and only if $\Gamma' \with M \RedE{} \Gamma \with \Delta$. 
\qed
\end{cor}

An important consequence of Corollary \ref{cor:extreme.reductions} 
is that the operator $\closure{\cdot}$ does not affect the 
set of outcomes of a network. 
\begin{cor}
\label{cor:closure.results}
For any network $\calM$, $\Results{\calM} = \Results{\closure{\calM}}$.
\end{cor}

\begin{proof}
Let $\calM = \Gamma \with M$, $\closure{\calM} = \Gamma' \with M$. 
Suppose that $p \in \Results{\calM}$. Then 
$\Gamma \with M \RedE{} \Gamma \with \Delta$, and 
$\sum_{\calN \in \support{\Delta}} \Val{\Delta} = p$. 
By Corollary \ref{cor:extreme.reductions} we have 
that $\Gamma' \with M \RedE{} \Gamma' \with \Delta$, 
hence $p = \sum_{\calN \in \support{\Delta}} \Val{\Delta} \in 
\Results{\Gamma' \with M}$. 

The converse implication is proved analogously.
\end{proof}

\noindent Thus the operator $\closure{\cdot}$ allows to relate weak 
extensional actions and reductions without affecting the 
set of outcomes of a network. As we will see in 
Section \ref{sec:soundness}, this operator is very 
helpful when exhibiting sound proof methods for the 
testing preorders.

\section{A Sound Proof Method for the Testing Preorders}

\label{sec:soundness}
In this Section we present the main results of the paper. 
Following \cite{DGHM09full} we introduce the notion of 
simulation between networks $\calM \forsim \calN$, and 
we prove that it is a sound proof technique for the 
may-testing preorder. This topic is addressed in Section 
\ref{sec:may.sound}.

In Section \ref{sec:must.sound} we give a similar 
result for the must testing relation. We introduce the 
concepts deadlocked network and terminal distributions. 
These will be used to define a novel coinductive 
relation for sub-distribution of networks, the \emph{deadlock 
simulation} $\deadsim$. We prove that the inverse of this 
relation is sound with respect to the must-testing preorder 
$\Mustleq$. Here the use of sub-distributions is necessary, 
since the must-testing preorder is sensitive to divergence.

Finally, in Section \ref{sec:convergent} we focus on convergent 
networks; We show that for such networks a slight 
variation of deadlock simulations can be used as 
a sound proof method for both the may and must testing 
preorders.

\subsection{The May Case}
\label{sec:may.sound}
We begin this section by reviewing the standard definition 
of simulations for probabilistic systems, applied to 
our calculus of probabilistic networks.
\begin{defi}[Simulation preorder]\rm\label{def:sim}
In  $\pLTSnets$ we let 
 $\forsim$ denote the largest relation in $\nets\times \nets$
such that if $\calM \forsim \calN$ then:
\begin{itemize}
\item $\inp{\calM} = \inp{\calN}$, $\outp{\calM} = \outp{\calN}$,
\item 
if $\omega(\calM) = \ttrue$, then 
 $\calN \dar{\tau} \Theta'$ such that for every $\calL \in \support{\Theta'}, \omega(\calL) = \ttrue$,

\item otherwise, whenever $\calM \extAr{\lambda}\Delta$, for $\lambda \in \Act_\tau$, then
  there is a $\Theta \in \dist{\nets}$ such that 
  $\calN \extAr{\lambda}\Theta$ and $\Delta \lift{(\forsim)^e}\Theta$.
\qed
\end{itemize}

\end{defi}
\noindent
This is a mild generalisation of the corresponding definition in \cite{DGHM09full} where we factor in the 
presence of the success predicate $\omega(\,)$ and we compare only networks with the 
same input and output nodes. 

Our aim in this Section is to prove the following theorem:
\begin{thm}[Soundness for May-testing]
\label{thm:may.sound}
Suppose ${\mathcal N}_1,\,{\mathcal N}_2$ are finitary networks.
Then  ${\mathcal N}_1 \forsim  {\mathcal N}_2$ in  $\pLTSnets$ 
implies ${\mathcal N}_1 \Mayleq  {\mathcal N}_2$.
\end{thm}

Before proving Theorem \ref{thm:may.sound}, let us review, this time with 
formal arguments, why our definition 
of weak extensional actions had to be so complicated, and why 
we decided to focus on well-formed networks.

\begin{exa}[On Well-formed Networks]
\label{ex:sound.wellformed}
Consider again the networks $\calM, \calN$ of 
Example \ref{ex:wellformed}. Here we recall 
that $\calN$ is not well-formed, as it has 
a connection between the two external nodes 
$o_1, o_2$. 

Let $(\forsim)'$ be the largest relation over 
(possibly non well-formed) networks which satisfies the requirements 
of Definition \ref{def:sim}.
It is immediate to show that $\calM (\forsim)' \calN$, 
However, $\calM \not\Mayleq \calN$  
In fact, it suffices to consider the test 
$\calT = \Gamma_T \with \Cloc{c?\pa{x}.c!\pc{w}}{o_1} \Cpar 
\Cloc{c?\pa{x}.c?\pa{y}.\omega}{o_2}$, where 
$\Gamma_T \vdash o_1, o_2$ and $\Gamma_T \vdash \notrconn{o_1}{o_2}$ 
to distinguish these two networks. Specifically 
$\Results{\calM \testP \calT} = \{0\}$, 
$\Results{\calN \testP \calT} = \{1\}$, 
and $\{0\} \not\sqsubseteq_{H} \{1\}$. 
Therefore, Theorem \ref{thm:may.sound} cannot 
be extended to non well-formed networks.
\end{exa}

\begin{exa}
\label{ex:bcast}
Consider the networks $\calM$ and $\calN$ in Figure \ref{fig:bcast},
discussed already in Example~\ref{ex:bcast1}.  It is easy to show that
both of them can perform the weak extensional action
$\extAr{\eout{c!v}{\{o_1,o_2\}}}$. However, the inference of the
action is different for the individual networks; while in network
$\calN$ it is implied by the execution of a single broadcast action,
detected by both nodes $o_1$ and $o_2$ simultaneously, in $\calM$ this
is implied by a sequence of weak extensional actions $\calM
\extAr{\eout{c!v}{\{o_1\}}}\extAr{\eout{c!v}{\{o_2\}}}$.

It is therefore possible to exhibit a simulation between $\calN$ and
$\calM$, thus showing that $ \calN \forsim \calM$; 
Theorem~\ref{thm:may.sound} shows that this implies
$\calN \Mayleq \calM$.
\end{exa}

\begin{figure}[t]
                                 
\begin{align*}
     \begin{tikzpicture}
          \node[state](k1){$m$}; 
          \node[state](l1)[ right=of k1]{$l_1$}; 
           \node[state](k2)[below =of k1]{$n$}; 
           \node[state](l2)[right=of k2]{$l_2$};
 \path[to]
       (k1) edge [thick](l1)
       (k2) edge[thick] (l2);
 \path[tofrom]
       (k1) edge [thick](k2);
   \begin{pgfonlayer}{background}
    \node [background,fit=(k1) (k2) ] {};
    \end{pgfonlayer}
    \end{tikzpicture}
&& 
   \begin{tikzpicture}
   \node[state](m){$m$};
   \node[state](l1)[above right= of m]{$l_1$};
   \node[state](l2)[below right= of m]{$l_2$};
   \path[to]
   (m) edge[thick] (l1)
   (m) edge[thick] (l2);
    \begin{pgfonlayer}{background}
    \node [background,fit=(m) ] {};
    \end{pgfonlayer}
    \end{tikzpicture}
&&
     \begin{tikzpicture}
          \node[state](l1){$l_1$}; 
          \node[state](o)[ below right =of l1]{$o$}; 
           \node[state](l2)[below left =of o]{$l_2$}; 
 \path[to]
       (l1) edge [thick](o)
       (l2) edge[thick] (o);
   \begin{pgfonlayer}{background}
    \node [background,fit= (l1) (l2) (o)] {};
    \end{pgfonlayer}
    \end{tikzpicture}  
\\
 \calM = \Gamma_M \with M
&&
 \calN = \Gamma_N \with N
&&
{\mathcal T} = \Gamma_T \with T 
\end{align*}

 \caption{Ensuring soundness}
  \label{fig:counterex2}
\end{figure}

\begin{exa}
Soundness requires that the extensional output actions records the set of target nodes, rather than single nodes.
Consider the networks $\calM, \calN$ depicted in Figure \ref{fig:counterex2}, where 
$M = \Cloc{c!\pc{v} + c?\pa{x}.\Cnil}{m} \Cpar \Cloc{c!\pc{v} + c?\pa{x}.\Cnil}{n}$ and 
$N = \Cloc{c!\pc{v}}{m}$. 
Intuitively, network $\calM$ can broadcast value $v$ to either node $l_1$ or node $l_2$, 
but it cannot broadcast the message to both nodes. On the other hand, in $\calN$ node $m$ can 
broadcast value $v$ simultaneously to both nodes $l_1$ and $l_2$.

Note that 
${\mathcal N} \not\Mayleq {\mathcal M}$ because of the test
${\mathcal T}$ given in Figure \ref{fig:counterex2}, where $T = \Cloc{c?\pa{x}.c!\pc{x}}{l_1} \Cpar \Cloc{c?\pa{x}.c!\pc{x}}{l_2} 
\Cpar \Cloc{c?\pa{x}.c?\pa{y}.\omega}{o}$. 
In fact, in $(\calN \testP \calT)$ both nodes $l_1$ and $l_2$ will receive the broadcast of value $v$ along channel 
$c$ performed by node $m$; 
each of these nodes will forward the received value to node $o$. Therefore, node $o$ will receive two values, one 
from node $l_1$ and one from node $l_2$, after which it will reach a successful state. That is, 
$1 \in \Results{\calN \testP \calT}$.

On the other hand, none of the 
computations of $(\calM \testP \calT)$ leads to a successful configuration. 
There are in fact two possibilities; either node $m$ broadcasts value $v$ 
to nodes $n$ and $l_1$, or node $n$ broadcasts value $v$ to nodes $m$ 
and $l_2$.  Note that one of the effects of  node $m$ (respectively $n$) broadcasting value 
$v$ is that of preventing node $n$ (respectively $m$) from performing a second 
broadcast. 
As a consequence, only one among nodes $l_1, l_2$ will receive value $v$ 
along channel $c$. When node $l_1$ (respectively $l_2$) receives value $v$, it will 
forward it to node $o$. After this broadcast has been performed the network 
reaches a configuration in which node $o$ is still waiting to receive a 
value along channel $c$ before entering a successful state; further, 
the computation of the network cannot proceed anymore, 
since none of its nodes can perform a broadcast.
There are no other possible behaviours of 
$\calM \testP \calT$, therefore we obtain $\Results{\calM \testP \calT} = \{0\}$. 

Since $1 \in \Results{\calN \testP \calT}$, but $1 \notin \Results{\calM \testP \calT}$, 
it follows that $\calN \not\Mayleq \calM$.
%

We also have ${\mathcal N} \not\forsim {\mathcal M}$ because ${\mathcal N}$ can perform the output action labelled
$\eout{c!v}{\sset{l_1,l_2}}$, which can not be matched by ${\mathcal M}$. 

However suppose we were to restrict $\eta$ in the definition of
extensional output actions, part (3) of Definition~\ref{def:sea}, to be
singleton sets of node names. Then in the resulting pLTS it is easy to
check that ${\mathcal M}$ can simulate ${\mathcal N}$. 
The broadcast of value $v$ in network $\calN$, which can be detected by both 
nodes $l_1$ and $l_2$, can be matched by either the broadcast performed by node $m$ 
(which can be detected by node $l_1$) or by the broadcast 
performed by node $n$ (which can be detected by node $l_2$) in $\calM$.
In other words, with the proposed 
simplification the resulting simulations would not be sound; that
is, Theorem~\ref{thm:may.sound} would no longer hold.
\end{exa}

Let us now turn to the proof of Theorem \ref{thm:may.sound};
it relies on the following two technical results, whose proofs are developed
in Section~\ref{sec:technical}. 
\begin{thm}[Compositionality]\label{thm:composition}
Let $\calM, \calN$ be finitary networks such that $
\inp{\calM} = \inp{\calN}$, $\outp{\calM} = \outp{\calN}$. 
Also, suppose $\calL$ is a network such that both 
$\calM \testP \calL$ and $\calN \testP \calL$ are defined. 
Then $\calM \forsim  \calN$ implies  
$(\calM \testP \calL)  \forsim (\calN \testP \calL)$.
\end{thm}
\begin{proof}
  See Corollary~\ref{cor:composition} in Section~\ref{sec:technical.may.comp}
\end{proof}

\begin{thm}[Outcome preservation]\label{thm:results}
  In $\pLTSnets$, $\Delta \lift{\forsim} \Theta$ implies
  $\Results{\Delta} \sqsubseteq_{H} \Results{\Theta}$.
\end{thm}
\begin{proof}
  See  Corollary~\ref{cor:results} in Section~\ref{sec:technical.result.pres}. 
\end{proof}

\textbf{Proof of Theorem~\ref{thm:may.sound}}: 
This is now a straightforward application of  Compositionality and Theorem~\ref{thm:results}. 

Let us assume that $\calN_1 \forsim \calN_2$.
To prove the conclusion,  ${\mathcal N}_1 \Mayleq  {\mathcal N}_2$, we must show that 
$\Results{{\mathcal N}_1 \testP {\mathcal T}} \sqsubseteq_{H}  
\Results{{\mathcal N}_2 \testP {\mathcal T}}$
for an arbitrary testing network ${\mathcal T}$ 
such that both  ${\mathcal N}_1 \testP {\mathcal T}$ and ${\mathcal N}_2 \testP {\mathcal T}$ are defined. 
For such a ${\mathcal T}$ Compositionality entails  
$({\mathcal N}_1 \testP {\mathcal T}) \forsim ({\mathcal N}_2 \testP {\mathcal T})$, and now we can apply 
Theorem~\ref{thm:results}. \qed

\subsection{The Must Case}
\label{sec:must.sound}
In this Section we give a sound proof method for the must-testing 
preorder. It has already been observed that, for standard 
probabilistic process calculi such as pCSP \cite{DGHM09full}, 
the must-testing can be characterised by looking at the 
set of actions which are not enabled in processes. 
This is because outputs in such a calculus are blocking 
actions; in order for an action to be performed, a synchronisation 
(either within the process or with the external environment) 
must occur. This leads to the notion of \emph{failure simulations}.

This is not true for broadcast systems, where the nature of 
the broadcast action is non-blocking. It has been observed 
in \cite{Ene02} that, if broadcast communication is assumed, 
then the must-testing relation can be used to observe only if  
a computation of a process cannot proceed (that is, no 
internal actions nor broadcasts are possible). 

Following this intuition, we readapt the notion of \emph{failure 
simulation} given in \cite{DGHM09full}. 
\begin{defi}[Deadlocked Networks, Terminal Distributions]
\label{def:deadlock}
The predicate $\delta : \nets \rightarrow \{\ttrue, \ffalse\}$ 
is defined by letting $\delta(\calM) = \ttrue$ whenever the 
following conditions are met:
\begin{enumerate}[label=(\roman*)]
\item $\omega(\calM) = \ffalse$,
\item $\calM \extar{\tau}\hspace{-10pt}\not\;\hspace{10pt}$,
\item $\calM \extar{\eout{c!v}{\eta}}\hspace{-18pt}\not\;\hspace{18pt}$ 
for any $c,v, \eta$. \qed
\end{enumerate}
\noindent
Networks for which the predicate $\delta$ is true 
are called \emph{deadlock networks}, or \emph{deadlocked}. 
Note that the term \emph{deadlock} network makes sense only 
in the reduction semantics. Deadlock networks 
are those whose  computation cannot proceed 
autonomously; however, it could be the case that an 
input from the external environment makes the network 
evolve in a distribution where the computation can 
proceed, thus resolving the deadlock.

A distribution $\Delta$ is said to be \emph{terminal} if 
any network in its support is either deadlocked or 
successful.
\end{defi}
\noindent
Next we present a notion of simulation which is sensitive 
to deadlocked networks:

\begin{defi}[Deadlock Simulations]
\label{def:ds}
The relation $\deadsim \subseteq \nets \times \subdist{\nets}$ is 
the largest relation such that whenever $\calM \deadsim \Theta$
\begin{enumerate}[label=(\roman*)]
\item if $\delta(\calM) = \ttrue$ then $\Theta \extAr{\tau} \Theta'$ 
for some $\Theta'$ such that $\delta(\calN) = \ttrue$ for any 
$\calN \in \support{\Theta'}$, 
\item if $\calM \extAr{\lambda} \Delta$ for some 
$\Delta \in \subdist{\nets}$ then $\Theta \extAr{\lambda} \Theta'$ 
for some $\Theta'$ such that $\Delta \lift{\deadsim} \Theta'$.\qed
\end{enumerate}
\end{defi}
\noindent
We use the notation $\Theta \invdeadsim \calM$ for 
$\calM \deadsim \Theta$.
Note that deadlock simulations are sensitive to divergence. 
That is, whenever $\calM \deadsim \Theta$ and $\calM \extAr{\tau} \varepsilon$, 
then $\Theta \extAr{\tau} \varepsilon$. 
To prove this, note first that for any relation 
$\calR \subseteq \nets \times \subdist{\nets}$ then 
whenever $\Delta \lift{\calR} \Theta$ we have that  
$\size{\Delta} \geq \size{\Theta}$; this follows at once 
from Definition \ref{def:lift}. 
Thus if $\calM \deadsim \Theta$ and $\calM \extAr{\tau} \varepsilon$, 
by definition it follows that $\Theta \extAr{\tau} \Theta'$ 
for some $\Theta'$ such that $\varepsilon \lift{\deadsim} \Theta'$; 
but this means that $0 = \size{\varepsilon} \geq \size{\Theta'}$, 
or equivalently $\Theta' = \varepsilon$.

Before discussing the soundness of deadlock simulations for must testing
let us discuss briefly the definition of 
deadlock simulations.
First, note that the deadlock simulation relation $\deadsim$ is 
lifted to a relation between sub-distributions, rather than 
to a relation between (full) distributions. 
This is needed, since the Must-testing preorder is 
sensitive to divergence.

\begin{exa}[Divergence]
\label{ex:divergence}
Let $\calM = \Gamma \with \Cloc{\Cnil}{m}$, 
$\calN = \Gamma \with \Cloc{\text{Div}}{m}$, 
where $\Gamma$ is the connectivity graph containing 
the sole node $m$ and no connections and 
$\text{Div} \Leftarrow \tau.\text{Div}$. 

It is immediate to show that $\calM \not\invdeadsim \calN$, 
since the move $\calN \extAr{\tau} \varepsilon$ cannot 
be matched by $\calM$. 

Even more, it is straightforward to 
note that $\calM \not\Mustleq \calN$. Consider in fact the test 
$\calT = \Gamma_T \with \Cloc{\tau.\omega}{n}$,  
where $\Gamma_T$ is the connectivity graph consisting of the 
sole node $n$. Note that, since $\calN \extAr{\tau} \varepsilon$ 
we also have $\calN \testP \calT \extAr{\tau} \varepsilon$, 
and therefore $0 \in \Results{\calN \testP \calT}$. Intuitively 
the last hyper-derivation can be inferred by always letting 
process $\text{Div}$ perform a $\tau$-action in 
$\calN \testP \calT$. On the other hand we have 
that $\Results{\calM \testP \calT} = \{1\}$, since the only possible 
transition for  $\calM \testP \calT$ is 
$(\calM \testP \calT) \extar{\tau} \pdist{\calM \testP (\Gamma_T \with \Cloc{\omega}{n})}$, 
and the latter is $\omega$-successful.  
Since $0 \in \Results{\calN \testP \calT}$, but $0 \notin \Results{\calN \testP \calT}$, 
it follows that $\calM \not\Mustleq \calN$.

However, suppose that Definition \ref{def:ds} is changed 
by only considering hyper-derivations of the form 
$\calM \extAr{\lambda} \Delta$, where $\Delta$ is 
a distribution. 
In this case we would have that the only possible (weak) move 
for the network $\calN$ above is $\calN \extAr{\tau} \pdist{\calN}$, 
which can be matched by $\calM \extAr{\tau} \pdist{\calM}$. 
Therefore we would have that $\calM \invdeadsim \calN$, and 
since we already proved that $\calM \not\Mustleq \calN$ Theorem 
\ref{thm:must.sound} would no longer hold.
\end{exa}

Also, deadlock simulation is defined as a 
relation between networks and sub-distributions of 
states, rather than a relation between networks. 
This is in contrast with the definition of simulation 
given in Section \ref{sec:may.sound}, which 
has been defined as a relation between 
networks.
In fact, since deadlock simulation also considers 
sub-distributions the latter approach would have led to 
a less discriminating relation.

\begin{remark}
Note that, for any sub-distribution $\Delta$ 
we have that $\Delta \lift{\deadsim} \varepsilon$; 
in fact, it is straightforward to show that 
$\varepsilon \extAr{\lambda} \varepsilon$ for 
any extensional action $\lambda$, and 
since $\support{\varepsilon} = \emptyset$ we 
also have that $\delta(\calN) = \ttrue$ 
for any $\calN \in \support{\varepsilon}$.

Now suppose that deadlock simulation had 
been defined as a relation between networks, 
by letting $\calM \deadsim' \calN$ be the largest 
relation such that 
\begin{enumerate}[label=(\roman*)] 
\item if $\delta(\calM) = \ttrue$ then $\calN \extAr{\tau} \Theta$ 
for some $\Theta$ such that $\delta(\calL) = \ttrue$ for any 
$\calL \in \support{\Theta}$, 
\item whenever $\calM \extAr{\lambda} \Delta$ then 
$\calN \extAr{\lambda} \Theta$ with $\Delta \lift{(\deadsim')^e} \Theta$. 
\end{enumerate}\medskip

\noindent While Theorem \ref{thm:must.sound} would still 
hold with this definition of deadlock simulations, 
it is straightforward to show that if 
$\Delta \lift{(\deadsim')^e} \Theta$ then $\size{\Delta} = 
\size{\Theta}$. As a consequence, whenever 
$\Delta \lift{(\deadsim')^e} \varepsilon$ it would 
follow that $\Delta = \varepsilon$. 

Therefore we have that, for any non-empty sub-distribution 
$\Delta$, $\Delta \lift{\deadsim} \varepsilon$, but not
$\Delta \lift{(\deadsim')^e} \varepsilon$;
that is, the definition of deadlock simulation proposed 
above is less discriminating than the one given in 
Definition \ref{def:ds}.
\end{remark}

The proof of soundness of deadlock simulations follows the same structure as
the corresponding proof for simulations in Section~\ref{sec:may.sound}. It relies on
the following two technical results. 

\begin{thm}[Outcome preservation]
\label{thm:ds.outcomes}
If $\Delta \lift{\deadsim} \Theta$ then 
$\Results{\Delta} \sqsubseteq_S \Results{\Theta}$. 
\end{thm}
\begin{proof}
  See Corollary~\ref{cor:ds.outcomes} in Section~\ref{sec:technical.result.pres}. 
\end{proof}

\begin{thm}[Compositionality]
\label{thm:ds.comp}
Let $\calM$ be a network and $\Theta$ be a stable sub-distribution 
such that $\calM \deadsim \Theta$. Then, for any network $\calN$ such 
that both $\calM \testP \calN$ and $\Theta \testP \calN$ are defined 
it follows that $(\calM \testP \calN) \deadsim (\Theta \testP \pdist{\calN})$.
\qed
\end{thm}

\begin{proof}

See Corollary~\ref{cor:ds.comp} in Section~\ref{sec:technical.must.comp}
\end{proof}

\begin{thm}[Soundness for Must-testing]
\label{thm:must.sound}
Let $\calM, \calN$ be two finitary networks 
such that $\inp{\calM} = \inp{\calN}$, $\outp{\calM} = 
\outp{\calN}$. If $\pdist{\calM} \invdeadsim \calN$ 
then $\calM \Mustleq \calN$.
\end{thm}
\begin{proof}
Suppose that $\calM \invdeadsim \calN$, 
and suppose that $\inp{\calM} = \inp{\calN}$, 
$\outp{\calM} = \outp{\calN}$. Note also that 
$\pdist{\calM}$ is a stable distribution. 

Let $\calT$ be a network such that both 
$\calM \testP \calT$ and $\calN \testP \calT$ 
are defined. 
Compositionality, Corollary \ref{cor:ds.comp} 
gives that $(\calM \testP \calT) \invdeadsim 
(\calN \testP \calT)$, while Theorem \ref{thm:ds.outcomes} 
states that $\Results{\pdist{\calM \testP \calT}} \sqsubseteq_S 
\Results{\pdist{\calN \testP \calT}}$. Since the 
testing network $\calT$ has been chosen arbitrarily, it 
follows that $\calM \Mustleq \calN$.
\end{proof}

\subsection{Proof Methods for Convergent Networks}
\label{sec:convergent}

One of the main drawbacks of deadlock simulations 
is that they require the use of probability sub-distributions. 
As we have seen in example \ref{ex:divergence}, 
using sub-distributions is necessary for ensuring 
the validity of Theorem \ref{thm:must.sound}. We 
have also emphasised that this constraint is necessary since 
the must-testing preorder is sensitive to divergence. 

However, sub-distributions are no longer needed if we 
focus on convergent networks, that this those whose 
generated pLTS (with respect to the strong extensional semantics) 
does not contain a state $\calM$ for which $\calM \extAr{\tau} \varepsilon$ 
holds. 

\begin{defi}[Divergence-free Deadlock Simulations]
\label{def:dfdeadsim}
The relation $\dfdeadsim \subseteq \nets \times \nets$ is defined 
as the largest relation such that whenever 
$\calM \dfdeadsim \calN$
\begin{enumerate}[label=(\roman*)]
\item if $\delta(\calM) = \ttrue$ then $\calN \extAr{\tau} \Theta$ 
for some $\Theta$ such that $\delta(\calL) = \ttrue$ whenever 
$\calL \in \support{\Theta}$,
\item if $\calM \extAr{\lambda} \Delta$ then 
$\calN \extAr{\lambda} \Theta$ for some $\Theta$ 
such that $\Delta \lift{\dfdeadsim^e} \Theta$.
\end{enumerate}
\end{defi}

\noindent We write $\calN \invdfdeadsim \calM$ for $\calM \dfdeadsim \calN$. 

\begin{thm}[Soundness for Convergent Networks, Must-testing]
\label{thm:must.dfsound}
Let $\calM, \calN$ be two convergent 
networks such that $\inp{\calM} = \inp{\calN}$ 
and $\outp{\calM} = \outp{\calN}$; if $\calM \invdfdeadsim \calN$ then 
$\calM \Mustleq \calN$.
\end{thm}

\begin{proof}
It suffices to show that $\dfdeadsim$ is included in 
$\deadsim$. To this end, note that if 
$\calM \extAr{\lambda} \Delta$ and $\calM$ is a convergent 
network, then $\size{\Delta} = 1$, by Proposition 
\ref{prop:convergent.moves}.
\end{proof}

Having a simpler sound proof technique is not the only 
advantage that we gain by focusing on convergent networks. 
In fact, if we make a further restriction and we compare networks 
whose codes running at nodes do not contain the success clause 
$\omega$, it follows that the relation $\dfdeadsim$ is also included 
in the may-testing preorder $\Mayleq$. This restriction is justified 
since in general we require the tests applied to a network, 
rather than the networks to be tested, to contain the clause 
$\omega$ to denote the success of an experiment. 

\begin{thm}[Soundness for Convergent Networks, May-testing]
\label{thm:may.dfsound}
A network $\Gamma \with M$ is proper if the term $M$ does not 
contain any occurrence of the special clause $\omega$. 

Let $\calM, \calN$ be convergent, proper networks such that 
$\inp{\calM} = \inp{\calN}$, $\outp{\calM} = \outp{\calN}$. 
If $\calM \dfdeadsim \calN$ 
then $\calM \Mayleq \calN$.
\end{thm}

\begin{proof}
It is trivial to note that the relation $\dfdeadsim$ is included in 
$\forsim$ when restricted to convergent, proper networks. 
The result follows then from Theorem \ref{thm:may.sound}.
\end{proof}

\section{Technical development}
\label{sec:technical}
In this section we collect the proofs of some technical results underlying our soundness
results; it may safely be skipped by the uninterested reader. 
The first to results concern the compositionality of the simulation preorders. The last 
outlines the proofs about \emph{Outcome preservation}. 
\subsection{Compositionality for  simulations}
\label{sec:technical.may.comp}

This section is devoted to the proof of Theorem~\ref{thm:composition}, namely
that the simulation preorder is preserved by  the extension operator $\testP$. 
In general such compositionality results depends on \emph{decomposition} and \emph{(re-)composition}
results for the actions used in the definition of simulations. Definition~\ref{def:sim} 
uses \emph{weak} extensional actions, and providing \emph{decomposition} results for 
these would be a  difficult undertaking. Instead we  first give an alternative 
characterisation of the simulation preorder, for which \emph{decomposition results} for
\emph{strong} actions is sufficient. These (strong) decomposition results, and 
\emph{(re)-composition} results for \emph{weak} actions have already been given in 
Section~\ref{sec:comp.decomp}.

\begin{defi}[Simple simulations]\rm\label{def:sims}
In  $\pLTSnets$ we let 
 $\sforsim$ denote the largest relation in $\nets\times {\nets}$
such that if ${\calM}  \sforsim \calN$ then: 
\begin{itemize}
\item $\inp{\calM} = \inp{\calN}, \outp{\calM} = \outp{\calN}$,
\item 
if $\omega(\calM) = \ttrue$  then $\calN \extAr{\tau} \Theta $
such that $\omega(\calL) = \ttrue$ for any $\calL \in \support{\theta}$,

\item otherwise, 
 \begin{enumerate}[label=(\roman*)]
  \item whenever $\calM\extar{\lambda}\Delta$
  there is a $\Theta \in \dist{\nets}$ with
  $\calN \extAr{\lambda}\Theta$ and $\Delta\lift{(\sforsim)^e}\Theta$.
 \end{enumerate}
  
\end{itemize}
\end{defi}

\begin{thm}[Alternative characterisation]\label{thm:altchar}
  In  $\pLTSnets$,
  $\calM \forsim \calN$ if and only if $\calM \sforsim \calN$, 
  provided that $\calM$ and $\calN$ are finitary networks.
\end{thm}

\begin{proof}(Outline)
The proof is similar in style to the one of Theorem 7.20 of \cite{DGHM09full}; however, there are some extra complications, mainly because of the more complicated definition of weak extensional actions. 
Here we report a detailed outline of the proof; we first prove that 
$\forsim$ is included in $\sforsim$, then we show that, for finitary networks, 
the converse inclusion also holds.

Showing that the relation $\forsim$ is included in $\sforsim$ is straightforward. 
We only need to show that $\forsim$ satisfies the constraints of Definition \ref{def:sims}.
Suppose that $\calM \forsim \calN$. Then this hypothesis ensures that 
\begin{itemize}
\item 
$\inp{\calM} = \inp{\calN}, \outp{\calM} = \outp{\calN}$ and 
\item if $\omega(\calM) = \ttrue$ then $\calN \extAr{\tau} \Theta$ such that 
$\omega(\calL) = \ttrue$ for any $\calL \in \support{\Theta}$.
\end{itemize}
Suppose however that $\omega(\calM) = \ffalse$ and $\calM \extar{\lambda} \Delta$ 
for some $\Delta$. Then we also have that $\calM \extAr{\lambda} \Delta$, from which 
it follows from the hypothesis $\calM \forsim \calN$ that there exists $\Theta$ 
such that $\calN \extAr{\tau} \Theta$ and $\Delta \lift{(\forsim)^e} \Theta$, which is exactly 
what we wanted to show.

It remains to show that, for finitary networks, the relation $\sforsim$ is included in $\forsim$. 
Here the main difficulty consists in showing that, whenever $\calM \sforsim \calN$, 
$\omega(\calM) = \ffalse$ and $\calM \extAr{\lambda} \Delta$, then $\calN \extAr{\lambda} \Theta$ 
for some $\Theta$ such that $\Delta \lift{(\sforsim)^e} \Theta$. 
The proof of this statement is performed by a case analysis on the action $\lambda$;

\begin{enumerate}
\item First suppose that $\lambda = \tau$.
This case can be proved in the analogous way of 
Theorem 7.20 of \cite{DGHM09full}. Note that we require networks to be finitary in this 
case, since the proof requires properties of hyper-derivations 
which in general are not satisfied by infinitary plTSs; see Lemma 6.12 of \cite{DGHM09full}.

\item Now suppose that $\lambda = i.c?v$ for some $i, c, v$; then $\calM \extAr{i.c?v} \Delta$ 
implies that there exist $\Delta', \Delta''$ such that $\calM \extAr{\tau} \Delta'\extar{i.c?v} \Delta'' \extAr{\tau}$. 
Since $\calM \sforsim \calN$ and $\calM \extAr{\tau} \Delta'$, by the previous case $\calN \extAr{\tau} \Theta'$ 
for some $\Theta'$ such that $\Delta' \lift{(\sforsim)^e} \Theta'$. Since $\Delta' \extar{i.c?v} \Delta''$, we can 
conclude that $\Theta' \extAr{i.c?v} \Theta''$ for some $\Theta''$ such that 
$\Delta'' \lift{(\sforsim)^e} \Theta''$. Finally, since $\Delta'' \extAr{\tau} \Delta$, 
by the previous case\footnote{Note that here it is necessary to decompose $\Delta''$ as a sum 
of state-based networks, each of which can perform a weak $\tau$-action.} we obtain that $\Theta'' \extAr{\tau} \Theta$ for some 
$\Theta$ such that $\Delta \lift{(\sforsim)^e} \Theta$. 
Therefore we have shown that $\calN \extAr{\tau} \Theta' \extAr{i.c?v}\Theta'' \extAr{\tau} \Theta$, 
or equivalently $\calN \extAr{i.c?v} \Theta$, and $\Delta \lift{(\sforsim)^e} \Theta$, which 
is exactly what we needed to prove.

\item Finally, suppose that $\lambda = \eout{c!v}{\eta}$ for some $c, v$ and non-empty set of 
nodes $\eta$. We perform an inner induction on the proof of the derivation $\calM \extAr{\eout{c!v}{\eta}} \Delta$;
\begin{itemize}
\item $\calM \extAr{\eout{c!v}{\eta}} \Delta$ because $\calM \extAr{\tau} \extar{\eout{c!v}{\eta}} \extAr{\tau} \Delta$. 
This case is identical to the one $\lambda = i.c?v$,
\item $\calM \extAr{\eout{c!v}{\eta}} \Delta$ because $\calM \extAr{\eout{c!v}{\eta_1}} \Delta' \extAr{\eout{c!v}{\eta_2}} \Delta$, 
where $\eta_1 \cup \eta_2 = \eta$, 
$\eta_1 \cap \eta_2 = \emptyset$. Since $\calM \sforsim \calN$, by the (inner) inductive hypothesis we have that 
$\calN \extAr{\eout{c!v}{\eta_1}} \Theta'$ for some $\Theta'$ such that $\Delta' \lift{(\sforsim)^e} \Theta'$. 
A second application of the inductive hypothesis to the last statement gives that $\Theta' \extAr{\eout{c!v}{\eta_2}} \Theta$ 
for some $\Theta$ such that $\Delta \lift{(\sforsim)^e} \Theta$. Therefore we have shown that 
$\calN \extAr{\eout{c!v}{\eta_1}} \Theta' \extAr{\eout{c!v}{\eta_2}} \Theta$, or equivalently 
$\calN\extAr{\eout{c!v}{\eta}} \Theta$ (recall that $\eta_1 \cup \eta_2 = \eta$ and $\eta_1 \cap \eta_2 = \emptyset$), 
and $\Delta \lift{(\sforsim)^e} \Theta$, as we wanted to prove.\qedhere
\end{itemize} 
\end{enumerate}
\end{proof}\smallskip

\noindent Theorem \ref{thm:altchar} enables us to exploit the results 
developed in Section \ref{sec:comp.decomp} for proving 
the compositionality of $\forsim$ with respect to the 
extension operator $\testP$. Since such results are valid 
only in the case that a network $\calM$ is composed
 with a generating network 
$\calG$, we first focus on compositionality with respect 
to a generating network.

\begin{thm}\label{thm:single.comp}
  Suppose $\inp{\calM} = \inp{\calN}$, $\outp{\calM} = \outp{\calN}$ 
  and both $\calM \testP \calG$ and 
$\calN \testP \calG$ are defined. Then 
$\calM \sforsim \calN $ implies 
$\calM \testP \calG  \sforsim \calN \testP \calG$.
\end{thm}

\begin{proof}
It suffices to show that the relation 
\[
{\mathcal S} = \{ ((\calM \testP \calG), (\calN \testP \calG) \;|\; \calM \sforsim \calN) \} 
\]
\noindent
satisfies the requirements of Definition \ref{def:sims}. 
We denote the network $\calG$ with $\Gamma_G \with \Cloc{s}{n}$. 

Suppose that $\calM \testP \calG \extar{\eout{c!v}{\eta}} \Delta$; 
note that the definition of extensional output ensures that 
$\eta \neq \emptyset$.
We need to show that $\calN \testP \calG \extAr{\eout{c!v}{\eta}} 
\Theta$, and $\Delta \lift{{\mathcal S}^e} \Theta$. By 
Proposition \ref{prop:decomp} we know that 
$\Delta = \Delta_M \testP (\Gamma_G \with \Cloc{\Theta_n}{n})$, 
where $\Delta_M$, $\Theta_n$ are determined according to 
six different cases, which are considered below.

\begin{enumerate}[label=(\roman*)]
\item $\calM \extar{\eout{c!v}{\eta}} \Delta_M$, $n \notin \eta$ and 
$\Theta_n = \pdist{s}$; 
since $\calM \sforsim \calN$ it follows that 
$\calN \extAr{\eout{c!v}{\eta}} \Theta_N$ for 
some $\Theta_N$ such that  
$\Delta_M \lift{(\sforsim)^e} \Theta_N$. Since 
$n \notin \eta$, By Proposition 
\ref{prop:wea.comp} it follows that $\calN \testP 
(\Gamma_G \with \Cloc{s}{n}) \extAr{\eout{c!v}{\eta}} 
\Theta_N \testP (\Gamma_G \with \Cloc{\pdist{s}}{n})$. 
Now it suffices to note that 
$(\Delta_M \testP \Gamma_G \with \Cloc{\pdist{s}}{n}) \lift{({\mathcal S})^e} 
(\Delta_N \testP \Gamma_G \with \Cloc{\pdist{s}}{n})$, 
as we wanted to prove
\item $\calM \extar{\eout{c!v}{\eta \cup \{n\}}} \Delta_M$, 
$n \notin \eta$ and $s \ar{c?v} \Theta_n$.  
Since $\calM \sforsim \calN$ it follows that 
$\calN \extAr{\eout{c!v}{\eta \cup \{n\}}} \Theta_N$, 
where $\Delta_M \lift{(\sforsim)^e} \Theta_N$. 
Since $\eta \neq \emptyset$, that is  
$\{n\} \subset (\eta \cup \{n\})$, we can apply Proposition \ref{prop:wea.comp}(2)(ii)
to the weak extensional transition $\calN \extAr{\eout{c!v}{\eta \cup \{n\}}} 
\Theta_N$ and the process transition $s \ar{c?v} \Theta_n$ 
to infer $\calN \testP (\Gamma_G \with \Cloc{s}{n}) 
\extAr{\eout{c!v}{(\eta \cup \{n\}) \setminus \{n\}}} 
\Theta_N \testP (\Gamma_G \with \Cloc{\Theta_n}{n})$; since 
we are assuming that $n \notin \eta$, the 
latter can be rewritten as  $\calN \testP (\Gamma_G \with \Cloc{s}{n} 
\extAr{\eout{c!v}{\eta}} 
\Theta_N \testP (\Gamma_G \with \Cloc{\Theta_n}{n})$; 
again, since $\Delta_M \lift{(\sforsim)^e} \Theta_N$, we have that 
$(\Delta_M \testP \Gamma_G \with \Cloc{\Theta_n}{n}) 
\lift{{\mathcal S}^e} \Theta_N \testP (\Gamma_G \with \Cloc{\Theta_n}{n})$

\item $\calM \extar{\eout{c!v}{\eta \cup \{n\}}} \Delta_M$, 
$s \nar{c?v}$ and $\Theta_n = \pdist{s}$. This case is 
similar to the previous one, this time employing Proposition 
\ref{prop:wea.comp}(2)(iii)

\item $\calM \extar{m.c?v} \Delta_M$, $s \ar{c!v} \Theta_n$ and 
$\eta = \outp{(\Gamma_G \with \Cloc{s}{n})}$. 
Since $\calM \sforsim \calN$ we have that 
$\calN \extAr{m.c?v} \Theta_N$ for some $\Theta_N$ 
such that $\Delta_M \lift{(\sforsim)^e} \Theta_N$. 
It follows from Proposition \ref{prop:wea.comp}(2)(iv) 
that $\calN \testP (\Gamma_G \with \Cloc{s}{n}) 
\extAr{\eout{c!v}{\eta}} \Theta_N \testP (\Gamma_G \with \Cloc{\Theta_n}{n})$. 
Now note that $(\Delta_M \testP \Gamma_G \with \Cloc{\Theta_n}{n}) \lift{(\mathcal S)^e} 
(\Theta_N \testP \Gamma_G \with \Cloc{\Theta_n}{n})$

\item $m \notin \inp{\calM}, \Delta = \pdist{\calM}$ and 
$s \ar{c!v} \Theta_n$. Since $\calM \forsim \calN$ it follows that 
$m \notin \inp{\calN}$, hence by Proposition 
\ref{prop:wea.comp} we have that $\calN \testP (\Gamma_G \with \Cloc{s}{n}) 
\extAr{m.c?v} \pdist{\calN} \testP (\Gamma_G \with \Cloc{\Theta_n}{n})$. 
Now the hypothesis $\calM \sforsim \calN$ ensures that 
$\pdist{\calM} \testP (\Gamma_G \with \Cloc{\Theta_n}{n}) lift{({\mathcal S})^e} 
\pdist{\calN} \testP (\Gamma_G \with \Cloc{\Theta_n}{n})$.

\end{enumerate}

\noindent The cases $(\calM \testP \calG) \extar{\tau}$ and 
$(\calM \testP \calG) \extar{i.c?v}$ are 
treated similarly, and are therefore left as an 
exercise for the reader.
\end{proof}

However, since generating networks can be used to generate 
all the networks in $\nets$, Proposition  
\ref{prop:network.decomp}, we can 
easily generalise Theorem \ref{thm:single.comp} to arbitrary 
networks.

\begin{cor}[Compositionality: Theorem~\ref{thm:composition}]\label{cor:composition}
Let $\calM, \calN$ be finitary networks such that $
\inp{\calM} = \inp{\calN}$, $\outp{\calM} = \outp{\calN}$. 
Also, suppose $\calL$ is a network such that both 
$\calM \testP \calL$ and $\calN \testP \calL$ are defined. 
Then $\calM \forsim  \calN$ implies  
$(\calM \testP \calL)  \forsim (\calN \testP \calL)$.
\end{cor}
\begin{proof}
 First note that the only elements used to define simulations 
 are network interfaces, the set of outcomes of networks and the extensional 
 transitions that networks can perform. These definitions are preserved by the structural congruence between networks
defined on page \pageref{sec:properties}.
 As a  consequence, 
 simulations are identified up-to structural congruence; 
 if $\calM \equiv \calM'$, $\calM' \forsim \calN'$ and 
 $\calN' \equiv \calN$, it follows that $\calM \forsim \calN$.
 
 Suppose then that $\calM \forsim \calN$, and let $\calL$ be 
 a network such that both $(\calM \testP \calL)$ and $(\calN \testP 
 \calL)$ are defined. We show that 
 $(\calM \testP \calL) \forsim (\calN \testP \calL)$ by induction on 
 $\nodes{\calL}$.
 \begin{itemize}
 \item $\nodes{\calL} = 0$. In this case we have that 
 \[ 
 (\calM \testP \calL) \equiv \calM \forsim \calN \equiv (\calN \testP \calL),
 \]
 \item $\nodes{\calL} > 0$. By Proposition \ref{prop:network.decomp} 
there exist two networks $\calL'$ and $\calG$ such that 
$\calL \equiv (\calL' \testP \calG)$. 
In particular $(\calM \testP \calL) \equiv \calM \testP (\calL' \testP \calG) = 
(\calM \testP \calL') \testP \calG$, where we have used the associativity 
of the operator $\testP$, Proposition \ref{prop:testp.assoc}. 
Note that $\nodes{\calL'} < \nodes{\calL}$, hence by the inductive 
hypothesis it follows that $(\calM \testP \calL') \forsim (\calN \testP \calL')$. 
By Theorem \ref{thm:altchar} we obtain that $(\calM \testP \calL') \sforsim (\calN \testP \calL')$, 
and Theorem \ref{thm:single.comp} gives that 
$((\calM \testP \calL') \testP \calG) \sforsim (\calN \testP \calL') \testP \calG)$.
 Another application of \ref{thm:altchar} and the associativity of the operator $\testP$ lead to 
 $(\calM \testP( \calL' \testP \calG)) \forsim \calN \testP (\calL' \testP \calG)$. 
 Therefore we have that 
 \[
 (\calM \testP \calL) \equiv (\calM \testP( \calL' \testP \calG)) \forsim \calN \testP (\calL' \testP \calG)
 \equiv (\calN \testP \calL)
 \]
 as we wanted to prove.\qedhere
 \end{itemize}
\end{proof}

\subsection{Compositionality for deadlock simulations}
\label{sec:technical.must.comp}

As we have already done in Section \ref{sec:technical.may.comp},
here we also rely on an alternative characterisation of the 
$\deadsim$ preorder, which is more amenable to \emph{decomposition/composition}
results.

\begin{defi}[Simple Deadlock Simulation]
Let $\deadsim^s \subseteq \nets \times \subdist{\nets}$ be 
the largest relation such that whenever $\calM \deadsim^s \Theta$ 
then $\inp{\calM} = \inp{\calN}$, $\outp{\calM} = \outp{\calN}$ 
for any $\calN \in \support{\Theta}$; further 
\begin{itemize}
\item if $\calM \extAr{} \varepsilon$ then $\calN \extAr{}\varepsilon$, 
\item otherwise, 
\begin{enumerate}[label=(\roman*)]
	\item if $\delta(\calM) = \ttrue$ then $\Theta \extAr{} \Theta'$ 
	such that $\delta(\calN) = \ttrue$ for any $\calN \in \support{\Theta'}$, 
	\item if $\calM \extar{\lambda} \Delta'$, then $\Theta \extAr{\lambda} 
	\Theta'$ and $\Delta' \lift{\deadsim^s} \Theta'$.
\end{enumerate}
\end{itemize}
\end{defi}

\begin{thm}[Alternative Characterisation of Deadlock Simulations]
\label{thm:ds.altchar}
If $\calM$ is a finitary network and $\Theta$ is a finitary 
sub-distribution of networks then 
$\calM \deadsim \Theta$ if and only if $\calM \deadsim^s \Theta$.
\end{thm}

\begin{proof}
The proof is identical to that of Theorem \ref{thm:altchar}.
\end{proof}

\begin{thm}[Single Node Compositionality for Deadlock Simulations]
\label{thm:ds.single.comp}
Let $\calM$ be a network and $\Theta$ be a stable sub-distribution 
such that $\calM \deadsim^s \Theta$. Then, for any network 
$\calG$ such that $\calM \testP \calG$ and 
$\Theta \testP \pdist{\calG}$ are well-defined it follows 
that $(\calM \testP \calG) \deadsim (\Theta \testP \pdist{\calG})$. 
\end{thm}

\begin{proof}
The proof is analogous to that of Theorem \ref{thm:single.comp}.
There are, however, some extra statements that we need to 
check. 

\begin{itemize}
\item If $\calM \testP \calG \extAr{\tau} \varepsilon$ 
then $\calN \testP \calG \extAr{\tau} \varepsilon$. 
This can be proved by rewriting  
by inspecting the infinite 
sequence of $\tau$-moves which constitute the 
hyper-derivation $\calM \testP \calG \extAr{\tau} \varepsilon$ 
to build an hyper-derivation of the form $\calM \testP \calG \extAr{\tau} 
\varepsilon$. This step requires employing the composition and 
decomposition results developed in Section \ref{sec:comp.decomp}.

\item For any $\calM \in \nets, \calG \in \mathbb{G}$ such that 
$\calM \testP \calG$ is defined,  $\delta(\calM \testP \calG) = \ttrue$ 
implies $\delta(\calM) = \ttrue$ and $\delta(\calG) = \ttrue$,
\item For any stable sub-distribution $\Delta$ and network $\calG \in \mathbb{G}$, 
suppose that $\Delta \extAr{\tau} \Delta'$, where $\Delta'$ is such that 
$\delta(\calM) = \ttrue$ for any $\calM \in \support \Delta'$; further, 
suppose that  
$\calG \extAr{\tau} \Theta$ for some $\Theta$ such that 
$\delta(\calN) = \ttrue$ for any $\calN \in \support{\Theta}$. 
Then $\Delta \testP \pdist{\calG} \extAr{\tau} \Lambda$, and 
$\delta(\calL) = \ttrue$ for any $\calN \in \support{\Lambda}$.
\end{itemize}
The validity of the two statements above can be easily proved 
using the results developed in Section \ref{sec:comp.decomp}.
\end{proof}

\begin{cor}
\label{cor:ds.comp}[Theorem~\ref{thm:ds.comp}]
Let $\calM$ be a network and $\Theta$ be a stable sub-distribution 
such that $\calM \deadsim \Theta$. Then, for any network $\calN$ such 
that both $\calM \testP \calN$ and $\Theta \testP \calN$ are defined 
it follows that $(\calM \testP \calN) \deadsim (\Theta \testP \pdist{\calN})$.
\qed
\end{cor}
\begin{proof}
By induction on the number of nodes contained in 
$\nodes{\calN}$, using Theorem 
\ref{thm:ds.altchar}. The proof is analogous to that of Corollary \ref{cor:composition}.
\end{proof}

\subsection{Outcome preservation}
\label{sec:technical.result.pres}

The aim of this section is to prove Theorem~\ref{thm:results} and 
Theorem~\ref{thm:ds.outcomes}, namely 
that simulations preserve, in the sense of Definition~\ref{def:relsets}, 
the outcomes  generated by networks.
To this end, we first need two technical results.

\begin{lem}
\label{lem:sim.closure}
If $\calM \forsim \calN$, then $\closure{\calM} \forsim \closure{\calN}$.
\end{lem}

\begin{proof}
It suffices to show that the relation 
\[
{\mathcal S} = \{ (\closure{\calM}, \closure{\calN}) \;|\; \calM \forsim \calN\}
\]
satisfies the requirements of Definition \ref{def:sim}.
Let $\calM, \calN$ such that $\calM \forsim \calN$, 
Since the only possibility for networks of the form 
$\closure{\calM}, \closure{\calN}$ 
is that of performing $\tau$-extensional actions, we only have 
to check that if $\omega(\closure{\calM}) = \ttrue$ then 
$\closure{\calN} \extAr{\tau} \Theta$ for some $\Theta$ 
such that $\omega(\calL) = \ttrue$ for any $\calL \in \support{\Theta}$. 

Suppose then that $\omega(\closure{\calM}) = \ttrue$. It follows 
immediately that $\omega(\calM) = \ttrue$, and since 
by hypothesis $\calM \forsim \calN$, we have that $\calN 
\extAr{\tau} \Theta'$ for some $\Theta'$ such that 
$\omega(\calL') = \ttrue$ for any $\calL' \in \support{\Theta'}$. 
Therefore we have that $\closure{\calN} \extAr{\tau} \closure{\Theta'}$. 
It follows that $\support{\closure{\Theta'}} = \{\closure{\calL'} \;|\; 
\calL' \in \support{\Theta'}\}$, hence any network $\calL \in 
\support{\closure{\Theta'}}$ satisfies the predicate $\omega$.
\end{proof} 

\begin{lem}
 \label{lem:results}
  Let $\Delta, \Theta$ be stable distributions in  $\pLTSnets$ such that $\Delta \lift{(\forsim)^e} \Theta$; 
  then $\Theta \RedE{\;\;\;\;} \Theta'$ such that $\Val{\Delta} \leq \Val{\Theta'}$. 
\end{lem}
\begin{proof}
	Note that in the statement of Lemma \ref{lem:results} 
	we use the extreme derivative associated with the reduction relation 
  $\red$ defined for the testing structures associated with networks. 
  We have seen in Section \ref{sec:relating} that such extreme derivatives 
  do not coincide with weak $\tau$-actions in the extensional semantics, 
  except in the case of closed distributions.

	Therefore, let us first prove the statement for closed distributions, 
	that is those whose networks in the support have an empty interface. 
	Let then $\Delta, \Theta$ be such that $\type{\Delta} = \type{\Theta} = \emptyset$.
	
  We have two different cases.
  \begin{enumerate}[label=(\roman*)]
  \item First suppose $\Delta$ is a point distribution $\pdist{\calM}$. If the predicate $\omega(\calM)$ 
    is equal to $\ffalse$, $\Val{\pdist{\calM}} = 0$. 
    In this case, we recall that Theorem \ref{thm:hyper} (4) ensures that 
    there exists at least one extreme derivation $\Theta \RedE \Theta'$, for which 
    $0 \leq \Val{\Theta'}$ trivially holds. 
    
    Otherwise the predicate $\omega(\calM)$ is satisfied and $\Val{\pdist{\calM}}$ has to be 1. Since 
    $\pdist{\calM} \lift{\forsim}
    \Theta$ we know $\Theta \extAr{\;\;\;} \Theta'$, such that for all $\calN \in \Theta', 
    \omega(\calN) = \ttrue$. Also, since $\Theta$ is closed, the hyper-derivation 
    above can be rewritten as $\Theta \Red{\;\;\;} \Theta'$. This means that $\Val{\Theta'} = 1$; 
    moreover, as every state in $\support{\Theta'}$ is a successful state, we also have that 
    $\Theta \RedE{} \Theta'$.

\item
Otherwise $\Theta$ can be written as 
\begin{math}
  \sum_{\calM \in \support{\Delta}} \Delta(\calM) \cdot \Theta_{\calM}
\end{math}
where $\calM (\forsim^e) \Theta_{\calM}$ for each ${\calM}$ in the support of $\Delta$. 
By part (i) each $\Theta_{\calM} \RedE{} \Theta'_{\calM}$ such that $\Val{\pdist{\calM}} 
\leq \Val{\Theta'_{\calM}}$.
As an extreme derivative is also a hyper-derivative, we can combine these to obtain a hyper 
derivation for $\Theta$, using Theorem \ref{thm:hyper} (3). This leads to 
\[
\Theta = \sum_{\calM \in \support{\Delta}} \Delta(\calM) \cdot \Theta_{\calM} \dar{} 
 \sum_{\calM \in \support{\Delta}} \Theta'_{\calM} = \Theta'
\]
As for every $\calM \in \support{\Delta}$, $\calN \in \support{\Theta'_s}$ we have that 
$\calN \red$ implies $\omega{\calN} = \ttrue$, this condition is respected 
also by all states in $\support{\Theta'}$. Thus, the hyper derivation 
$\Theta \Red{} \Theta'$ is also an extreme derivation. 
Finally, the quantity $\Val{\Delta} = \sum \{ \Delta(\calM)\;|\; \omega(\calM) = \ttrue\}$ 
can be rewritten as $\sum_{\calM \in \support{\Delta}} \Val{\pdist{\calM}}$, leading to 
\[
 \Val{\Delta} = \sum_{\calM \in \support{\Delta}} \Val{\pdist{\calM}} \leq 
 \sum_{\calM \in \support{\Delta}} \Val{\Theta'_{\calM}} = \Val{\Theta'}.
\]
 
  \end{enumerate}
  
\noindent For more general distributions $\Delta, \Theta$ it suffices to note 
that if $\Delta \lift{(\forsim)^e} \Theta$ then $\closure{\Delta} 
\lift{(\forsim^e)} \closure{\Theta}$. For such (closed) distributions 
of networks we have that $\closure{\Theta} \RedE{} \Theta''$, 
and $\Val{\closure{\Delta}} \leq \Val{\Theta''}$. Now Corollary 
\ref{cor:extreme.reductions} ensures that $\Theta'' = \closure{\Theta'}$ 
for some $\Theta'$ such that $\Theta \RedE{} \Theta'$. 
Finally, by Corollary \ref{cor:closure.results} we obtain that 
\[
\Val{\Delta} = \Val{\closure{\Delta}} \leq \Val{\closure{\Theta'}} = \Val{\Theta'}
\]
which concludes the proof.
\end{proof}

\begin{cor}[Theorem~\ref{thm:results}]\label{cor:results}
   In $\pLTSnets$, $\Delta \lift{\forsim} \Theta$ implies
  $\Results{\Delta} \sqsubseteq_{H} \Results{\Theta}$.
\end{cor}
\begin{proof}
	We first prove the result for closed distributions $\Delta, \Theta$.
  Suppose $\Delta \RedE{} \Delta'$. We have to find  a derivation $\Theta \RedE{} \Theta'$ such
  that $\Val{\Delta'} \leq \Val{\Theta'}$. 
  Since we are assuming that $\Delta$ is closed, then $\Delta \RedE{} \Delta'$ implies 
  $\Delta \extArE{} \Delta'$, Corollary \ref{cor:extreme.reductions}, 
  which in turn gives $\Delta \extAr{\tau} \Delta'$.
  We can use the definition of $\forsim$ to find a 
        derivation $\Theta \extAr{\tau} \Theta''$ such that $\Delta' \lift{\forsim^e} \Theta''$. 
        Applying the previous lemma we obtain $\Theta'' \extArE{} \Theta'$ such that 
        $\Val{\Delta'} \leq \Val{\Theta'}$. By Theorem \ref{thm:hyper} we have that 
        $\Theta \extArE{} \Theta'$, and Corollary \ref{cor:extreme.reductions} 
        gives $\Theta \RedE{} \Theta'$. 

	Suppose now that $\Delta, \Theta$ are not closed distributions. 
	In this case Corollary \ref{cor:closure.results} ensures that 
	\[
	\Results{\Delta} = \Results{\closure{\Delta}} \sqsubseteq_{H} 
	\Results{\closure{\Theta}} = \Results{\Theta}
	\]
	\noindent and there is nothing left to prove. 
\end{proof}

\bigskip

We now repeat the above argument to show that outcomes are also preserved by
deadlock simulations; the details are quite similar. 
\begin{lem}
\label{lem:ds.closure}
Whenever $\calM \deadsim \Theta$ it follows that 
$\closure{\calM} \deadsim \closure{\Theta}$.
\end{lem}

\begin{proof}
If suffices to note that, for any network 
$\calM$, $\delta(\calM) = \ttrue$ 
if and only if $\delta(\closure{\calM}) = \ttrue$. 
Using this fact the result can be proved as in Lemma 
\ref{lem:sim.closure}.
\end{proof}

The main use of Lemma \ref{lem:ds.closure} is that of 
showing that whenever $\Delta \deadsim \Theta$ then 
the sets of outcomes of $\Delta$ and $\Theta$ are 
related in some appropriate manner. 

\begin{lem}  
\label{lem:ds.outcomes}
Suppose that $\Delta \lift{\deadsim} \Theta$ 
for some terminal distribution $\Delta$. 
Then $\Theta \RedE{} \Theta'$ for some 
$\Theta'$ such that $\Val{\Theta'} \leq 
\Val{\Delta}$.
\end{lem}

\proof
First suppose that $\Delta$ is a closed 
distribution, that is for any network $\calM 
\in \support{\Delta}$ we have that 
$\type{\Delta} = \emptyset$. 

\begin{enumerate}
\item if $\Delta = \pdist{\calM}$ we have two 
possible cases:
\begin{enumerate}[label=(\roman*)]
\item $\delta(\calM) = \ffalse$; since 
we are assuming that $\pdist{\calM}$ is a 
terminal distribution then it has to be 
$\omega{\calM} = \ttrue$, which implies 
$\Val{\pdist{\calM}} = 1$. In this case 
we are ensured that $\Theta \RedE{} \Theta'$ 
for some $\Theta'$, and $\Val{\Theta'} \leq 1$, 
trivially holds.
\item $\delta(\calM) = \ttrue$; in this 
case $\omega(\calM) = \ffalse$, 
hence $\Val{\pdist{\calM}} = 0$. Therefore we have 
to show that $\Theta \RedE{} \Theta'$ for some $\Theta'$ 
such that $\Val{\Theta'} = 0$. Since $\calM \deadsim \Theta$, 
we have that $\Theta \extAr{} \Theta'$ for some $\Theta'$ 
such that $\delta(\calN) = \ttrue$ for any $\calN \in \support{\Theta'}$, 
which is equivalent to $\Val{\Theta'} = 0$. 
Further, $\Theta \extAr{} \Theta'$ implies $\Theta \Red{} \Theta'$. 
Since $\calN \extar{\tau}\hspace{-10pt}\not\;\hspace{10pt}$, 
$\calN \extar{\eout{c!v}{\eta}}\hspace{-17pt}\not\;\hspace{17pt}$ 
for any $\calN \in \support{\Theta'}$, which is equivalent to 
$\calN \not\red$ for any $\calN \in \support{\Theta'}$, we 
also have $\Theta \RedE{} \Theta'$.
\end{enumerate}

\item otherwise, $\Delta = \sum_{i \in I} p_i \cdot \pdist{\calM_i}$, 
$\Theta = \sum_{i \in I} p_i \cdot \Theta_i$ with 
$\sum_{i \in I} p_i \leq 1$ and $\calM_i \deadsim \Theta_i$ 
for any $i \in I$. This part of the Lemma can be 
proved as in Lemma \ref{lem:results}.
\end{enumerate}
Now let $\Delta$ be a general distribution, 
not necessarily closed. 
By Lemma \ref{lem:ds.closure} it follows that 
$\closure{\Delta} \deadsim \closure{\Theta}$, 
therefore there exists a sub-distribution 
$\Theta'$ such that $\closure{\Theta} \RedE{} 
\closure{\Theta'}$, and $\Val{\closure{\Theta'}} 
\leq \Val{\closure{\Delta}}$. This in turn 
implies that $\Theta \RedE{} \Theta'$, and 
\[
\Val{\Theta'} = \Val{\closure{\Theta'}} \leq 
\Val{\closure{\Delta}} = \Val{\Delta}.\eqno{\qEd}
\]

\begin{cor}[Theorem~\ref{thm:ds.outcomes}]
\label{cor:ds.outcomes}
If $\Delta \lift{\deadsim} \Theta$ then 
$\Results{\Delta} \sqsubseteq_S \Results{\Theta}$. 
\end{cor}
\begin{proof}
First suppose that $\Delta$ is a closed sub-distribution; 
let $\Delta'$ be a sub-distribution such that 
$\Delta \RedE{} \Delta'$; we have to show that 
$\Theta \RedE{} \Theta'$ for some $\Theta'$ such that 
$\Val{\Theta'} \leq \Val{\Delta'}$.
Note that, since $\Delta$ is 
closed, this is equivalent to $\Delta \extArE{} \Delta'$. 
Further, it is straightforward to note that $\Delta'$ is 
terminal. Since $\Delta \lift{\deadsim} \Theta$ it follows 
that $\Theta \extAr{} \Theta''$ for some $\Theta''$ 
(which also implies $\Theta \Red{} \Theta''$)
such that $\Delta' \lift{\deadsim} \Theta''$. 
Since $\Delta'$ is terminal, by Lemma 
\ref{lem:ds.outcomes} we also have that 
$\Theta'' \RedE{} \Theta'$ and $\Val{\Theta'} \leq 
\Val{\Delta'}$. Now we just need to combine the 
reductions $\Theta \Red{} \Theta''$ and 
$\Theta'' \RedE{} \Theta'$ to obtain $\Theta \RedE{} \Theta'$.

If $\Delta$ is not a closed sub-distribution, we have 
that $\closure{\Delta} \lift{\deadsim} \closure{\Theta}$; 
since $\closure{\Delta}$ is closed it follows that 
$\Results{\Delta} = \Results{\closure{\Delta}} \sqsubseteq_{S} 
\Results{\closure{\Theta}} = \Results{\Theta}$.
\end{proof}

\section{Failure of Completeness}
\label{sec:completeness}

Although the simulation preorder $\forsim$ provides a proof
methodology for establishing 
 that two networks are be related via
the testing preorder $\Mayleq$, it is not complete.
 
That is, it is possible to find two networks $\calM, \calN$
such that $\calM \Mayleq \calN$ holds, but $\calM$ cannot be simulated
by $\calN$. 
Similarly, for the must-testing preorder, we have that 
it is possible to exhibit two networks $\calM, \calN$ 
such that $\calM \Mustleq \calN$, but $\pdist{\calM} \not\invdeadsim 
\calN$. 
This results are quite surprising, as simulation preorder has been
already proved to provide a characterisation of the may-testing
preorder for more standard process calculi such as pCSP, 
while the must-testing preorder has been proved to be characterised 
by failure simulations \cite{DGHM09full}. 
Here, for simplicity, we discuss the failure of completeness 
for the sole $\Mayleq$ preorder; however, the examples discussed 
here can be used to show that the relation $\Mustleq$ is also 
incomplete.
  
  The main problem that arises in our setting is that the  
  mathematical basis of simulation preorders rely on (full) probability distributions, which are a 
  suitable tool in a framework where a weak action from a process term has to be matched with the same 
  action performed by a distribution of processes.
  
  This is not true in our calculus; we have already shown that, due to
  the presence of local broadcast communication, it is possible to
  match a weak broadcast action with a sequence of outputs whose sets
  of target nodes are pairwise disjoint. This behaviour has been
  formalised by giving a non-standard  definition of weak extensional
  actions in Definition \ref{def:wea}. 
  
  Such a definition 
  captures the possibility of simulating a broadcast through a
  multicast only when the former action is performed with probability
  $1$.
  
  However, when comparing distributions of networks we have to
also match actions which are performed with probabilities less than 1, at least
informally; here the simulation of broadcast using multicast runs into problems,
as the following example shows.

  
  
  \begin{exa}
  \label{ex:compfail}
  \begin{figure}[t]
  \begin{align*}
     \begin{tikzpicture}
          \node[state](m){$m$}; 
          \node[state](o1)[above right=of m]{$o_1$}; 
          \node[state](o2)[below right=of m]{$o_2$};
 \path[to]
       (m) edge[thick] (o1)
       (m) edge[thick] (o2);
   \begin{pgfonlayer}{background}
    \node [background,fit=(m)] {};
    \end{pgfonlayer}
    \end{tikzpicture}
&&
    \begin{tikzpicture}
          \node[state](m){$m$}; 
          \node[state](n)[below=of m]{$n$};
          \node[state](o1)[right=of m]{$o_1$}; 
          \node[state](o2)[right=of n]{$o_2$};  
 \path[to]
       (m) edge[thick] (o1)
       (m) edge[thick] (n)
       (n) edge[thick] (o2);
   \begin{pgfonlayer}{background}
    \node [background,fit=(m) (n)] {};
    \end{pgfonlayer}
    \end{tikzpicture}
\\
\calM = \Gamma_M \with \Cloc{\tau.(c!\pc{v} \probc{0.9} \Cnil)}{m} 
&&
\calN = \Gamma_N \with \Cloc{c!\pc{v}}{m} \Cpar \Cloc{c?\pa{x}.(c!\pc{v} \probc{0.9} \Cnil)}{n}  \\  
&&   
\end{align*}

  \caption{Two testing related networks}
  \label{fig:compfail}
  \end{figure}
  
  Consider the two networks $\Gamma_M \with M$, $\Gamma_N \with N$ depicted in Figure \ref{fig:compfail}; 
  let 
  \begin{eqnarray*} 
  M &=& \Cloc{\tau.(c!\pc{v} \probc{0.9} \Cnil}{m}\\ 
  N &=& \Cloc {c!\pc{v}}{m} \Cpar \Cloc{c?\pa{x}.(c!\pc{x} \probc{0.9} \Cnil)}{n}
  \end{eqnarray*}
  \noindent 
  In $\Gamma_M \with M$ a message is 
  broadcast to nodes $o_1, o_2$ with probability $0.9$, while in $\Gamma_N \with N$ two different broadcasts happen in 
  sequence. The first broadcast, which can be detected by node $o_1$, happens with probability $1$. 
  The second broadcast, detectable by node $o_2$ happens with probability $0.9$. As a result,  
  the overall probability of message $v$ to be detected by both nodes $o_1, o_2$ is again $0.9$.
  
  We first show that $\Gamma_M \with M \Mayleq \Gamma_N \with N$, then we prove that 
  $\Gamma_M \with M \not\forsim \Gamma_N \with N$.\\ 
  For the first statement, we only supply informal details, as a complete proof would be 
  rather long and technical. 
  Consider a test $\Gamma_T \with T$, such that both 
  $(\Gamma_M \with M) \testP (\Gamma_T \with T)$ and $\Gamma_N \with N \testP 
  (\Gamma_T \with T)$ are defined. Without loss of generality, suppose that both $o_1, o_2 \in 
  \nodes{\Gamma_T \with T}$, that is $T \equiv \Cloc{t_1}{o_1} \Cpar \Cloc{t_2}{o_2} \Cpar T'$. 
  We consider only the most interesting case, that is when the testing component 
  reaches (with some probability $p$) an $\omega$-successful configuration after network 
  $\Gamma_M \with M$ broadcasts the message $v$. In this case, 
  a computation fragment of $(\Gamma_M \with M) \testP (\Gamma_T \with T)$ can 
  be summarised as follows:
  \begin{enumerate}
  	\item The testing component $\Gamma_T \with T$ performs some internal activity, thus leading to 
  	$\Gamma_T \with T \extAr{\tau} \Gamma_T \with \Cloc{\Lambda_1}{o_1} \Cpar \Cloc{\Lambda_2}{o_2} \Cpar \Lambda_T$,
  	\item At this point, the network $\Gamma_M \with M$ 
  	performs a $\tau$-extensional action, specifically $\Gamma_M \with M \extar{\tau} \Gamma_M \with 
  	\Delta$, where 
  	\begin{eqnarray*}
  	\Delta &=& 0.9 \cdot \pdist{M_1} + 0.1 \cdot \pdist{M_2}\\
  	M_1 &=& \pdist{\Cloc{c!\pc{v}}{m}}\\
  	M_2 &=& \pdist{\Cloc{\Cnil}{m}}
  	\end{eqnarray*}
  	
  \item The distribution $\Gamma_T \with \Cloc{\Lambda_1}{o_1} \Cpar \Cloc{\Lambda_2}{o_2} \Cpar \Lambda_T$ 
  performs some other internal activity, that is  
  \[
  \Gamma_T \with \Cloc{\Lambda_1}{o_1} \Cpar \Cloc{\Lambda_2}{o_2} \Cpar \Lambda_T \extAr{\tau} 
  \Gamma_T \with \Cloc{\Lambda'_1}{o_1} \Cpar \Cloc{\Lambda'_2}{o_2} \Cpar \Lambda'_T 
  \]
  \item At this point, the distribution 
 $\Gamma_M \with \Delta$ broadcasts the message $v$ with probability $0.9$, causing 
  	the testing component to evolve in $\Gamma_T \with \Cloc{\Lambda''_1}{o_1} \Cpar \Cloc{\Lambda''_2}{o_2} 
  	\Cpar \Lambda'_T$; note that only nodes $o_1$ and $o_2$ are affected by the broadcast performed by node $m$. 
  	After performing the broadcast, the tested network becomes deadlocked.
  \end{enumerate}
  	
	\noindent Consider now the network $(\Gamma_N \with N) \testP \Gamma_T \with T$. 
  	For such a network, a matching computation will proceed as follows:
  	\begin{enumerate}
  	\item The testing component $\Gamma_T \with T$ performs the two sequences of internal activities as before, ending 
  	up in the distribution $\Gamma_T \with \Cloc{\Lambda'_1}{o_1} \Cpar \Cloc{\Lambda'_2}{o_2} \Cpar \Lambda'_T$.
  	\item At this point, the network $\Gamma_N \with N$ performs a broadcast, $\Gamma_N \with N 
  	\extar{\tau} \Gamma_N \with \Theta$, where 
		\begin{eqnarray*} 
		\Theta &=& 0.9 \cdot \pdist{N_1} + \cdot \pdist{N_2}\\
		N_1 &=& \Cloc{\Cnil}{m} \Cpar \Cloc{c!\pc{v}}{n}\\
		N_2 &=& \Cloc{\Cnil}{m} \Cpar \Cloc{\Cnil}{n}
  	\end{eqnarray*}
  	Here note that, since the broadcast of message $v$ fired by the network $\Gamma_M \with M$ 
  	can be detected by the sole node $o_1$, only the code running at this node is affected 
  	in the test. Further, the resulting distribution at this node is again $\Lambda''_{1}$; 
  	the test component after the broadcast of message $v$ to node $o_1$ is then in the 
  	distribution 
  	$\Gamma_T \with \Cloc{\Lambda''_1}{o_1}  \Cpar \Cloc{\Lambda'_2}{o_2} \Cpar \Lambda'_T$.
    \item Before allowing the testing component to perform any
        activity, we require the distribution $\Gamma_N \with \Theta$ 
        to perform the second
        broadcast, which will be heard by node $o_2$; this 
        happens with probability $0.9$. Further, such a broadcasts 
        affects the probability distribution of processes running at the sole node
        $o_2$. Thus, after the second message has been broadcast by
        the tested network, the testing component will have the form
        $\Gamma_T \with \Cloc{\Lambda''_1}{o_1} \Cpar \Cloc{\Lambda''_2}{o_2}
        \Cpar \Lambda'_T$. This is exactly the same configuration
        obtained in the first experiment, after $\Gamma_M \with
        M$ has broadcast the message to both nodes $o_1, o_2$.  Further,
        note that the overall probability  $\Gamma_N \with N$  delivering the 
        message to both the external nodes is again $0.9$. 
        Finally, after the broadcast has been fired, the tested network reaches 
        a deadlocked configuration.
  \end{enumerate}
  
  \noindent We have shown that, whenever the broadcast of message
  $v$ by $\Gamma_M \with M$ affects the testing network $\Gamma_T
  \with T$ in some way, then $\Gamma_N \with N$ is able to
  multicast the message to both $o_1$ and $o_2$, causing $\Gamma_T
  \with T$ to behave in the same way. Note also that 
  $\inp{\Gamma_M \with M} = \inp{\Gamma_N \with N} = \emptyset$, 
  so that the behaviour of the testing component $\Gamma_T \with T$ 
  does not affect that of the tested networks. 
  Now the reader should 
  be convinced that $\Gamma_M \with M 
  \Mayleq \Gamma_N \with N$.
  
  Next we show that it is the case that $\Gamma_M \with M$ cannot
  be simulated by $\Gamma_N \with N$. The proof is obtained by
  contradiction. Suppose that $\Gamma_M \with M
  \forsim \Gamma_N \with N$. 
  Since $\Gamma_M \with M \extar{\tau} \Gamma_M \with \Delta$, 
  we have that $\Gamma_N \with N \extAr{\tau} \Gamma_N \with \Theta'$ 
  for some distribution $\Theta'$ such that $\Delta \lift{\forsim^e} \Theta'$. 
  It is straightforward to note that whenever $\Gamma_N \with N 
  \extAr{\tau} \Gamma_N \with \Theta'$ then $\Theta' = \pdist{\calN}$. 
  
  Recall that $\Delta = 0.9 \cdot \pdist{M_1} + 0.1 \cdot \pdist{M_2}$. 
  Since $\Gamma_M \with \Delta \lift{\forsim^e} \Gamma_N \with \pdist{N}$, 
  the decomposition property of lifted relations, 
  Definition \ref{def:lift} ensures that we can rewrite 
  $\pdist{N}$ as $0.9 \cdot \pdist{N} + 0.1 \cdot{\pdist{N}}$, 
  and $\Gamma_M \with M_1 \forsim \Gamma_N \with N$, 
  
  Let us focus on the network $\Gamma_M \with M_1$. 
  This network is equipped with the extensional transition 
  $\Gamma_M \with M \extar{\eout{c!v}{\{o_1,o_2\}}} \pdist{\Gamma_M \with \Cloc{c!\pc{v}}{m}}$. 
  Since $\Gamma_M \with M_1 \forsim \Gamma_N \with N$, it follows that 
  $\Gamma_N \with N \extAr{\eout{c!v}{\{o_1,o_2\}}} \Gamma_N \with \Theta''$ 
  for some distribution $\Theta''$. We show that this is not possible. 
  
  This is because the only action that can be performed by 
  $\Gamma_N \with N$ is $\Gamma_N \with N \extar{\eout{c!v}{\{o_1\}}} 
  \Gamma_N \with \Theta$; in order for $\Gamma_N \with N$ to be 
  able to perform the weak action $\extAr{\eout{c!v}{\{o_1,o_2\}}}$ 
  we require the distribution $\Gamma_N \with \Theta$ 
  to perform a weak broadcast to node $o_2$. 
  However, this is possible if every network in $\support{\Gamma_N \with 
  \Theta}$ can perform such an action; this is not true, since 
  $N_2 \in \support{\Theta}$, and the network $\Gamma_N \with N_2$ 
  is deadlocked. 
  
  We have shown that $\Gamma_M \with M_1 \not\forsim \Gamma_N \with N$, 
  which in turn gives $\Gamma_M \with \Delta \lift{\forsim^e}\hspace{-20pt}\not\;\hspace{20pt}
  \Gamma_M \with \pdist{N}$. Since $\Gamma_M \with M \extar{\tau} 
  \Gamma_M \with \Delta$, and $\Gamma_N \with \pdist{N}$ is 
  the only hyper-derivative of $\Gamma_N \with N$, we conclude 
  that $\Gamma_M \with M \not\forsim \Gamma_N \with N$. 
  \end{exa}
  
  Note that the example above can be readapted to show that 
  deadlock simulations are incomplete with respect to the 
  must-testing relations. 
  In fact, for the networks $\Gamma_M \with M$, $\Gamma_N \with N$ of 
  Example \ref{ex:compfail} it is easy to show that 
  $(\Gamma_N \with N) \Mustleq (\Gamma_M \with M)$, but 
  $(\Gamma_N \with N) \not\invdeadsim (\pdist{\Gamma_M \with M})$. 
  
 Example \ref{ex:compfail} has more serious consequences than just showing that simulation 
  preorder is not complete with respect to the may testing preorder. 
  One could in fact expect that the notion of simulation can be modified, 
  leading to a less discriminating preorder for networks which characterises 
  the $\Mayleq$ preorder. We show that this is not the case. 
  
  \begin{defi}[$\tau$-Simulations]
  A relation $\calR \subseteq \nets \times \dist{\nets}$ is 
  a $\tau$-simulation if whenever $\calM \calR \calN$ then 
  $\inp{\calM} = \inp{\calN}$, $\outp{\calM} = \outp{\calN}$ 
  and whenever $\calM \extar{\tau} \Delta$ it follows that 
  $\calN \extAr{\tau} \Theta$ for some $\Theta$ such that 
  $\Delta \lift{\calR} \Theta$. \qed
  \end{defi}
  Note that the definition of $\tau$-simulations is very general, 
  since the only constraints that we have placed on them, 
  apart from the standard checks on the input and output nodes 
  in the interface of networks, is that a strong $\tau$-action 
  has to be matched with a weak one. It follows at once 
  that $\forsim$ is a $\tau$-simulation.
    
  \begin{thm}
  \label{thm:nolift}
  There exists no $\tau$-simulation $\calR \subseteq \nets \times \dist{\nets}$ such that
  $\calM \Mayleq \calN$ iff $\calM \calR \pdist{\calN}$.
  \end{thm}
  
  \begin{proof}
   The proof is carried out by contradiction. 
   Suppose  $\calR \subseteq \nets \times \dist{\nets}$ 
   is a $\tau$-simulation such that $\calM \calR \pdist{\calN}$ 
   if and only if $\calM \Mayleq \calN$, and consider the 
   networks $\Gamma_M \with M, \Gamma_N \with N$ from Example \ref{ex:compfail}.
    We have already proved that $\Gamma_M \with M \Mayleq \Gamma_N \with N$ 
   and so by the hypothesis we have $\Gamma_M \with M \calR \Gamma_N \with N$. 
   Note that $\Gamma_M \with M \extar{\tau} \Gamma_M \with \Delta$, 
   where $\Delta = 0.9 \cdot \pdist{M_1} + 0.1 \cdot \pdist{M_2}$, 
   where $M_1, M_2$ have been already defined in Example \ref{ex:compfail}. 
         
   Since $\calR$ is a $\tau$-simulation, 
   the $\tau$-action performed by $\Gamma_M \with M$ has to 
   be matched by a hyper-derivation in $\Gamma_N \with N$; 
   we have already noted that the only possible hyper-derivation 
   for such a network is given by $\Gamma_N \with N \extAr{\tau} 
   \pdist{\Gamma_N \with N}$. Therefore we have that 
   $\Gamma_M \with \Delta \lift{\calR} \pdist{\Gamma_N \with N}$. 
   The decomposition property of lifted relations, Definition 
   \ref{def:lift} ensures that we can rewrite 
   $\pdist{\Gamma_N \with N}$ as $0.9 \cdot \Theta_1 + 
   0.1 \cdot \Theta_2$, and  
   $\Gamma_M \with M_i \calR \Gamma_N \with \Theta_i$, 
   $i=1,2$. It is trivial to note that here the only 
   possibility is that $\Theta_1 = \Theta_2 = \pdist{N}$. 
   Therefore $\Gamma_M \with M_1 \calR \pdist{\Gamma_N \with N}$, 
   and by hypothesis this implies that $\Gamma_M \with M_1 
   \Mayleq \Gamma_N \with N$. 
   
   However, this is not possible. We show that there is 
   a test that distinguishes the network $\Gamma_M \with M_1$ 
   from $\Gamma_N \with N$.
   Consider the test $\Gamma_T \with T$, 
   where $\Gamma_T$ is the connectivity graph consisting of 
   the sole node $o_2$ with no connections, while 
   $T = \Cloc{c?\pa{x}.\omega}{o_2}$. 
   It is straightforward to note that 
   $1 \in \Results{(\Gamma_M \with M_1) \testP (\Gamma_T \with T)}$. 
   However, for any $p \in \Results{(\Gamma_N \with N) \testP (\Gamma_T \with T)}$ 
   we have $p \leq 0.9$. If follows that 
   $\Results{(\Gamma_M \with M_1) \testP (\Gamma_T \with T)} \not\sqsubseteq_{H} 
   \Results{(\Gamma_N \with N) \testP (\Gamma_T \with T)}$. That is, 
   $\Gamma_M \with M_1 \not\Mayleq \Gamma_N \with N$. Contradiction. 
  \end{proof}

\section{Case Study: Probabilistic Routing}
\label{sec:prob.routing}
While our proof methods for relating probabilistic networks 
via the testing preorders are not complete, they are still 
useful for comparing practical examples of wireless networks. 
Even more, they can be used to perform a model-based verification
of network protocols, showing that their behaviour 
is consistent with respect to some formal specification. 
In this section we show how this can be done by proving the 
correct behaviour of a simple probabilistic routing protocol. 
For the sake of simplicity, we focus on an abstract implementation of 
a geographic routing protocol, in which much of the details are left 
unspecified. However, it is worth mentioning that 
the proposed implementation can be refined, leading to a concrete 
representation of the  \emph{SAMPLE} probabilistic routing protocol \cite{sample}.

By \emph{formal specification} we mean a network $\calM$, while by \emph{network protocol} 
we mean a set of networks $\mathscr{N}$ whose elements share the same input and output nodes.
Proving that the behaviour of a protocol $\mathscr{N}$ is sound with respect 
to a formal verification $\calM$ consists then in showing that 
for any network $\calN \in \mathscr{N}$ it has to be 
$\calM \simeq \calN$.

Let us now turn our attention on how this task can be achieved 
for a probabilistic (connection-less) routing protocol. 
At least intuitively, the routing policy states that 
messages broadcast by a location in a network, 
called source, are 
eventually delivered to a desired node, called 
destination. 
For the sake of simplicity, here we consider a 
situation in which the source and the destination 
of a routing policy are two fixed external nodes, 
$i$ and $o$ respectively. 

Designing a specification for the routing policy is 
easy; however, there are some details that need to be 
taken into account. 
First, we need to introduce some mathematical 
tools that will enable us to equip a node in a 
network with some sort internal memory; this is necessary, 
since in a routing protocol nodes have to store 
the values they have received and which they 
have not yet forwarded to another 
node.

This can be done by relying on \emph{multisets}. 
Roughly speaking, a multiset $\mathbb{A}$ 
is a set which can contain more than one copy of 
the same element. Formally, a multiset 
$\mathbb{A}$ from a set universe $U$ is 
a function $\mathbb{A}: U \rightarrow \mathbb{N}$ 
which assigns to each element $u \in U$ the 
number of copies of $u$ contained in $\mathbb{A}$. 
For our purpose the universe $U$ consists of the 
set of (closed) message values, and we only deal with 
finite multisets, that is those for which 
$\left(\sum_{v \in U} \mathbb{A}(v)\right) < \infty$.

We denote with $\emptyset$ the empty multiset, 
that is the multiset such that $\emptyset(v) = 0$ 
for any value $v$, and we say that 
$\mathbb{A} \subseteq \mathbb{B}$ if $\mathbb{A}(v) \leq 
\mathbb{B}(v)$ for any value $v$. 
We say that $v \in \mathbb{A}$ if $\mathbb{A}(v) > 0$. 
Given a finite collection of multisets $\mathbb{A}_1,\cdots, 
\mathbb{A}_n$, the multiset $(\bigcup_{i=1}^n \mathbb{A}_i)$ 
is defined by letting $(\bigcup_{i=1}^n \mathbb{A}_i)(v) = 
\sum_{i=1}^n \mathbb{A}_i(v)$. 

Finally, for any multiset $\mathbb{A}$ and a value $v$, we denote 
with $\mathbb{A} + v$ the multiset such that 
$(\mathbb{A} + v)(v) = \mathbb{A}(v) + 1$ and 
$(\mathbb{A} + v)(w) = \mathbb{A}(w)$ for any 
$w \neq v$. Similarly, the multiset $\mathbb{A} - v$ 
is defined by letting $(\mathbb{A} - v)(v) = 
\mathbb{A}(v) - 1$ if $\mathbb{A}(v) > 0$, 
$(\mathbb{A} - v)(v) = 0$ if $\mathbb{A}(v) = 0$ 
and $(\mathbb{A} - v)(w) = \mathbb{A}(w)$ for any 
$w \neq v$.

\begin{figure}
\begin{tikzpicture}
          \node[state](i){$i$}; 
          \node[state](m)[right=of i]{$m$}; 
          \node[state](o)[right=of m]{$o$};
 \path[to]
       (i) edge[thick] (m)
       (m) edge[thick] (o);
   \begin{pgfonlayer}{background}
    \node [background,fit=(m)] {};
    \end{pgfonlayer}
    \end{tikzpicture}
\caption{The specification $\Gamma_M \with M$ for the routing policy.}
\label{fig:routing.model}
\end{figure}

The second problem we need to tackle is that of ensuring 
that the specification we define for the routing policy 
is a finitary network. 
This is necessary because our proof techniques 
are valid only for such networks. As we will see, 
this can be accomplished by considering a more restricted 
routing policy, in which only a finite amount of messages 
will be routed from the source to the destination.

Let $k \geq 0$; the specification we propose for the connection-less 
routing policy of $k$ values is given by the network 
$\calM = \Gamma_M \with M^k_{\emptyset}$, where 
$\Gamma_M$ is the connectivity graph depicted 
in Figure \ref{fig:routing.model}
and $M^k_{\mathbb{A}}$ is a system term (parametrised 
by a multiset $\mathbb{A}$ and an integer $k \geq 0$) 
defined as 
\begin{eqnarray*}
M^k_{\mathbb{A}} &=& \Cloc{P^{k}_\mathbb{A}}{m}\\
P^0_{\mathbb{A}} &\Leftarrow& \sum_{v \in \mathbb{A}} c!\pc{v}.P^0_{\mathbb{A} - v}\\
P^{k+1}_{\mathbb{A}} &\Leftarrow& \left(\sum_{v \in \mathbb{A}} c!\pc{v}.P^{k+1}_{\mathbb{A} - v}\right) +  
\left(c?\pa{x}.P^{k}_{\mathbb{A} + x}\right)
\end{eqnarray*}\smallskip

\noindent Let us discuss the intuitive behaviour of a network of the form 
$\calM^k_{\mathbb{A}}$; at any given point, the internal node 
$m$ can either receive a message from node $i$, provided that 
there are still messages to be routed, or it can forward 
one of the messages in the multiset $\mathbb{A}$ to the 
output node $o$, if any. Note that we require the use 
of multisets since any value $v$ 
can be broadcast more than once by the input node $i$.

Formally, the behaviour of a network $\calM^k_{\mathbb{A}}$ 
can be described as follows. 
\begin{prop}
\label{prop:routing.model.transitions}
For any $k \geq 0$ and finite multiset $\mathbb{A}$ 
\begin{enumerate}
\item $\calM^k_{\mathbb{A}}$ is convergent and finitary,
\item $\delta(\calM^k_{\mathbb{A}}) = \ttrue$ if and 
only if $\mathbb{A} = \emptyset$,
\item if $k > 0$ then $v$ 
$\calM^k_{\mathbb{A}} \extar{i.c?v} \Delta$ if and only 
if $\Delta = \pdist{\calM^{k-1}_{\mathbb{A} + v}}$, 
\item if $k=0$ then $\calM^k_{\mathbb{A}} \extar{i.c?v} 
\Delta$ if and only if $\Delta = \pdist{\calM^k_{\mathbb{A}}}$, 
\item $\calM^{k}_{\mathbb{A}}
\extar{\eout{c!v}{\{o\}}} \Delta$ if 
and only if $v \in \mathbb{A}$ and $\Delta = \pdist{\calM^{k}_{\mathbb{A} - v}}$.\qed
\end{enumerate}
\end{prop}

\begin{figure}
     \begin{tikzpicture}
          \node[state](i){$i$}; 
          \node[state](n1)[right=of i]{$n_1$};
          \node(vdots)[right=of n1]{$\vdots$};
          \node[state](n2)[right=of vdots]{$n_2$}; 
          \node[state](o)[right=of n2]{$o$};
          \node[state](n3)[above=of vdots]{$n_3$};
          \node[state](nj)[below=of vdots]{$n_j$}; 
 \path[to]
       (i) edge [thick] (n1)
       (n2) edge[thick] (o)
       (n1) edge[thick] (n3)
       (n1) edge[thick] (vdots)
       (n1) edge[thick] (nj);
 \path[tofrom]
 				(n3) edge[thick] (n2)
 				(vdots) edge[thick] (n2)
 				(nj) edge[thick] (n2)
 				(n3) edge[thick] (vdots)
 				(nj) edge[thick] (vdots);   
   \begin{pgfonlayer}{background}
    \node [background,fit=(n1) (n2) (n3) (nj) (vdots)] {};
    \end{pgfonlayer}
    \end{tikzpicture}
\caption{The connectivity graph of the networks in the protocol 
$\mathscr{N}$.}
\label{fig:routing.impl.conngraph}
\end{figure}

Let us now define a protocol which is consistent with 
the specification $\calM^k_{\emptyset}$. 
As we already mentioned, a protocol is a collection of 
networks. 
We consider only the set of networks of the form 
$\calN^k_{\mathbb{A}} = \Gamma_N \with N^k_{\mathbb{A}}$ which satisfy the 
following conditions.

\begin{enumerate}
\item $\inp{\calN^k_{\mathbb{A}}} = \{i\}, \outp{\calN^k_{\mathbb{A}}} = \{o\}$,
\item $\nodes{\calN^k_{\mathbb{A}}} = \{n_1, n_2, \cdots, n_j\}$ for some 
$j \geq 2$, 
\item $\Gamma_N \vdash \rconn{i}{m}$ if and only if $m = n_1$, 
\item $\Gamma_N \vdash \rconn{m}{o}$ if and only if $m = n_2$,
\item $\Gamma_N \vdash \notrconn{n_h}{n_1}$ for any $h=1,\cdots, j$, 
\item for any node $h = 1,\cdots,j$ there exists a path 
from $n_h$ to $n_2$ in $\Gamma_N$,
\item for any node $n_h$, there exists a probability 
distribution $\Lambda_h \in \dist{\{1,\cdots,j\}}$ such that 
$\support{\Lambda_h} = \{h' \;|\; \Gamma_N \vdash \rconn{n_h}{n_{h'}}\}$, 
\item we assume a set of distinct channels $c_1, \cdots, c_j$ such that 
$c_h \neq c$ for any $h=1,\cdots, j$, 
\item The system term $N^k_{\mathbb{A}}$ is in the 
support of a distribution $\Delta^k_{\mathbb{A}}$, defined as 
\[
\Delta^k_{\mathbb{A}} = \interprP{\Cloc{Q^k_{\mathbb{A}_1}}{n_1} \Cpar \Cloc{R_{\mathbb{A}_2}}{n_2} 
\Cpar \prod_{h=3}^j \Cloc{S^h_{\mathbb{A}_h}}{n_h}}
\]
where $\left(\bigcup_{h=1}^j \mathbb{A}_h \right) = \mathbb{A}$ and 
\begin{eqnarray*}
Q^0_{\mathbb{A}} &\Leftarrow& \bigoplus_{h=1}^j \Lambda_1(j) 
\cdot \left(\sum_{v \in \mathbb{A}} c_h!\pc{v}.Q^0_{\mathbb{A} - v}\right)\\
Q^{k+1}_{\mathbb{A}} &\Leftarrow& \bigoplus_{h=1}^j \Lambda_1(h) \cdot 
\left[ c?\pa{x}.Q^{k}_{\mathbb{A} + x} + \left(\sum_{v \in \mathbb{A}} c_h!\pc{v}.Q^{k+1}_{\mathbb{A} - v}\right)\right]\\
R_{\mathbb{A}} &\Leftarrow& c_2?\pa{x}.R_{\mathbb{A} + x} + \left(\sum_{v \in \mathbb{A}} c!\pc{v}.R_{\mathbb{A}-v}\right)\\
S^h_{\mathbb{A}} &\Leftarrow& \bigoplus_{h'=1}^j \Lambda_{h}(h') \cdot 
\left[c_h?\pa{x}.S^h_{\mathbb{A} + x} + \left(\sum_{v \in \mathbb{A}} c_{h'}!\pc{v}.S^h_{\mathbb{A} - v}\right)\right] 
\end{eqnarray*}
Here we use $\bigoplus_{i=1^n}p_i \cdot s_i$ to denote the process such that 
\[
\interprP{\bigoplus_{i=1}^n p_i \cdot s_i} = \sum_{i=1^n} p_i \cdot \pdist{s_i}
\]
\end{enumerate}
\noindent
We denote with $\mathscr{N}^k_{\mathbb{A}}$ the set of networks 
$\calN^k_{\mathbb{A}}$ described above. The connectivity graph 
of such networks is depicted in Figure \ref{fig:routing.impl.conngraph}

\begin{remark}
Note that we committed an abuse of notation in defining the 
distribution $\Delta^k_{\mathbb{A}}$, by associating a 
process definition with a (probabilistic) process, rather than to 
a state. However, a process definition of the form $A \Leftarrow 
\bigoplus_{i=1}^n p_i \cdot s_i$ can be seen as the probabilistic 
process $\bigoplus_{i=1}^n p_i \cdot A_i$, where 
$A_i \Leftarrow s'_i$ and $s'_i$ is obtained from $s_i$ by replacing 
each occurrence of $A$ with $\bigoplus_{i=1}^n p_i \cdot A_i$.
\end{remark}
\noindent
Our aim is to show that for any $\calN \in \mathscr{N}^k_{\emptyset}$ 
we have that $\calM^k_{\emptyset} \simeq \calN$. 

Before supplying the details of the proof of the statement above, 
let us describe informally the behaviour of a distribution 
$\Delta \in \dist{\mathscr{N}^k_{\mathbb{A}}}$; we also discuss the 
requirements that we have placed on the structure of the 
connectivity graph $\Gamma_N$. 
In a distribution $\Delta \in \dist{\mathscr{N}^k_{\mathbb{A}}}$ 
a network is waiting to receive exactly $k$ messages from node $i$, 
and whose multiset of received messages which have been 
received but have not yet been forwarded to the external node $o$ 
is $\mathbb{A}$.
Note that we have placed many requirements in the definition of 
the connectivity graph of such networks; first we require that 
$i$ is their only input node, while $o$ is their only output node. 
This requirement is necessary, since to show that such 
networks are testing equivalent to the specification we 
have to ensure that they share with the latter the same sets 
of input and output nodes.

We require the connectivity graph of the networks in $\Delta$ to have 
a path from $n_h$ to $n_2$ for every $h=1,\cdots,j$. 
This condition is needed to ensure that messages detected by 
node $n_1$ (which in turn have been broadcast by $i$) can flow 
through the network until reaching node $n_2$, which in turn 
can broadcast the message to the output node $o$. 
As we will see this always happens with probability $1$.

The other constraints that we placed on the connectivity graph 
of $\calN \in \support{\Delta}$ are purely technical; we require that 
the only node connected to $i$ is $n_1$, while the only node 
connected to $o$ is $n_2$. As we will see when discussing the 
code running at $N^k_{\mathbb{A}}$, nodes $n_1$ and $n_2$ 
have the role of handling the values broadcast by $i$, and 
which have to be forwarded to $o$, respectively. 
We also require that $\Gamma_N \vdash \notrconn{n_h}{n_1}$ 
for any $h=1,\cdots,j$. This constraint ensures that all 
the messages received by $n_1$ have been broadcast by the 
input node $i$; note in fact that, in general, a node 
cannot detect the name of the node that fired 
a broadcast.

Let us now turn our attention to the code defined for 
the distribution 
$\Delta^k_{\mathbb{A}}$. Here we assume a set of 
channels $c_1,\cdots,c_j$;   
each node $n_h, h \neq 1$ can only detect messages 
broadcast along the channel $c_h$. Intuitively, when a 
message is broadcast along channel $c_h$ by a 
node $n_{h'}$, then it will be delivered to 
node $n_h$. In other words, node $n_{h'}$ has 
selected $n_h$ as the next hop in a routing path. 

We also assume a set of probability distributions 
$\Lambda_1,\cdots,\Lambda_j$. When a node $n_h, h \neq 2$ 
wishes to select the next hop in a routing path, 
it selects it according to the probability distribution 
$\Lambda_h$. Note that we require that $h' \in \support{\Lambda_h}$ 
if and only if $\calN^k_{\mathbb{A}} \vdash \rconn{n_h}{n_{h'}}$, 
that is a node can be selected as the next hop in a routing 
path by $n_h$ if and only if it is in the range of transmission 
of $n_h$. Further, any neighbour of $n_h$ can be selected 
as the next hop in a routing path. As we will see, this constraint 
ensures that, in unbounded time, a message stored in node $n_h$ 
will reach the node $n_2$ with probability $1$.

Any network distribution $\Delta \in \dist{\mathscr{N}^k_{\emptyset}}$ 
can be seen as a probability distribution of networks running a 
(connection-less) routing algorithm of $k$ messages. 
Such an algorithm is designed by letting any node $n_h$, with the exception of 
$n_2$, to select the next hop 
in a routing path probabilistically among its neighbours. 
For node $n_2$, the message is broadcast along channel $c$ 
with probability $1$, thus forwarding it to the only output node 
$o$.
Also, the message to be forwarded to a next-hop in a 
routing path by node $n_h$ is selected non-deterministically 
among those stored in such a node, that is the nodes in 
the multiset $\mathbb{A}_h$. 


Roughly speaking, the behaviour of a network $\Delta \in \
\dist{\mathscr{N}^k_{\mathbb{A}}}$ 
can be described as follows:
\begin{enumerate}
\item node $n_1$ can receive a message $v$ broadcast by node 
$i$ along channel $c$, provided $k \geq 0$. Then it stores it 
in the multiset associated to it,
\item At any given point, any node $n_h, h \neq 2$ 
can select the next hop in a routing path among its neighbours. 
Then it selects the message to be forwarded non-deterministically 
among those stored in its internal multiset
\item At any given point, node $n_2$ can broadcast one of 
the messages stored in its multiset along channel $c$. 
This broadcast is detected by the output node $o$.
\end{enumerate}

\noindent The behaviour of a network $\calN \in \mathscr{N}^k_{\emptyset}$ 
is similar, with the only exception that the first time 
each node receives a message, the next-hop of a routing 
path it chooses is fixed.

Let us now turn our attention to the extensional 
transitions performed by a network $\calN \in \mathscr{N}^k_{\mathbb{A}}$, 
and more generally by a distribution $\Delta \in \dist{\mathscr{N}^k_{\mathbb{A}}}$.
To this end it is useful to introduce 
some notation. 
First we define the (state based) processes
\begin{eqnarray*}
q^{0,h}_{\mathbb{A}} &=& \sum_{v \in \mathbb{A}} c_h!\pc{v}.Q^0_{\mathbb{A} - v}\\
q^{k+1,h}_{\mathbb{A}} &=& c?\pa{x}.Q^{k}_{\mathbb{A} + x} + \left(\sum_{v \in \mathbb{A}} c_j!\pc{v}.Q^{k+1}_{\mathbb{A} - v}\right)\\
s^{h,h'}_{\mathbb{A}} &=& c_h?\pa{x}.S^h_{\mathbb{A} + x} + 
\left(\sum_{v \in \mathbb{A}} c_{h'}!\pc{v}.S^h_{\mathbb{A} - v}\right) 
\end{eqnarray*}
\noindent
and we note that any network $\calN \in \mathscr{N}^k_{\mathbb{A}}$ has the form 
\[
\calN = \Gamma_{N} \with \Cloc{q^{k,h}_{\mathbb{A}_1}} \Cpar \Cloc{R_{\mathbb{A}_2}} 
\Cpar \prod_{h = 3}^j \Cloc{s^{h,h'}_{\mathbb{A}_h}}{n_h} 
\]
\noindent
where $(\bigcup_{h=1}^{j} \mathbb{A}_h) = \mathbb{A}$. 
For such networks, we define $\mbox{Values}_{\calN}(h) = \mathbb{A}_h$. 
Intuitively, this function returns the multiset of values stored at 
node $n_h$ in the network $\calN$.

Finally, let $(\calN^k_{\mathbb{A}})^{x}$ be the unique 
network such that $\mbox{Values}_{(\calN^k_{\mathbb{A}})^{x}}(2) = \mathbb{A}$.
This is the network where all the nodes that have to be routed are 
stored in the node $n_2$; therefore they are ready to be forwarded to 
the destination $o$.

We are now ready to characterise the set of strong extensional 
transitions performed by any network $\calN \in \mathscr{N}^k_{\mathbb{A}}$.

\begin{prop}
\label{prop:routing.impl.strong}
For any network $\calN \in \mathscr{N}^k_{\mathbb{A}}$, 
\begin{enumerate} 
\item 
\begin{enumerate}[label=(\roman*)]
\item $\calN \extar{\tau}\hspace{-10pt}\not\;\hspace{10pt}$ 
iff $\calN = (\calN^{k}_{\mathbb{A}})^{x}$, 
\item otherwise $\calN \extar{\tau} \Delta$ for some 
$\Delta \in \dist{\mathscr{N}^k_{\mathbb{A}}}$, 
\end{enumerate}
\item conversely, whenever $\calN \extar{\tau} \Delta$ 
then $\Delta \in \dist{\mathscr{N}^{k}_{\mathbb{A}}}$,
\item if $k > 0$ then 
\begin{enumerate}[label=(\roman*)]
\item $\calN \extar{i.c?v} \Delta$ 
for some $\Delta \in \dist{\mathscr{N}^{k-1}_{\mathbb{A} + v}}$, 
\item conversely, whenever $\calN \extar{i.c?v} \Delta$ then 
$\Delta \in \dist{\mathscr{N}^{k-1}_{\mathbb{A} + v}}$, 
\end{enumerate}
\item if $k = 0$ then $\calN \extar{i.c?v} \Delta$ iff 
$\Delta = \pdist{\calN}$,
\item for any $v \in \mbox{Values}_{\calN}(2)$ then 
$\calN \extar{\eout{c!v}{\{o\}}} \Delta$ for some 
$\Delta \in \dist{\mathscr{N}^k_{\mathbb{A} - v}}$,
\item conversely, whenever $\calN \extar{\eout{c!v}{\{o\}}} 
\Delta$ then $v \in \mbox{Values}_{\calN}(2)$ and 
$\Delta \in \dist{\mathscr{N}^k_{\mathbb{A} - v}}$.\qed
\end{enumerate}
\end{prop}

However, in order to show that the specification 
$\calM^k_{\emptyset}$ is testing equivalent to any network 
in $\mathscr{N}^k_{\emptyset}$ we have to characterise 
also the set of weak extensional actions performed 
by such networks. To this end, we first analyse 
the structure of any $\tau$-extensional transition performed 
by any distribution $\Delta \in \dist{\mathscr{N}^k_{\mathbb{A}}}$. 

\begin{prop}
\label{prop:stop.reachable}
Let $k \geq 0$, $\mathbb{A}$ be a finite multiset and 
suppose $\Delta \in \dist{\mathscr{N}^k_{\mathbb{A}}}$. 
\begin{enumerate}
\item $\Delta \extAr{\tau} \pdist{(\calN^k_{\mathbb{A}})^{x}}$,
\item whenever $\Delta \extAr{\tau} \Delta'$, then 
$\Delta' \extAr{\tau} \pdist{(\calN^k_{\mathbb{A}})^x}$.
\end{enumerate}
\end{prop}

\begin{proof}[Outline of the Proof]
First note that, for any $\Delta \in \dist{\mathscr{N}^k_{\mathbb{A}}}$ 
Proposition \ref{prop:routing.impl.strong} 
ensures that $\Delta \extAr{\tau} \Delta'$ 
implies $\Delta' \in \subdist{\mathscr{N}^k_{\mathbb{A}}}$. 

Let us focus on the proof of the first statement.
Let $\Delta \in \dist{\mathscr{N}^k_{\mathbb{A}}}$ 
for some $k \geq 0$ and finite multiset $\mathbb{A}$. 
We actually prove a stronger statement than (1), 
that is that $\Delta \extArE{} \pdist{(\calN^k_{\mathbb{A}})^{x}}$. 
First note that Theorem \ref{thm:hyper} (4) 
ensures that there exists a sub-distribution 
$\Theta$ such that $\Delta \extArE{} \Theta$. 
Such a distribution $\Theta$ has to be an element 
of the set $\subdist{\mathscr{N}^k_{\mathbb{A}}}$; 
further, any state in its support should not be 
able to perform an extensional $\tau$-action. 
It follows from Proposition \ref{prop:routing.impl.strong} 
that the only possibility is that $\support{\Theta} \supseteq 
\{(\calN^k_{\mathbb{A}})^{x}\}$, or equivalently that 
$\Theta = p \cdot \pdist{(\calN^k_{\mathbb{A}})^{x}}$ 
for some $0 \leq p \leq 1$. It remains to prove that 
$p = 1$.

This follows because the probability distribution used 
by any node $n_h, h \neq 2$, to select the next-hop in a 
routing path is defined so that any neighbour of 
$n_h$ can be chosen with probability strictly greater 
than $0$; in particular, since we are 
assuming that there exists a path from node $n_h$ to node $n_2$, 
a node $n_{h'}$ whose distance to $n_2$ is less than the 
distance between $n_h$ and $n_2$ can be selected with 
non-negligible probability. As a consequence, in 
the long run the average distance between the node 
where a message $v \in \mathbb{A}$ is stored and 
the node $n_2$ decreases to $0$; that is, 
with probability $1$ message $v$ is stored in 
the node $n_2$. 
Since this line of reasoning is independent from 
the value $v$, we also have that in the long run 
any message in $\mathbb{A}$ 
will be stored in $n_2$ with probability $1$; 
formally, $\Theta = 1 \cdot \pdist{(\calN^{k}_{\mathbb{A}})^{x}}$.

Now statement (2) follows trivially. Whenever $\Delta \extAr{\tau} \Delta'$ 
we have that $\Delta' \in \dist{\mathscr{N}^k_{\mathbb{A}}}$, and 
by (1) above it follows that $\Delta' \extArE{} \pdist{\calN^{k}_{\mathbb{A}}}$. 
\end{proof}

\begin{cor}
\label{cor:impl.convergent}
Any $\Delta \in \dist{\mathscr{N}^k_{\mathbb{A}}}$ is 
convergent.
\end{cor}

\begin{proof}
Suppose $\Delta \extAr{\tau} \Delta'$ for some $\Delta'$; 
then $\Delta' \extAr{\tau} \pdist{(\calN^k_{\mathbb{A}})}^{x}$; 
it follows that $ \size{\Delta'} \geq 1$, 
hence $\size{\Delta'} = 1$. As a consequence, 
for no network $\calN \in \support{\Delta}$ we have 
$\calN \extAr{} \varepsilon$.
\end{proof}

The last step that we need to take is that of 
characterising the set of (weak) input and output 
transitions for any distribution $\Delta \in 
\dist{\mathscr{N}^k_{\mathbb{A}}}$. This 
can be done by using both propositions 
\ref{prop:routing.impl.strong} and 
\ref{prop:stop.reachable}.

\begin{prop}
\label{prop:routing.impl.weak}
Let $k \geq 0$ and $\mathbb{A}$ be a multiset. 
Then for any distribution $\Delta \in \dist{\mathscr{N}^k_{\mathbb{A}}}$,

\begin{enumerate}[label=(\roman*)]
\item $\Delta \extAr{\tau} \Delta'$ with 
$\delta(\calN) = \ttrue$ for any $\calN \in \support{\Delta'}$ 
if and only if $\mathbb{A} = \emptyset$, 
\item if $k > 0$ then 
\begin{enumerate}[label=(\roman*)] 
\item $\Delta \extAr{i.c?v} \Delta'$ for some 
$\Delta'$ for some $\Delta' \in \dist{\mathscr{N}^{k-1}_{\mathbb{A} + v}}$, 
\item conversely, whenever $\Delta \extAr{i.c?v} \Delta'$ then 
$\Delta' \dist{\mathscr{N}^{k-1}_{\mathbb{A} + v}}$,
\end{enumerate}

\item if $k=0$ then 
\begin{enumerate}[label=(\roman*)]
\item $\Delta \extAr{i.c?v} \Delta'$ for some 
$\Delta' \in \dist{\mathscr{N}^{0}_{\mathbb{A}}}$, 
\item whenever $\Delta \extAr{i.c?v} \Delta'$ then 
$\Delta' \in \dist{\mathscr{N}^{0}_{\mathbb{A}}}$, 
\end{enumerate}

\item if $\mathbb{A} \neq \emptyset$, 
\begin{enumerate}
\item $\Delta \extAr{\eout{c!v}{\{o\}}} \Delta'$ for 
some $\Delta' \in \dist{\mathscr{N}^{k}_{\mathbb{A} - v}}$, 
\item conversely, whenever $\Delta \extAr{\eout{c!v}{\{o\}}} 
\Delta'$ it follows that $\Delta' \in \dist{\mathscr{N}^{k}_{\mathbb{A} - v}}$.
\qed
\end{enumerate}
\end{enumerate}
\end{prop}

\noindent We are now ready to show that the protocol $\mathscr{N}^k_{\mathbb{A}}$ 
satisfies the specification $\calM^k_{\mathbb{A}}$.
\begin{thm}
\label{thm:routing.equivalence}
For any $k \geq 0$ and $\calN \in \mathscr{N}^k_{\emptyset}$ 
we have that $\calM^k_{\emptyset} \simeq \calN^k_{\emptyset}$. 
\end{thm}

\begin{proof}
Let $k \geq 0$ and $\mathbb{A}$ be a finite multiset. 
We have already noted that the network $\calM^k_{\mathbb{A}}$  
is finitary. Further, it is easy to show that any network 
$\calN \in \mathscr{N}^k_{\mathbb{A}}$ is finite state, 
and by Corollary \ref{cor:impl.convergent} it follows that 
it is also finitary.

Therefore, it suffices to show that for any $\calN \in \mathscr{N}^k_{\emptyset}$ 
we have both $\calM^k_{\emptyset} \dfdeadsim \calN$ and 
$\calN \dfdeadsim \calM^k_{\emptyset}$. Theorem 
\ref{thm:must.dfsound} gives that $\calM^k_{\emptyset} \Musteq 
\calN$, while Theorem \ref{thm:may.dfsound} ensures that 
$\calM^k_{\emptyset} \Mayeq \calN$. 

In fact we prove a stronger statement. 
For any $k \geq 0$, finite multiset $\mathbb{A}$ 
and network $\calN \in \mathscr{N}^k_{\mathbb{A}}$ 
we have that $\calM^k_{\mathbb{A}} \dfdeadsim \calN$, 
and conversely $\calN \dfdeadsim \calN^k_{\mathbb{A}}$.
Theorem \ref{thm:routing.equivalence} follows by 
letting $\mathbb{A} = \emptyset$.

To this end, consider the relation
\[
{\mathcal S} = \{(\calM^k_{\mathbb{A}}, \calN') \;|\; \calN \in \mathscr{N}^k_{\mathbb{A}}\}
\]
\noindent
We show that this relation satisfies the requirements of 
Definition \ref{def:dfdeadsim}. 
First suppose that $\delta(\calM^k_{\mathbb{A}}) = \ttrue$. 
Then $\mathbb{A} = \emptyset$ by Proposition 
\ref{prop:routing.model.transitions}, and
Proposition \ref{prop:routing.impl.weak} ensures that 
$\pdist{\calN^k_{\emptyset}} \extAr{\tau} \Theta$ 
for some $\Theta$ such that for any $\calN' \in \support{\Theta}$ 
we have that $\delta(\calN') = \ttrue$. 

Now, suppose that $\calM^k_{\mathbb{A}} \extar{i.c?v} \Delta$. 
By Proposition \ref{prop:routing.model.transitions} we have 
two possible cases: 

\begin{enumerate}
\item $k = 0$; in this case $\Delta = \pdist{\calM^k_{\mathbb{A}}}$; 
by Proposition \ref{prop:routing.impl.weak} we have 
that $\calN \extAr{i.c?v} \Theta$ for some $\Theta \in 
\dist{\mathscr{N}^k_{\mathbb{A}}}$, and trivially 
$\calM^k_{\mathbb{A}} \;{\mathcal S}^e\; \Theta$. 
\item $k \geq 0$; here $\Delta = \pdist{\calM^{k-1}_{\mathbb{A}}}$. 
The action $\calM^k_{\mathbb{A}} \extar{i.c?v} \pdist{\calM^{k-1}_{\mathbb{A}+v}}$ 
can be matched by $\calN \extAr{i.c?v} \Theta$, 
where $\Theta \in \dist{\mathscr{N}^{k-1}_{\mathbb{A} + v}}$, 
again using Proposition \ref{prop:routing.impl.weak}.
\end{enumerate}

\noindent The last case we need to check is $\calM \extar{\eout{c!v}{\{o\}}} \Delta$. 
This case is handled in the same way of the previous ones, again using 
Propositions \ref{prop:routing.model.transitions} 
and \ref{prop:routing.impl.weak}.

For the opposite implication, $\calN \dfdeadsim \calM^k$,
it is sufficient to consider the converse relation ${\mathcal S}^{-1}$, 
showing that it satisfies the requirements of Definition \ref{def:dfdeadsim}.
The proof is similar to the one above, this time by using 
Proposition \ref{prop:routing.impl.strong} to infer the 
structure of an extensional action of the form 
$\calN \extar{\lambda} \Theta$, and by matching 
it with an action performed by $\calM^k_{\mathbb{A}}$ 
according to Proposition \ref{prop:routing.model.transitions}. 
Here it is important to note that every action of the 
form $\calM^k_{\mathbb{A}} \extar{\lambda} \Delta$ is also a  
weak action, that is $\calM^k_{\mathbb{A}} \extAr{\lambda} \Delta$, 
and that a strong $\tau$-extensional action of the form 
$\calN \extar{\tau} \Theta$ can be matched by the weak 
action $\calM^k_{\mathbb{A}} \extAr{\tau} \pdist{\calM^k_{\mathbb{A}}}$.
\end{proof}

\section{Conclusions}
\label{sec:conclusions}

In this paper we have developed a calculus for wireless systems, which enjoys both probabilistic behaviour 
and local broadcast communication. We have developed a theory based on the probabilistic testing preorders, 
and provided sound proof methods for finitary networks to prove that they can be related via our behavioural preorders.
We have applied our proof techniques to check that a probabilistic routing protocol is consistent 
with a given specification.

While testing theories have been analysed for process algebras \cite{Ene02} with broadcast communication 
over a flat topology, we believe that this is the first work that considers testing theories for 
a calculus which enjoys local broadcast communication.

In the past the development of formal tools for wireless networks has focused either on other forms of 
behavioural theories (such as variants of weak bisimulation) and the analysis of protocols. Here we give a brief 
review of the main works which have inspired our calculus.

To the best of our knowledge, the first paper describing a process calculus for broadcasting systems, \textbf{CBS}, is 
\cite{Prasad95}. In this paper the author presents a simple process calculus 
in which a synchronisation between a sender and a receiver is modelled as an 
output action, rather than an internal activity as in standard 
process calculi such as CCS. This allows multiple receivers to detect 
a message sent by a sender, thus implementing broadcast communication. 
In \cite{HenRat98} different notions of barbed congruence for a variant of CBS are introduced; 
these correspond to strong barbed congruence and weak barbed congruence. For each of 
them, a characterisation result in terms of strong and weak bisimulation, 
respectively, is proved. 

Another calculus to model broadcast systems, known as the $b\pi$-calculus 
and inspired by both CBS and the $\pi$-calculus \cite{pibook}, is introduced in \cite{Ene01}; 
as the author points out, broadcast communication is modelled in the same style of 
CBS. In this paper the authors define three different behavioural equivalences, 
corresponding to barbed congruence, step equivalence and labelled bisimilarity. 
The author proves that such behavioural equivalences coincide. 

In \cite{Ene02} the authors define both the may and must testing preorders for processes of the 
$b\pi$-calculus, and they prove a characterisation result for each of them. 
The main contribution here lies in the characterisation of the must-testing 
preorder; as the authors point out, in fact, broadcast communication leads 
to a non-standard characterisation of the latter. In particular, the non-blocking 
nature of broadcast actions does not allow acceptance sets to be used in their 
characterisation result.

In the last decade, broadcast calculi have been modified in several ways by 
equipping processes with a topological structure, thus modelling wireless 
networks; the idea is that 
of representing a process as a set of locations, running different code 
for broadcasting and receiving messages; the topology defined for a process 
establishes how communication is modelled, for example by letting only some 
locations being able to detect the messages broadcast at another one.

In \cite{NanHan06} the authors propose to model the topological 
structure of a network by using a connectivity graph; a process 
is viewed as a set of locations running code, while a graph 
whose vertices are locations 
is used to determine how communication is carried out. Intuitively, 
a transmission originated at a given location can only be detected 
by those vertices which are connected to the former. The transition 
relation of processes is defined as parametric in a connectivity graph. 
This framework has been proposed by the authors as a basis for the 
analysis of security protocols in wireless networks.

In \cite{nanz} an allocation environment is used to represent the 
topological structure of a wireless networks. A wireless network 
is intended as a parallel composition of processes, each of which 
is associated with a set of locations to which the process belongs 
and a probability distribution 
over locations; intuitively, the latter describes the probability 
with which a message broadcast by the process is detected at a 
given location.

In \cite{restrbroad} the authors propose a \emph{restricted 
broadcast process theory} to model wireless networks. 
Here a network consists of a parallel composition 
of different processes; each process is associated with a location name, 
and a function between locations to sets of locations 
is used to represent the network topology. The authors 
propose the standard notion of weak bisimulation as 
the behavioural equivalence to be used to relate networks 
and they show a case study in which they prove the correctness 
of a routing protocol.

In \cite{GwFM10} an extension of the restricted broadcast process theory, 
the \emph{Computed Network Theory}, is proposed; here the expressive power of a network 
is augmented through different operators. 
For the resulting calculus, a variant of strong bisimulation is 
defined and proved to be a congruence. The main result in the paper is a sound 
axiomatisation of the strong bisimulation, thus 
enabling equational reasoning for wireless networks. The authors 
also show that the proposed axiomatisation is complete in a setting 
where only non-recursive networks are considered.
The \emph{Computed Network Theory} framework is also used in 
\cite{wfmodcheck} to check properties of mobile networks; 
the authors show how both the equational theory and model checking 
can be used to verify the correctness of a routing protocol.

In \cite{omegacalc} the authors view a network as a collection of 
processes, each of which is associated with one or more groups. 
Processes which belong to the same group are assumed to be 
neighbours; as a consequence, a broadcast performed by a 
process can be detected by all the processes which belong 
to at least one group of the broadcaster. 
The authors show that in their framework state reachability is 
a decidable problem; further, they introduce different notions 
of behavioural equivalences, based on late bisimilarity and its weak 
variant, and they show that such equivalences are in fact congruences.
Finally, they apply their calculus by formalising and analysing the 
behaviour of a leader election protocol and 
a routing protocol.

In \cite{LaneseS10} the authors describe wireless 
networks by using metric spaces; they assume that a network consists 
of a set of processes, each of which has an associated location 
and a radius of transmission; a metric distance over the set of locations 
is assumed to determine how communication is modelled. The authors 
describe the behaviour of a wireless network in terms of both a 
reduction semantics and a labelled transition semantics. These 
two semantics are proved to be equivalent up-to a notion of 
structural congruence. We remark that in their paper the authors 
assume that a communication between two stations consists of 
two phases, one for the beginning and one for termination. 
This allows the authors to model collision-prone communication. 

In \cite{wang} the authors present an extension of the calculus described above, 
in which node mobility and timed communication are introduced. The 
authors give both a reduction semantics and a labelled transition semantics, 
and they prove that they are congruent up-to structural congruence.

Another calculus for wireless networks in which collision-prone behaviour is taken 
into account is described in \cite{merro}. In their work, the authors 
describe a network as a set of processes running in parallel, each of 
which has a location name and a semantic tag associated with it; the latter 
consists of a set of locations names and it corresponds to the set 
of locations which can detect messages broadcast by the process. 
The calculus includes a notion of discrete-time, in the style of 
\cite{HenReg95}, and broadcasts of messages start and end at different 
time slots. The authors develop a notion of barbed congruence for 
wireless systems and they 
propose a sound, but not complete, characterisation result in terms 
of weak bisimulation. 

A variant of this calculus which considers only networks with 
flat topology is presented in \cite{CHM12}. Here the 
authors develop a notion of reduction barbed congruence for their 
calculus; they 
also introduce an extensional semantics whose induced weak 
bisimulation principle is proved to be sound and complete 
with respect to the barbed congruence.

In \cite{GodSon} the authors propose a model in which the topological 
structure of a network is represented as a graph whose vertices are 
locations; further, they assign 
to each edge in the graph a (possibly unknown) probability as a likelihood estimate of 
whether a message broadcast by a location at the starting end-point of an 
edge will be delivered to the location at the terminal end-point of the same. 
The proposed model also allows the network topology of a system to 
evolve according to a probabilistic mobility function. 
The authors prove that, in the proposed calculus, the logical equivalence 
defined over a variant of PCTL coincides with weak bisimulation. 

In \cite{songphd} several models for modelling probabilistic ad hoc networks 
are developed; the author first defines a probabilistic process calculus 
where connections between nodes are probabilistic. Behavioural theories 
based on bisimulation and temporal logics are defined for analysing the properties 
of networks in such a calculus. 
The presented calculus is then extended in order to model different features of 
wireless networks, such as exponential time delays and changes in the network topology. 

In \cite{LanotteM11} the authors define a language for wireless networks 
in which the code running at network locations contains both non-probabilistic 
and non-deterministic behaviour. The topological structure of a network 
is defined in the same way of \cite{merro}; the authors introduce a 
notion of simulation, parametrised in a probability value, in order 
to capture the concept of two networks exhibiting the same behaviour 
up-to such a probability. The model used to represent wireless networks 
and define their formal behaviour is that of a pLTS.

In \cite{gallina2011} the authors propose 
a probabilistic, energy-aware process calculus of networks. 
In this calculus nodes can move probabilistically 
among locations of a given metric space. Nodes can also choose the transmission 
radius of a broadcast in an optimal way with respect to energy consumption. 
The authors propose a notion of probabilistic barbed congruence, parametrised in a 
set of schedulers, for which they give a characterisation in terms of bisimulations. 
The authors also introduce a preorder which compare networks which exhibit 
the same behaviour according to the proposed contextual equivalence, but differ 
in terms of energy consumption.

In \cite{bugliesi2012} a variant of the calculus above is proposed, where 
energy consumption is no longer considered and the possibility of interferences in 
communications is introduced. The authors define a contextual equivalence in 
terms of probabilistic reduction barbed congruence, for which they develop a 
sound and complete proof technique based on bisimulations.

In \cite{ghassemi} a different approach is made to formalise a wireless 
network. The authors identify a network as a set of processes associated
 with a location address and a queue, representing the data 
at the datalink layer that a station has not yet broadcast. The calculus 
they use is a probabilistic generalisation of the restricted broadcast process theory 
of \cite{GwFM10}; here the sending primitive consists of a message to be broadcast and a probability 
rate, representing the likelihood that such a message will be sent. The 
model used to describe the behaviour of a system is that of 
Continuous Time Markov Automata.

\appendix

\section{Properties of the Operator $\testP$}
\label{sec:operator.results}

\paragraph{\textbf{Proof of Proposition \ref{prop:testp.assoc}(1)}:} 
\label{proof:testP.closed}
Let $\calM = \Gamma_M \with M$, $\calN = \Gamma_N \with N$, 
and suppose $(\calM \testP \calN)$ is defined. 
Then such a network is equal to $(\Gamma_M \cup \Gamma_N) 
\with (M \Cpar N)$, for which we have to verify two statements:

\begin{itemize}
\item $(\Gamma_M \cup \Gamma_N) \with (M \Cpar N)$ satisfies the constraints we 
have placed over all, possibly non well-formed, networks. These require 
$M \Cpar N \in \sys$ (that is, it does not contain replicated node names),
$\nodes{M \Cpar N} \subseteq \Gamma_V$, and 
$(\Gamma_M \cup \Gamma_N)_E$ being irreflexive.

\begin{enumerate}

\item $M \Cpar N \in \sys$; note that if there were a node name $m$ which appears more 
than once in $M \Cpar N$, then we should have $m \in \nodes{M}, m \in \nodes{N}$. This is 
because $M \in \sys, N \in \sys$, so that $m$ cannot appear more than once in $M$, nor in $N$. 
Thus the statement follows if we can prove that $\nodes{M} \cap \nodes{N} = \emptyset$; 
since $\calM \testP \calN$ is defined, 
it follows that 
$\nodes{M} \cap (\Gamma_N)_V = \emptyset$. 
Since $\nodes{N} \subseteq (\Gamma_N)_V$,
we also have $\nodes{M} \cap \nodes{N} = \emptyset$, and there is nothing left to prove.

\item $\nodes{M \Cpar N} \subseteq (\Gamma_M \cup \Gamma_N)_V$; 
note that $\nodes{M} \subseteq (\Gamma_M)_V$ and $\nodes{N} \subseteq 
(\Gamma_N)_V$. Therefore we have that
 \[ \nodes{M \Cpar N} = \nodes{M} \cup \nodes{N} \subseteq (\Gamma_M)_V \cup (\Gamma_N)_V = 
 (\Gamma_M \cup \Gamma_N)_V.
 \]

\item $(\Gamma_M \cup \Gamma_N)_E$ is irreflexive. Suppose $(\Gamma_M \cup \Gamma_N) \vdash \rconn{m}{n}$; We need to 
show that $m \neq n$.
Note that we have either $\Gamma_M \vdash \rconn{m}{n}$ or $\Gamma_N \vdash \rconn{m}{n}$; 
without loss of generality, assume $\Gamma_M \vdash \rconn{m}{n}$. Since $(\Gamma_M)_E$ is irreflexive, 
it follows that $m \neq n$.
\end{enumerate}

\item $(\Gamma_M \cup \Gamma_N) \with (M \Cpar N)$ satisfies the constraints 
of Definition \ref{def:well.formed}. This amounts to prove the following:

\begin{enumerate}

\item for any $m,n$ such that $\Gamma_M \cup \Gamma_N \vdash \someconn{m}{n}$, 
either $m \in \nodes{M \Cpar N}$ or $n \in \nodes{M \Cpar N}$.

Let $m,n$ be two nodes for which $(\Gamma_M \cup \Gamma_N) \vdash \someconn{m}{n}$; 
that is either $(\Gamma_M \cup \Gamma_N) \vdash \rconn{m}{n}$ or $(\Gamma_M \cup \Gamma_N) \vdash \lconn{m}{n}$. 
Without loss of generality, assume that $(\Gamma_M \cup \Gamma_N) \vdash \rconn{m}{n}$.

In this case either $\Gamma_M \vdash \rconn{m}{n}$ or $\Gamma_N \vdash \rconn{m}{n}$. 
We only give details for the case in which $\Gamma_M \vdash \rconn{m}{n}$, as the proof 
for the second case is analogous. 
Since $\calM \in \nets$, we have 
that either $m \in \nodes{M}$ or $n \in \nodes{N}$. 
If $m \in \nodes{M}$ then $m \in \nodes{M \Cpar N}$, while 
if $n \in \nodes{M}$, then $n \in \nodes{M \Cpar N}$. 
Thus, either $m \in \nodes{M \Cpar N}$ or $n \in \nodes{M \Cpar N}$.

\item Let $m \in (\Gamma_N \cup \Gamma_N)_V$ be a node such that $(\Gamma_M \cup \Gamma_N) \vdash \someconn{m}{n}$ 
for no node $n \in (\Gamma_N)_V$. Then $m \in \nodes{M \Cpar N}$.

Since $m \in (\Gamma_M \cup \Gamma_N)_V$ then either $\Gamma_M \vdash m$ or $\Gamma_N \vdash m$. 
Without loss of generality let $\Gamma_M \vdash m$. Also, since $(\Gamma_M \cup \Gamma_N) \vdash \someconn{m}{n}$ 
for no $n \in (\Gamma_M \cup \Gamma_N)_V$, $(\Gamma_M)_V \subseteq (\Gamma_M \cup \Gamma_N)_V$ and $(\Gamma_M)_E \subseteq 
(\Gamma_M \cup \Gamma_N)_E$, then 
we also have that $\Gamma_M \vdash \someconn{m}{n}$ for no $n \in (\Gamma_M)_V$. 

Thus we have that $\Gamma_M \vdash m$ and $\Gamma_M \vdash \someconn{m}{n}$ for no $n \in (\Gamma_M)_V$. 
Since $(\Gamma_M \with M)$ is well-formed by hypothesis, we must have $m \in \nodes{M}$, from which it 
follows that $m \in \nodes{M \Cpar N}$.\qed
\end{enumerate}
\end{itemize}

\paragraph{\textbf{Proof of Proposition \ref{prop:testp.assoc}(2)}:}
\label{proof:testP.assoc}

It is sufficient to check that 
$\nodes{\calM} \cap (\calN)_V = \emptyset$ and 
$\nodes{\calM \testP \calN} \cap (\calL)_V = \emptyset$ 
if and only if $\nodes{\calN} \cap (\calL)_V= \emptyset$ and 
$\nodes{\calM} \cap (\calN \testP \calL)_V = \emptyset$.
In fact, from this claim it follows that $(\calM \testP \calN) \testP \calL$ is 
defined if and only if $\calM \testP (\calN \testP \calL)$ is defined; the 
equality of these two networks follows from the associativity of both set 
union and parallel composition of system terms.

Let $\calM = \Gamma_M \with M$, 
$\calN = \Gamma_N \with N$ and $\calL = \Gamma_L 
\with L$. We prove the two implications above separately.

Suppose that 
\begin{eqnarray}
(\nodes{M} \cap (\Gamma_N)_V) &=& \emptyset\label{hltest.eq:assoc.1}\\
(\nodes{M \Cpar N} \cap (\Gamma_L)_V) &=& \emptyset\label{hltest.eq:assoc.2}
\end{eqnarray}
We want to show that $\nodes{N} \cap (\Gamma_L)_V = \emptyset$, and 
$\nodes{M} \cap (\Gamma_N \cup \Gamma_L)_V = \emptyset$. The former 
statement is a straightforward consequence of Equation \eqref{hltest.eq:assoc.2}, 
since $\nodes{N} \subseteq \nodes{M \Cpar N}$.
The second statement can be proved as follows: let $m \in \nodes{M}$. By Equation 
\eqref{hltest.eq:assoc.1} we have that $\Gamma_N \not\vdash m$, so that 
it remains to show $\Gamma_L \not\vdash m$. This is a trivial consequence 
of Equation \eqref{hltest.eq:assoc.2}; in fact, since $m \in \nodes{M}$, we 
also have $m \in \nodes{M \Cpar N}$, and therefore $\Gamma_L \not\vdash m$.

Now suppose that 
\begin{eqnarray}
(\nodes{N} \cap (\Gamma_L)_V) &=& \emptyset\label{hltest.eq:assoc.3}\\
\nodes{M} \cap (\Gamma_N \cup \Gamma_L)_V &=& \emptyset\label{hltest.eq:assoc.4}
\end{eqnarray}
We need to show that $\nodes{M} \cap (\Gamma_N)_V = \emptyset$, and 
$\nodes{M \Cpar N} \cap (\Gamma_L)_V = \emptyset$. The first statement is 
an immediate consequence of Equation \eqref{hltest.eq:assoc.4}, by noticing 
that $(\Gamma_N)_V \subseteq (\Gamma_N \cup \Gamma_L)_V$. 
For the second 
statement, let $m$ be a node such that $\Gamma_L \vdash m$. By Equation 
\eqref{hltest.eq:assoc.3} we have that $m \notin \nodes{N}$. Also, 
by Equation \eqref{hltest.eq:assoc.4} it holds that $m \notin \nodes{M}$; 
in fact, since $\Gamma_L \vdash m$, we also have $\Gamma_N \cup \Gamma_L 
\vdash m$, and therefore $m \notin \nodes{M}$. 
Since $m \notin \nodes{M}$ and $m \notin \nodes{N}$, it follows that 
$m \notin \nodes{M \Cpar N}$, as we wanted to prove.\hfill\qed

\paragraph{\textbf{Proof of Proposition \ref{prop:network.decomp}:}}
\label{proof:generators}.
Let $\calM = \Gamma_M \with M$ be a well-formed network, and assume that 
$\nodes{\calM} \neq \emptyset$. That is, there exists a node name $m$, a state $s$ 
and a system term $N$ such that  $M \equiv 
\Cloc{s}{m} \Cpar N$. 

Let $\calG = \Gamma_G \with \Cloc{s}{m}$, where 
$\Gamma_G$ is defined by 
\begin{eqnarray}
(\Gamma_G)_V &=& \{m\} \cup \{n \in (\Gamma_M)_V \;|\; \Gamma_M \vdash m \leftrightarrows n\}\label{hltest.eq:generators.3}\\
(\Gamma_G)_E &=& \{(m,n) \;|\; \Gamma_M \vdash \rconn{m}{n}\} 
\cup \{(n,m)\;|\; \Gamma_M \vdash \rconn{n}{m}\} \label{hltest.eq:generators.4}.
\end{eqnarray}

Let also $\calN = \Gamma_N \with N$, where $\Gamma_N$ is defined by letting 
\begin{eqnarray}
(\Gamma_N)_V &=& \nodes{N} \cup \{n \;|\; n \neq m, \Gamma_M \vdash m' \leftrightarrows n \mbox{ for some } 
m' \in \nodes{N}\}\label{hltest.eq:generators.1}\\
(\Gamma_N)_E &=& (\Gamma_M)_E \setminus (\Gamma_G)_E\label{hltest.eq:generators.2}
\end{eqnarray}

We need to show the following facts:
\begin{enumerate}
\item \label{generators.1} $\calG \in \mathbb{G}$,
\item \label{generators.2} $\calN \in \nets$, 
\item \label{generators.3} $\nodes{\calG} \cap (\calN)_V = \emptyset$,
\item \label{generators.4} $(\calM)_V = (\calG)_V \cup (\calN)_V$,
\item \label{generators.5} $(\calM)_E = (\calG_E) \cup (\calN)_E$,
\item \label{generators.6} $ M \equiv \Cloc{s}{m} \Cpar N$.
\end{enumerate}

Each of the statements above is proved separately. Note that \eqref{generators.6} 
follows by the hypothesis.

\begin{description}
\item[\textbf{Proof of Statement \ref{generators.1}}] $\calG \in \mathbb{G}$.\\
First note that $\size{\nodes{\Cloc{P}{m}}} = 1$, so that it suffices 
to show that $\calG$ is  well-formed.
To this end, we show that $\calG$ satisfies both the constraints that 
we have placed over networks and those given in Definition \ref{def:well.formed}. 
\begin{enumerate}
\item $\Cloc{s}{m} \in \sys$; this is trivial, since no node name can appear more 
than once in a system term which contains only one node name
\item $\nodes{\Cloc{s}{m}} \subseteq (\Gamma_G)_V$; this statement follows from the 
definition of $(\Gamma_G)_V$,  Equation \eqref{hltest.eq:generators.3}, which gives that $\{m\} \subseteq 
(\Gamma_G)_V$.
\item $(\Gamma_G)_E$ is irreflexive; note that $(\Gamma_G)_E \subseteq 
(\Gamma_M)_E$, and the latter is irreflexive. Therefore, $(\Gamma_G)_E$ has 
to be irreflexive as well.

\item Whenever $\Gamma_G \vdash \someconn{l}{k}$ for 
some nodes $l,k$, then either 
$l \in \nodes{\Cloc{s}{m}}$ or $k \in \nodes{\Cloc{s}{m}}$. 
Equivalently, we prove that whenever $\Gamma_G \vdash \someconn{l}{k}$ 
for some nodes $l,k$, then either $l=m$ or $k=m$.

Suppose $\Gamma_G \vdash \someconn{l}{k}$; then either 
$\Gamma_G \vdash \rconn{l}{k}$ or $\Gamma_G \vdash \lconn{l}{k}$. 
Due to the arbitrariness of $l,k$, it is sufficient to consider only the first case. 

By Definition of $(\Gamma_G)_E$, Equation \eqref{hltest.eq:generators.4}, 
either $(l,k) = (m,n)$ for some $n$ such that $\Gamma_M \vdash \rconn{m}{n}$, 
or $(l,k) = (n,m)$ for some $n$ such that $\Gamma_M \vdash \rconn{n}{m}$. 
In the first case we obtain $l = m$, in the second $k = m$, and there is nothing left to prove.

\item If $\Gamma_G \vdash n$ and $\Gamma_G \vdash \someconn{n}{l}$ for 
no $l \in (\Gamma_G)_V$, then $n \in \nodes{\Cloc{s}{m}}$, or 
equivalently $n = m$. 

Note that if $(\Gamma_G \vdash n$ then by Equation \eqref{hltest.eq:generators.3} either $n = m$, 
in which case there is nothing to prove, or  
$\Gamma_M \vdash \someconn{m}{n}$. By Equation \ref{hltest.eq:generators.4} we also 
have that $\Gamma_G \vdash \someconn{m}{n}$, which contradicts the hypothesis.
\end{enumerate}

\item[\textbf{Proof of Statement \ref{generators.2}}] $\calN \in \nets$.\\
We need to show that $\calN$ satisfies the standard requirements we placed over 
all networks, plus the requirements required for a network to be well-formed, that is those 
listed in Definition \ref{def:well.formed}
\begin{enumerate}
\item $N \in \sys$. 
By hypothesis we already know 
that $M \in \sys$, and since $M \equiv \Cloc{s}{m} \Cpar N$, it follows that 
no node name appears in $N$ more than once.

\item $\nodes{N} \subseteq (\Gamma_N)_V$. This follows immediately from  
Equation \eqref{hltest.eq:generators.1}.

\item $(\Gamma_N)_E$ is irreflexive; this follows since, by Equation \eqref{hltest.eq:generators.2}, 
$(\Gamma_N)_E \subseteq (\Gamma_M)_E$; the latter is irreflexive by hypothesis.

\item Whenever $\Gamma_N \vdash \someconn{n}{l}$ for some nodes $n$ and $l$, 
then either $n \in \nodes{N}$ or $l \in \nodes{N}$. 

Due to the arbitrariness of node names $n,l$, it is sufficient to show that 
the property holds whenever $\Gamma_N \vdash \rconn{n}{l}$. 
Let $n,l$ be two nodes such that $\Gamma_N \vdash \rconn{n}{l}$. 
Note that by Equation \eqref{hltest.eq:generators.2} we have that $(\Gamma_N)_E 
\subseteq (\Gamma_M)_E$; since $\calM$ is well-formed, it follows 
that either $n \in \nodes{M}$ or $l \in \nodes{M}$.

However, since $\Gamma_N \vdash \rconn{n}{l}$, we also have that $\Gamma_N \vdash n$ and $\Gamma_N \vdash l$, 
Equation \eqref{hltest.eq:generators.1} also ensures that $n,l \neq m$.

Thus either $n \in \nodes{M} \setminus \{m\}$ or $l \in \nodes{M} \setminus \{m\}$; 
but since $M \equiv \Cloc{s}{n} \Cpar N$, and $M \in \sys$,  
$\nodes{M} \setminus \{m\}$ is exactly $\nodes{N}$.

\item if $\Gamma_N \vdash n$ and $\Gamma_N \vdash \someconn{n}{l}$ for no $l \in (\Gamma_N)_V$, 
then $n \in \nodes{N}$. By Equation \ref{hltest.eq:generators.1} either $n \in \nodes{N}$, in 
which case there is nothing to prove, or there exists $m' \in \nodes{N}$ such that $\Gamma_N \vdash \someconn{n}{m'}$. 
But this last case is not possible, since it contradicts the hypothesis.

\end{enumerate}

\item[\textbf{Proof of Statement \ref{generators.3}}] $\nodes{G} \cap (\Gamma_N)_V = \emptyset$.\\
Let $n \in (\Gamma_N)_V$; we need to show that $n \neq m$. By Equation \ref{hltest.eq:generators.1} 
there are two possible cases:
\begin{enumerate}
\item $n \in \nodes{N}$, in which case $m = n$ would contradict the hypothesis that 
$M \equiv \Cloc{s}{m} \Cpar N \in \sys$, or 
\item $n \in \{l \;|\; ; l \neq m, \Gamma_M \vdash \someconn{m'}{l} \mbox{ for some } m' \in \nodes{N}\}$; 
again $n \neq m$.
\end{enumerate}

\item[\textbf{Proof of Statement \ref{generators.4}}] $(\Gamma_M)_V = (\Gamma_G)_V \cup (\Gamma_N)_V$.\\ 
Note that equations \eqref{hltest.eq:generators.1} and \eqref{hltest.eq:generators.3} ensure that 
$(\Gamma_N)_V \subseteq (\Gamma_M)_V$ and $(\Gamma_G)_V \subseteq (\Gamma_M)_V$, respectively. 
Therefore, it is sufficient to show that $(\Gamma_M)_V \subseteq (\Gamma_G)_V \cup (\Gamma_N)_V$. 

Suppose that $\Gamma_M \vdash n$. There are two possible cases:
\begin{enumerate}
\item $n \in \nodes{M}$; Since $M \equiv \Cloc{s}{m} \Cpar N$, either 
$n = m$, in which case $m \in (\Gamma_G)_V$ by Equation \eqref{hltest.eq:generators.3}, 
or $n \in \nodes{N}$, in which case $n \in (\Gamma_N)_V$ by Equation \eqref{hltest.eq:generators.1},

\item $n \in \type{\calM}$; since $\calM$ is well-formed, there exists a node $l \in \nodes{\calM}$ 
such that $\Gamma_M \vdash \someconn{l}{n}$. Then either $l = m$, in which case 
Equation \eqref{hltest.eq:generators.3} ensures that $n \in (\Gamma_G)_V$, or 
$l \in \nodes{N}$, in which case $n \in (\Gamma_N)_V$ by Equation \eqref{hltest.eq:generators.1}.
\end{enumerate}

\item[\textbf{Proof of Statement \ref{generators.5}}] $(\Gamma_M)_E = (\Gamma_G)_E \cup (\Gamma_N)_E$.\\ 
This follows immediately from Equation \eqref{hltest.eq:generators.2} and the fact that $(\Gamma_G)_E \subseteq 
(\Gamma_N)_E$. \hfill\qed
\end{description}

\section{Decomposition and Composition Results}
\label{sec:decomposition.results}
To prove Propositions \ref{prop:decomp} and \ref{prop:wea.comp}, 
we first need to prove the following statements for actions which 
can be derived in the intensional semantics: 
\begin{prop}[Weakening]
\label{prop:weakening}

Let $\Gamma_1 \with M$ be a network, and let $\Gamma_2$ be such 
that$(\Gamma_2)_V \cap \nodes{M} = \emptyset$.
Then 
\[
\Gamma_1 \with M \ar{\mu} \Delta \mbox{ implies } 
(\Gamma_1 \cup \Gamma_2) \with M \ar{\mu} \Delta
\]
\noindent
where $\mu$ ranges over the actions $m.\tau, m.c!v, m.c?v$.
\end{prop}
\begin{proof}
By structural induction on the proof of the derivation 
$\Gamma_1 \with M \ar{\mu} \Delta$.
\end{proof}

\begin{prop}[Strengthening]
\label{prop:strengthening}
Let $\Gamma_1 \with M$ be a network, and let $\Gamma_2$ such 
that $(\Gamma_2)_V \cap \nodes{M} = \emptyset$.
Then 
\[
(\Gamma_1 \cup \Gamma_2) \with M \ar{\mu} \Delta \mbox{ implies }
\Gamma_1 \with M \ar{\mu} \Delta 
\]
\noindent
where $\mu$ ranges over the actions $m.\tau, m.c!v, m.c?v$.
\end{prop}

\begin{proof}
By structural induction on the proof of the transition 
$(\Gamma_1 \cup \Gamma_2) \ar{\mu} \Delta$.
\end{proof}

\paragraph{\textbf{Proof of Proposition \ref{prop:decomp}}}
\label{proof:decomp}
Let $\calM = \Gamma_M \with M$ be a network and $\calG = \Gamma_N \with \Cloc{s}{n}$ be 
a generating network such that $(\calM \testP \calG)$ is defined. 
We prove only the first statement of Proposition \ref{prop:decomp}.  
The details for the other statements are similar. 

Suppose then that $(\calM \testP \calG) \extar{\tau} \Lambda$. 
By definition of extensional actions we have two possible cases.
\begin{enumerate}
\item $(\calM \testP \calG) \ar{m.\tau} \Lambda$ for some $m \in \nodes{\calM 
\testP \calG} = \nodes{M} \cup \{n\}$. 
We perform a case analysis on whether 
$m \in \nodes{M}$ or $m = n$.
	\begin{itemize}
	\item If $m \in \nodes{M}$ then by Proposition \ref{prop:tau} we have that 
	$\calM \testP \calG \equiv (\Gamma_M \cup \Gamma_N) \with 
	\Cloc{\tau.p+t}{m} \Cpar M' \Cpar \Cloc{s}{n}$ for some $t,p,M'$ such 
	that $M \equiv \Cloc{\tau.p+t}{m} \Cpar M'$ and $\Lambda \equiv  
	(\Gamma_M \cup \Gamma_N) \with (\interprP{\Cloc{p}{m}} \Cpar \pdist{M'} 
	\Cpar \Cloc{\pdist{s}}{n})$. Note that if we let $\Delta = (\Gamma_M \with 
	\interprP{\Cloc{p}{m}} \Cpar \pdist{M'})$ we can rewrite $\Lambda$ as 
	$\Delta \testP \Cloc{\pdist{s}}{n}$; further, Proposition \ref{prop:tau} 
	gives that $\calM \ar{m.\tau} \Delta$, which by definition of extensional 
	actions gives $\calM \extar{\tau} \Delta$.
	\item If $m = n$ then by Proposition \ref{prop:tau} we have that 
	$s \equiv \tau.p+t$ for some $p,t$, while $\Lambda \equiv 
	(\Gamma_M \cup \Gamma_N) \with (\pdist{M} \Cpar \Cloc{\interprP{p}}{n})$. 
	If we let $\Theta = \interprP{p}$ we can rewrite $\Lambda = 
	(\Gamma_M \with \pdist{M}) \testP (\Gamma_N \with \Cloc{\Theta}{n})$. 
	Further, by the definition of the rules in the intensional 
	semantics we have that $s \ar{\tau} \Theta$.
	\end{itemize}
\item $(\calM \testP \calG) \ar{m.c!v} \Lambda$ for some $m\in \nodes{M} \cup \{n\}$ 
such that $\{l \;|\; (\Gamma_M \cup \Gamma_N \vdash \rconn{m}{l})\} \cap 
\outp{\calM \testP \calG} = \emptyset$. Again we have to consider two different cases. 
	\begin{itemize}
	\item $m \in \nodes{M}$. By Proposition \ref{prop:broadcast} we have that 
	$(\calM \testP \calG) \equiv 
	(\Gamma_M \cup \Gamma_N) \with 
	(\Cloc{c!\pc{e}.p + t}{m} \Cpar M' \Cpar \Cloc{s}{n})$ 
	for some $e,p,t,M'$ such that $\interpr{e} = v$ and $M \equiv 
	\Cloc{c!\pc{e}.p+t}{m} \Cpar M'$, while 
	$\Lambda = (\Gamma_M \cup \Gamma_G) \with 
	\interprP{\Cloc{p}{m}} \Cpar \Lambda'$ for some 
	$\Lambda'$ such that $(\Gamma_M \cup \Gamma_N) 
	\with (M' \Cpar \Cloc{s}{n}) \ar{m.c?v} \Lambda'$. 
	It follows from Proposition \ref{prop:input} that 
	$\Lambda' \equiv (\Gamma_M \cup \Gamma_N) \with 
	(\Delta' \Cpar \Cloc{\Theta}{n})$, where $\Delta'$ and 
	$\Theta$ are such that 
	$(\Gamma_M \cup \Gamma_N) \with M \ar{m.c?v} \Delta'$ 
	and $(\Gamma_M \cup \Gamma_N) \with \Cloc{s}{n} \ar{m.c?v} 
	\Cloc{\Theta}{n})$. Now we can apply Proposition \ref{prop:strengthening} 
	to the transitions 
	$(\Gamma_M \cup \Gamma_N) \with \Cloc{c!\pc{e}.p+t}{m} 
	\ar{m.\tau} \interprP{\Cloc{p}{m}}$ and $(\Gamma_M \cup 
	\Gamma_N) \with M' \ar{m.c?v} \Delta'$ to obtain 
	$\Gamma_M \with \Cloc{c!\pc{e}.p+t}{m} \ar{m.c!v} \interprP{\Cloc{p}{m}}$ 
	and $\Gamma_M \with M' \ar{m.c?v} \Delta'$, respectively. 
	These two transitions induce, via an application of rule $\Rlts{Sync}$, the 
	transition $\Gamma_M \with M 
	\ar{m.c!v} \Gamma_M \with \interprP{\Cloc{p}{m}} \Cpar \Delta'$; 
	let $\Delta = \Gamma_M \with \interprP{\Cloc{p}{m}} \Cpar \Delta'$. 
	The extensional transition induced by the broadcast derived for 
	$\Gamma_M \with M$ can be either an internal action or an extensional 
	broadcast, depending on the topology of $\Gamma_M$. 
	First note that, since $\outp{\calM \testP \calG} \cap 
	\{l\;|;\ (\Gamma_M \cup \Gamma_N) 
	\vdash \rconn{m}{l}\}= \emptyset$, we have that 
	$\outp{\calM} \cap \{l\;|\; \Gamma_M \vdash \rconn{m}{l}\} \subseteq 
	\nodes{\calG} = \{n\}$. Therefore we have two possible cases
		\begin{itemize}
		\item If $\Gamma_M \vdash \notrconn{m}{n}$ then we have the 
		transition $\Gamma_M \with M \extar{\tau} \Delta$. 
		Now note that, since $m \notin (\Gamma_N)_V$, we also have 
		that $(\Gamma_M \cup \Gamma_N) \vdash \notrconn{m}{n}$. 
		Then the transition $(\Gamma_M \cup \Gamma_N) \with 
		\Cloc{s}{n} \ar{m.c?v} \Cloc{\Theta}{n}$ could have been derived 
		only via an application of either Rule $\Rlts{deaf}$ or Rule $\Rlts{disc}$. 
		In both cases we have $\Theta = \pdist{s}$. 
		\item If $\Gamma_M \vdash \rconn{m}{n}$ then we have 
		that $\Gamma_M \with M \extar{\eout{c!v}{\{n\}}}\Delta$. 
		In this case the transition $(\Gamma_M \cup \Gamma_N) \with 
		\Cloc{s}{n} \ar{m.c?v} \Cloc{\Theta}{n}$ could have been derived 
		only via an application of either Rule $\Rlts{rec}$ or 
		Rule $\Rlts{rec}$. In the first case we have that 
		$s \ar{c?v} \Theta$, while in the second case we obtain that 
		$s \nar{c?v}$ and $\Theta = \pdist{s}$.	
		\end{itemize}
	Finally, it is now easy to show that $\Lambda = 
	(\Gamma_M \cup \Gamma_N) \with 
	(\interprP{\Cloc{p}{m}} \Cpar \Delta' \Cpar \Cloc{\Theta}{n}) = 
	(\Gamma_M \with \interprP{\Cloc{p}{m}} \Cpar \Delta') 
	\testP (\Gamma_N \with \Cloc{\Theta}{n}) = \Delta \testP 
	(\Gamma_N \with \Cloc{\Theta}{n})$.
	
	\item $m = n$. In this case we have that $\{l \;|\; 
	(\Gamma_M \cup \Gamma_N) \vdash \rconn{n}{l}\} \cap 
	\outp{\calM \testP \calG} = \outp{\calG}$, that is 
	$\outp{\calG} = \emptyset$. By Proposition \ref{prop:broadcast} 
	we have that $s \equiv c!\pc{e}.p+t$ for some $e,p,t$ such that 
	$\interpr{e} =v$. This ensures that $s \ar{c!v} \Theta$, where 
	$\Theta = \interprP{p}$. Further, $(\Gamma_M \cup \Gamma_N) \with M 
	\ar{n.c?v} \Delta_M$ for some $\Delta_M$ such that 
	$\Lambda \equiv (\Gamma_M \cup \Gamma_N) \with 
	(\Delta_M \Cpar \Cloc{\Theta}{n})$. 
	By applying Proposition \ref{prop:strengthening} to the last transition 
	we obtain $\Gamma_M \with M \ar{n.c?v} \Delta$. Whether this intensional 
	transition induces an extensional one depends on the topology 
	$\Gamma_M$. 
		\begin{itemize}
		\item if $n \in \inp{\calM}$ then we have the 
		extensional transition $\Gamma_M \with M \ar{n.c?v} \Delta$, 
		\item otherwise the transition above does not induce an extensional 
		input. However, in this case it is easy to show, using Proposition 
		\ref{prop:input} that $\Delta = \pdist{\calM}$. 
		\end{itemize}
	Finally, let $\Delta = \Gamma_M \with \Delta_M$. Note that 
	$\Lambda \equiv (\Gamma_M \cup \Gamma_N) 
	\with (\Delta_M \Cpar \Cloc{\Theta}{n}) = 
	\Delta \testP (\Gamma_N \with \Cloc{\Theta}{n})$.
	\end{itemize}
\end{enumerate}

\begin{lem}[Strong Composition of tau-actions]
\label{lem:tau.strong}
Let $\calM$ be a network, and $\calG = (\Gamma_N 
\with \Cloc{s}{n})$ be a generating network such that 
$\calM \testP \calG$ is well-defined. 
If $\calM \extar{\tau} \Delta$ then 
$\calM \testP (\Gamma_N \with \Cloc{s}{n}) 
\extar{\tau} \Delta \testP (\Gamma_N \with \Cloc{\pdist{s}}{n})$.
\end{lem} 
\begin{proof}
Let $\calM = \Gamma_M \with M$
If $\calM \extar{\tau} \Delta$; then $\Delta = \Gamma_M \with \Delta_M$ 
for some $\Delta_M$. By definition of extensional 
tau actions, there are two possibilities: 
\begin{enumerate}
\item $\calM \ar{m.\tau} \Delta_M$ for some $m \in \nodes{M}$. 
By Proposition \ref{prop:weakening} we obtain that 
$(\Gamma_M \cup \Gamma_N) \with M \ar{m.\tau} \Delta$, and 
finally $(\Gamma_M \cup \Gamma_N) \with M \Cpar \Cloc{s}{n} 
\ar{m.\tau} (\Delta_M \with \Cloc{\pdist{s}{n}})$ by Rule 
$\Rlts{tau-prop}$.
Note that $(\Gamma_M \cup \Gamma_N) \with ( \Delta_M \with 
\Cloc{\pdist{s}{n}}) = \Delta \testP (\Gamma_N \with \Cloc{\pdist{s}}{n})$; 
\item $\calM \ar{m.c!v} \Delta_M$, and $\Gamma_M \vdash \rconn{m}{l}$ for 
no $l \in \outp{\calM}$; in particular $\Gamma_M \vdash \notrconn{m}{n}$, which 
also gives $(\Gamma_M \cup \Gamma_N) \vdash \notrconn{m}{n}$. Therefore 
we can infer the transition $(\Gamma_M \cup \Gamma_N) \with 
\Cloc{s}{n} \ar{m.c?v} \Cloc{\pdist{s}}{n}$. 
By Proposition \ref{prop:weakening} we have that 
$\Gamma_M \with M \ar{m.c!v} \Delta_M$ implies 
$(\Gamma_M \cup \Gamma_N) \with M \ar{m.c!v} \Delta_M$. 
Now we can apply Rule $\Rlts{Sync}$ to obtain the transition 
$(\Gamma_M \cup \Gamma_N) \with (M \Cpar \Cloc{s}{n}) \ar{m.c!v} 
(\Gamma_M \cup \Gamma_N) \with (\Delta_M \Cpar \Cloc{\pdist{s}}{n})$. 
Note that the last network can be rewritten as $\Delta \testP (\Gamma_N \with 
\Cloc{\pdist{s}}{n})$. Finally, since $\Gamma_M \vdash \rconn{m}{l}$ for no 
$l \in \outp{\calM}$ and $\Gamma_N \not\vdash m$, it follows that 
$(\Gamma_M \cup \Gamma_N) \vdash \notrconn{n}{l}$ for any $l 
\in \outp{\calM \testP (\Gamma_N \with \Cloc{s}{n})}$. Hence we have 
the extensional transition 
$\calM \testP (\Gamma_N \with \Cloc{s}{n}) \extar{\tau} 
\Delta \testP (\Gamma_N \with \Cloc{\pdist{s}}{n}$. 

\end{enumerate}
\end{proof}

\paragraph{\textbf{Proof of Proposition \ref{prop:wea.comp}}}
\label{proof:composition}
We only prove statements (1)(i), (2)(i) and (2)(ii); details for the other statements 
are similar.

\begin{enumerate}
\item Suppose that $\Delta \extAr{\tau} \Delta'$ and let $\Gamma_N, n, s$ be such that 
$\Delta \testP (\Gamma_N \with \Cloc{\pdist{s}}{n})$ is well-defined. We have to show that 
$\Delta \testP (\Gamma_N \with \Cloc{\pdist{s}}{n}) \extAr{\tau} 
\Delta' \testP (\Gamma_N \with \Cloc{\pdist{s}}{n})$. 

We first prove a weaker result: 
if $\Delta \extar{\tau} \Delta'$ then $\Delta \testP (\Gamma_N \with \Cloc{\pdist{s}}{n}) 
\extar{\tau} \Delta' \testP (\Gamma_N \with \Cloc{\pdist{s}}{n})$. 
To see why this is true, rewrite $\Delta$ as 
$\sum_{i \in I} p_i \cdot \pdist{\calM_i}$, where $\sum_{i \in I} p_i \leq 1$. 
Then there exists a collection of distributions $\{\Delta'_i\}_{i \in I}$ such that 
$\calM_i \extar{\tau} \Delta'_i$ and $\Delta' = \sum_{i \in I} p_i \cdot \Delta'_i$. 
We can apply Lemma \ref{lem:tau.strong} to each of such transitions to 
obtain 
$\calM_i \testP (\Gamma_N \with \Cloc{s}{n}) \extar{\tau} \Delta'_i \testP 
(\Gamma_N \with \Cloc{\pdist{s}}{n})$. It follows that 
\begin{eqnarray*}
\Delta \testP (\Gamma_N \with \Cloc{\pdist{s}}{n}) &=& 
\sum_{i \in I} p_i \cdot \pdist{\calM_i \testP (\Gamma_N \with \Cloc{s}{n})}\\ 
&\extar{\tau}& \sum_{i \in I} p_i \cdot \Delta'_i \testP (\Gamma_N \with \Cloc{\pdist{s}}{n})\\
&=& \Delta' \testP (\Gamma_N \with \Cloc{\pdist{s}}{n})
\end{eqnarray*}

Now suppose that $\Delta \extAr{\tau} \Delta'$. Then there exist two collections 
if sub-distributions
$\{\Delta_k^{\rightarrow}\}_{k \geq 0}$ and $\Delta_{k}^{\times}$ such that 
$\Delta' = \sum_{k=0}^{\infty} \Delta_k^{\times}$ and
\begin{align*} 
\Delta &&=&& \Delta_{0}^{\rightarrow} &&+&& \Delta_{0}^{\times}\\
\Delta_0^{\rightarrow} && \extar{\tau} && \Delta_{1}^{\rightarrow} &&+&& 
\Delta_{1}^{\times}\\
\vdots&&&&\vdots&&&&\vdots\\
\Delta_k^{\rightarrow} && \extar{\tau} && \Delta_{k+1}^{\rightarrow} &&+&& 
\Delta_{k+1}^{\times}\\
\vdots&&&&\vdots&&&&\vdots\\
\end{align*}

For any $k\geq 0$, let $\Theta_k^{\rightarrow} = \Delta_k^{\rightarrow} \testP (\Gamma_N \with 
\Cloc{\pdist{s}}{n}$, and define $\Theta_k^{\times}$ analogously. 
Note that $\Theta_0^{\rightarrow} + \Theta_0^{\times} = 
\Delta \testP (\Gamma_N \with \Cloc{\pdist{s}}{n})$. 
Also, from the previous statement we can infer that 
$\Theta_k^{\rightarrow} \extar{\tau} (\Delta_{k+1}^{\rightarrow} + \Delta_{k+1}^{\times})
\testP (\Gamma_N \with \Cloc{\pdist{s}}{n})$. This last sub-distribution is exactly 
$\Theta_{k+1}^{\rightarrow} + \Theta_{k+1}^{\times}$. Therefore we have that 
$\Delta \testP (\Gamma_N \with \Cloc{\pdist{s}}{n}) \extAr{\tau} 
\sum_{k=0}^{\infty} \Theta_{k}^{\times}$. It remains to note that 
\begin{eqnarray*}
\sum_{k=0}^{\infty} \Theta_{k}^{\times} &=&
\Theta_k^{\times} = \Delta_k^{\times} \testP (\Gamma_N \with 
\Cloc{\pdist{s}}{n})\\
&=& \left(\sum_{k=0}^{\infty} \Delta_k^{\times}\right) \testP (\Gamma_N \with 
\Cloc{\pdist{s}}{n})\\
&=& \Delta' \testP (\Gamma_N \with 
\Cloc{\pdist{s}}{n})\\
\end{eqnarray*}

\item Suppose now that $\Delta \extAr{\eout{c!v}{\eta}} \Delta'$ for some 
$\eta$ with $n \notin \eta$; we have to show 
that 
$\Delta \testP (\Gamma_N \with \Cloc{\pdist{s}}{n}) 
\extAr{\eout{c!v}{\eta}} \Delta' \testP (\Gamma_N \with \Cloc{\pdist{s}}{n})$. 

First note that, whenever $\calM \extar{\eout{c!v}{\eta}} \Delta_M$, with 
$n \notin \eta$ then $\calM \testP (\Gamma_N \with \Cloc{s}{n}) 
\extar{\eout{c!v}{\eta}} \Delta \testP (\Gamma_N \with \Cloc{\pdist{s}}{n})$. 
We leave the proof of this result to the reader. 
An immediate consequence of the result above is that whenever 
$\Delta \extar{\eout{c!v}{\eta}} \Delta'$ and $s \ar{c?v} \Theta$, then 
$\Delta \testP (\Gamma_N \with \Cloc{\pdist{s}}{n}) \extar{\eout{c!v}{\eta}} \Delta' 
\testP (\Gamma_N \with \Cloc{\pdist{s}}{n})$. 

Finally, suppose that $\Delta \extAr{\eout{c!v}{\eta}} \Delta'$, where $n \notin \eta$. 
We proceed by induction on the definition of weak extensional outputs. 
\begin{itemize}
\item The base case is $\Delta \extAr{\tau} \Delta_1 \extar{\eout{c!v}{\eta}}
\Delta_2 \extAr{\tau} \Delta'$; 
in this case we have that $\Delta \testP (\Gamma_N \with \Cloc{s}{n}) 
\extAr{\tau} \Delta_1 \testP (\Gamma_N \with \Cloc{\pdist{s}}{n}) 
\extar{\eout{c!v}{\eta}} \Delta_2 \testP (\Gamma_N \with \Cloc{\pdist{s}}{n})
\extAr{\tau} \Delta' \testP (\Gamma_N \with \Cloc{\pdist{s}}{n})$, as we wanted to prove. 
\item Suppose now that $\Delta \extAr{\eout{c!v}{\eta_1}} \Delta_1 
\extAr{\eout{c!v}{\eta_2}} \Delta'$, where $\eta_1 \cap \eta_2 = \emptyset$ and 
$\eta_1 \cup \eta_2 = \eta$. Note that in this case $n \notin \eta_1$ and 
$n \notin \eta_2$. By inductive hypothesis we have that 
$\Delta \testP (\Gamma_N \with \Cloc{\pdist{s}}{n}) 
\extAr{\eout{c!v}{\eta_1}} \Delta_1 \testP (\Gamma_N \with \Cloc{\pdist{s}}{n}) 
\extAr{\eout{c!v}{\eta_2}} \Delta' \testP (\Gamma_N \with \Cloc{\pdist{s}}{n})$, 
and there is nothing left to prove.
\item Suppose that $\Delta \extAr{\eout{c!v}{\eta}} \Delta'$ for some $\eta$ 
such that $\{n\} \subset \eta$. Also, suppose that $s \ar{c?v} \Theta$. 
In this case we want to prove that 
$\Delta \testP (\Gamma_N \with \Cloc{\pdist{s}}{n}) 
\extAr{\eout{c!v}{\eta'}} \Delta' \testP (\Gamma_N \with \Cloc{\pdist{s}}{n})$, 
where $\eta' = \eta \setminus \{n\}$. 
The proof of these statements relies on the following technical result, whose 
proof is left to the reader: if 
$\calM \extar{\eout{c!v}{\eta}} \Delta'$ and $s \ar{c?v} \Theta$, then 
$\calM \testP (\Gamma_N \with \Cloc{\pdist{s}}{n}) 
\extar{\eout{c!v}{\eta'}} \Delta' \testP (\Gamma_N \with \Cloc{\Theta}{n})$, 
where $\eta' = \eta \setminus \{n\}$. 
Then the proof of the main result can be performed as in the previous case, 
by noting that if the transition $\Delta \extAr{\eout{c!v}{\eta}} \Delta'$ is induced by 
$\Delta \extAr{\eout{c!v}{\eta_1}} \Delta_1 \extAr{\eout{c!v}{\eta_2}} \Delta'$, 
where $\eta_1 \cup \eta_2 = \eta$ and $\eta_1 \cap \eta_2 = \emptyset$, then 
it cannot be $n \in \eta_1$ and $n \in \eta_2$. In this case it is necessary to rely 
on Proposition \ref{prop:wea.comp}(2)(i), which has already been proved.
\end{itemize}
\end{enumerate}

\bibliographystyle{plain}
\bibliography{broadcast}

\end{document}